\documentclass[12pt]{article}
\usepackage{amsmath,amssymb,amsthm}
\usepackage{graphics,epsfig}
\usepackage{hyperref}
\usepackage{natbib}
\usepackage{color}
\usepackage{graphicx}
\usepackage{caption}
\usepackage{subcaption}
\usepackage{float}
\usepackage{mathrsfs}
\usepackage{booktabs}
\usepackage{multirow}
\usepackage{comment}
\usepackage{setspace}
\usepackage{bbm}
\usepackage{bm}
\usepackage{amsfonts}
\usepackage{xcolor}
\usepackage{pslatex}
\usepackage{setspace}
\usepackage{bbm}
\usepackage{listings}
\usepackage[margin=1.25in]{geometry}
\usepackage{titling}
\newtheorem{definition}{\bf Definition}

\newtheorem{remark}{\bf Remark}
\newtheorem{rmk}[remark]{\bf Remark}

\newtheorem{theorem}{\bf Theorem}
\newtheorem{prop}[theorem]{\bf Proposition}
\newtheorem{lem}[theorem]{\bf Lemma}

\newtheorem{as}[theorem]{\bf Assumption}

\begin{document}
	
	\title{\bfseries
		Stochastic Choice with Distribution-Dependent Preferences}
	
	\author{
		Paramahansa Pramanik$^{1}$
	}
	
	\date{
		\small
		$^{1}$ Department of Mathematics and Statistics, University of South Alabama, Mobile, AL 36688, United States.\\
		 email: {\color{blue}\texttt{ppramanik@southalabama.edu}}
	}
	
	%\date{\today}
	\maketitle

\begin{abstract}
	We develop a continuous-time stochastic choice theory with endogenous preference evolution. Unlike dynamic random utility, observed behavior affects future preferences through the conditional distribution of latent preference states, generating endogenous distributional feedback. We show that this feedback has observable behavioral implications and characterize stochastic choice by a behavioral representation consisting of contemporaneous choice and continuation behavior. This representation is identified from stochastic choice, yields a rigidity result linking structural preference dynamics to observable behavior, and characterizes exactly when distribution dependent utility is behaviorally reducible to dynamic random utility. We further prove a behavioral impossibility theorem: stochastic choice arrays exhibiting behavioral distributional feedback admit no dynamic random utility representation. On the probabilistic side, we establish existence and weak uniqueness for the underlying conditional McKean-Vlasov system with conditional law feedback. The structure unifies endogenous information, latent preference dynamics, behavioral identification, and stochastic choice within a single continuous-time model.
\end{abstract}

\noindent\textbf{Keywords:} stochastic choice; distribution-dependent utility; dynamic random utility; conditional McKean-Vlasov dynamics.

\section{Introduction} \label{sec:introduction}
\subsection{Motivation}

Random utility models provide a fundamental framework for analyzing stochastic choice under unobserved preference heterogeneity. Their applications span consumer demand, industrial organization, labor economics, marketing, political economy, and dynamic discrete choice, where stochastic behavior reflects informational asymmetry between decision makers, who observe their realized preferences, and the analyst, who observes only realized choices \citep{mcfadden1974conditional,manski1977structure}. In continuous-time dynamic random utility, this asymmetry is represented by an evolving but exogenous latent preference, so history dependence arises because past choices reveal information about latent preferences rather than affect
their evolution. Consequently, learning and preference dynamics remain separated, a feature underlying the behavioral foundations and econometric applications of dynamic random utility \citep{frick2019dynamic}.

This paper develops the first continuous-time theory of endogenous preference evolution. We allow the analyst's conditional distribution of latent preferences, inferred from observed behavior, to enter both current utility and future preference dynamics. Our first main result shows that this feedback generates a new form of stochastic choice behavior that cannot arise under dynamic random utility. Past choices
not only reveal latent preferences but also alter their future
evolution, so history dependence reflects both information revelation and endogenous preference change. We characterize this behavioral feedback and show that it gives rise to observable restrictions on stochastic choice.

Our second main result characterizes the observable content of
endogenous preference evolution. We identify stochastic choice through a behavioral representation consisting of contemporaneous choice and continuation behavior. This representation is identified directly from stochastic-choice data and yields a rigidity theorem linking observable
behavior to structural preference dynamics. We further show that dynamic random utility is behaviorally equivalent to distribution dependent utility if and only if both behavioral components are invariant, providing a complete behavioral characterization of reducibility.

Our third main result establishes a sharp behavioral separation between distribution dependent utility and dynamic random utility. We prove that behavioral distributional feedback is impossible under any dynamic random utility representation. Consequently, stochastic-choice arrays generated by endogenous preference evolution belong to a strictly larger behavioral class than those generated by exogenous preference dynamics.
This impossibility theorem identifies endogenous preference evolution as a genuinely new behavioral phenomenon rather than an alternative parameterization of existing dynamic random utility models.

Finally, we provide structural foundations for the model through a conditional McKean-Vlasov system of endogenous preference dynamics. We establish existence and weak uniqueness for the underlying conditional-law, identify the conditional preference distribution as the unique structural state, and characterize its equilibrium fixed-point dynamics. This links endogenous information, preference evolution, and observable stochastic choice within a unified continuous-time.

Our results are complementary to the literatures on dynamic random utility \citep{frick2019dynamic}, dynamic discrete choice \citep{rust1987optimal,aguirregabiria2010dynamic}, and mean-field economics \citep{lasry2007mean,lacker2016general,CarmonaDelarue2018}. The first studies the behavioral and decision-theoretic foundations of stochastic choice under exogenous preference dynamics, the second focuses on the identification and estimation of dynamic decision problems with unobserved heterogeneity, and the third analyzes equilibrium interactions through distribution-dependent state dynamics. This paper connects these literatures by developing a behavioral theory of endogenous preference evolution, characterizing its observable implications through a behavioral representation, establishing its behavioral separation from dynamic random utility, and providing structural foundations through conditional McKean-Vlasov dynamics \citep{mckean1966class,huang2003individual,sznitman2006topics}. We hope that this framework provides a foundation for future theoretical and empirical work on stochastic choice in environments where information, beliefs, and preferences evolve jointly.

\subsection{Overview}

Section~\ref{sec:dynamicrandom} develops the continuous-time 
distribution-dependent utility (DDU) model. The economy consists of a latent preference process $X$, an observable information process $Y$, and the conditional preference distribution $\mu_t=\mathcal L(X_t\mid\mathcal F_t^Y)$.
Unlike dynamic random utility (DRU), $\mu_t$ enters both instantaneous utility and the dynamics of $X$, so observed behavior influences future preferences through the information it generates. The analyst observes only stochastic choices and the induced observation filtration, from which the conditional law is inferred. This leads to a conditional McKean-Vlasov system whose solution simultaneously determines latent
preferences, observable behavior, and the endogenous preference
distribution. We establish existence and weak uniqueness and show that the model generates well-defined contemporaneous and continuation stochastic choice operators, providing a continuous-time extension of DRU with endogenous preference.

Section~\ref{sec:behavior} develops the behavioral theory of DDU. The central object is the behavioral representation $\Phi=(\Phi_u,\Phi_P)$,
where $\Phi_u$ summarizes contemporaneous stochastic choice and
$\Phi_P$ summarizes continuation behavior generated by endogenous preference evolution. We show that stochastic choice data identify this behavior, characterize its observable image, and prove that it is invariant to observationally equivalent structural specifications. We further establish exact behavioral conditions under which DDU reduces to DRU and prove a behavioral impossibility theorem showing that stochastic-choice arrays exhibiting behavioral distributional feedback admit no dynamic random utility
representation.

Section~\ref{sec:cmkv} develops the structural theory of DDU through its conditional McKean-Vlasov system. We show that the conditional distribution of latent preferences is the unique endogenous state governing equilibrium preference dynamics and derive an equivalent structure in which all feedback operates through this conditional law. The resulting fixed-point yields a
rigidity theorem linking structural preference evolution to the
behavioral representation, thereby providing structural foundations for behavioral identification, equilibrium analysis, and comparative statics.

A central implication of the behavioral characterization is the
distinction between endogenous and exogenous preference evolution. DRU generates history dependence solely through
learning about exogenous latent preferences, whereas DDU additionally allows observed behavior to alter future preference dynamics through the conditional preference distribution. The behavioral impossibility theorem
shows that these two classes coincide only when behavioral
distributional feedback vanishes.

\subsection{Illustrative Example: Social Learning and Investment Decisions}
Consider a financial market in which investors repeatedly choose between a safe asset and a risky asset. Each investor privately observes a latent preference state reflecting their attitude toward risk, expected returns, or private information about market conditions, while the analyst observes only the sequence of portfolio allocations. Under standard DRU, observed investment histories reveal information about latent preferences, but the evolution of those preferences is exogenous. Although investment decisions may display persistence, past choices influence future decision only through Bayesian updating about an independently evolving preference process. In contrast, DDU permits observed investment decision to affect future preferences through the conditional distribution of latent preference states. As investors observe aggregate trading activity, the inferred distribution of market sentiment changes, and this distribution enters both current utility and subsequent preference dynamics. The resulting feedback generates endogenous persistence: today's investment decisions alter tomorrow's preference distribution, which in turn influences future investment decision. Section~\ref{sec:behavior} characterizes the observable implications of this feedback, while Section~\ref{sec:cmkv} shows that the resulting equilibrium dynamics admit a conditional MVSDE.
	
\section{Dynamic Random Utility vs. Distribution-Dependent Utility}\label{sec:dynamicrandom}

\subsection{Dynamic Random Utility}\label{sec:dru}

We begin with a continuous-time version of DRU that will serve as the benchmark for the DDU in subsection \ref{sec:ddu}. The construction preserves the informational structure emphasized by \citet{frick2019dynamic}; the decision maker privately observes the realization of their preferences, while the analyst observes choices and therefore treats behavior as stochastic. Throughout the paper, \(0\) denotes the initial time, \(t>0\) the terminal time, and \(s\in[0,t]\) an arbitrary time. Let \(Z\) be a finite set of instantaneous consumption alternatives and let \(K(Z)\) denote the collection of nonempty finite subsets of \(Z\). At time \(s\), the decision maker faces a choice set \(A_s\in K(Z)\) and selects \(z_s\in A_s\). The resulting consumption path is denoted by \(Z_{\cdot}=(Z_s)_{0\leq s\leq t}\). When useful, \(z_{\cdot}\) denotes a deterministic consumption path and \(Z_{\cdot}\) its stochastic counterpart. Randomization can be incorporated by letting \(\Delta(Z)\) denote the set of probability distributions over \(Z\). A lottery is denoted by \(p\in\Delta(Z)\), and a finite menu of lotteries by \(A_s\in K(\Delta(Z))\). A deterministic alternative \(z\in Z\) is identified with the degenerate lottery placing probability one on \(z\).

Let \((\Omega,\mathcal{F},\mathbb{P})\) be a complete probability space equipped with a filtration \(\mathbb{F}=(\mathcal{F}_s)_{0\leq s\leq t}\) satisfying the usual conditions. The filtration represents the payoff-relevant information available to the decision maker. At each \(s\in[0,t]\), let \(u_s:\Omega\rightarrow\mathbb{R}^{Z}\) be an \(\mathcal{F}_s\)-measurable random instantaneous utility index. Thus \(u_s(z)\) is an \(\mathcal{F}_s\)-measurable random variable for every \(z\in Z\). We write \(u=(u_s)_{0\leq s\leq t}\) for the random felicity process. For a lottery \(p\in\Delta(Z)\), expected utility is \(u_s(p)=\sum_{z\in Z}p(z)u_s(z)\). Let \(U_s\) be the continuation utility at time \(s\). Whereas \(u_s\) evaluates instantaneous consumption, \(U_s\) evaluates consumption opportunities over the remaining interval \([s,t]\). This distinction between instantaneous felicity and continuation utility will be maintained throughout the paper. Assume \(r=(r_s)_{0\leq s\leq t}\) be a progressively measurable discount-rate process with \(r_s\geq0\), and define \(D_{s,v}:=\exp\{-\int_s^v r_\tau\,d\tau\}\) for \(s\leq v\leq t\). For an admissible consumption process \(Z_{\cdot}\), continuation utility is
\begin{equation}
	U_s
	=
	\mathbb{E}\!\left[
	\left.
	\int_s^t D_{s,v}u_v(Z_v)\,dv
	+
	D_{s,t}G(Z_t)
	\right|
	\mathcal{F}_s
	\right],
	\qquad 0\leq s\leq t,
	\label{eq:continuous-dru-value}
\end{equation}
where \(G:Z\rightarrow\mathbb{R}\) is terminal utility. Let
\(\mathcal C_s\)
denote the collection of admissible consumption processes on
\([s,t]\).
The associated value process is
\[
V_s
:=
\operatorname*{ess\,sup}_{Z_{\cdot}\in\mathcal C_s}
\mathbb E
\!\left[
\left.
\int_s^t
D_{s,v}u_v(Z_v)\,dv
+
D_{s,t}G(Z_t)
\,\right|\,
\mathcal F_s
\right].
\]
Hence,
\(U_s=V_s\)
under optimal consumption.
The felicity process
\(u=(u_s)_{0\le s\le t}\)
is otherwise unrestricted and may exhibit arbitrary serial dependence, so the conditional law of
\((u_v)_{v\ge s}\)
given
\(\mathcal F_s\)
need not be Markov or independent across time. Thus,
\(\mathcal F_s\)
may contain information about future felicities, while the law of
\(u\)
remains exogenous.

An important specialization arises when variation in instantaneous utility reflects learning about a fixed but initially unknown felicity index. Let \(\widetilde{u}:\Omega\rightarrow\mathbb{R}^{Z}\) be an \(\mathcal{F}\)-measurable random felicity function. Bayesian evolving felicity requires
$u_s= \mathbb{E}\!\left[
	\left.
	\widetilde{u}
	\right|
	\mathcal{F}_s
	\right],$ for all $s\in[0,t]$. Hence, for every \(z\in Z\), the process \((u_s(z))_{0\leq s\leq t}\) is an \(\mathbb{F}\)-martingale. Equivalently, for \(0\leq s\leq v\leq t\),
\(\mathbb{E}[u_v(z)\mid\mathcal{F}_s]=u_s(z)\).
This restriction distinguishes changes in beliefs about a fixed latent felicity from genuine changes in the underlying preference state. The decision maker observes \(\mathcal{F}_s\), whereas the analyst observes the history of choices. Let \(Y=(Y_s)_{0\leq s\leq t}\) denote the observable choice process and define the analyst's observation filtration by
$\mathcal{F}_s^{Y}:=\sigma(Y_v:0\leq v\leq s)\vee\mathcal{N},$
where \(\mathcal{N}\) is the collection of \(\mathbb{P}\)-null sets. We assume \(\mathcal{F}_s^{Y}\subseteq\mathcal{F}_s\), so the analyst generally has less information than the decision maker. For a finite menu \(A\subseteq Z\), let \(M_s(A):=\arg\max_{z\in A}u_s(z)\). Ignoring ties for the moment, the analyst's conditional probability of observing \(z\in A\) at time \(s\) is
\begin{equation}
	\rho_s(z;A)
	:=
	\mathbb{P}\!\left(
	z\in M_s(A)
	\,\middle|\,
	\mathcal{F}_s^{Y}
	\right).
	\label{eq:continuous-stochastic-choice-rule}
\end{equation}
Thus \(\rho_s(\cdot;A)\) is an \(\mathcal{F}_s^{Y}\)-measurable random probability vector. Its dependence on observed history reflects the informational content of past choices, when felicities are persistent, observations prior to \(s\) alter the analyst's conditional assessment of current preferences.
The distinction between \(\mathcal{F}_s^{Y}\) and \(\mathcal{F}_s\) will be central below. In the benchmark model, the conditional distribution induced by \(\mathcal{F}_s^{Y}\) summarizes what the analyst has learned about an otherwise exogenous preference process. In the distribution-dependent model, that conditional distribution becomes an endogenous state variable and enters the evolution of preferences themselves.

\subsection{Distribution-Dependent Utility}\label{sec:ddu}

The DRU benchmark in Section~\ref{sec:dru} treats the felicity process
\(u=(u_s)_{0\le s\le t}\)
as exogenous, allowing arbitrary serial dependence while excluding dependence on the analyst's information
\(\mathbb F^Y=(\mathcal F_s^Y)_{0\le s\le t}\).
We replace this specification by a latent preference process
\(X=(X_s)_{0\le s\le t}\)
and its conditional law
\(\mu_s=\mathcal L(X_s\mid\mathcal F_s^Y)\).
The process
\((\mu_s)_{0\le s\le t}\)
enters both the felicity functional
\(u(s,\cdot,X_s,\mu_s)\)
and the state dynamics of
\(X\),
thereby inducing the endogenous feedback
\(X\rightarrow Y\rightarrow\mu\rightarrow X\).
Accordingly,
\(\mathcal F^Y\)
determines both posterior beliefs about
\(X\)
and the subsequent evolution of
\((X_s,\mu_s)_{0\le s\le t}\). Let $X=(X_s)_{0\leq s\leq t}$ be an $\mathbb{R}^d$-valued latent preference process defined on the complete probability space $(\Omega,\mathcal{F},\mathbb{P})$. The process $X$ is privately observed by the decision maker but not by the analyst.Assume
$\mathbb{E}\left[\sup_{0\leq s\leq t}|X_s|^2\right]<\infty$. Let
$\mathcal{P}_2(\mathbb{R}^d)$ denote the set of Borel probability measures
$\mu$ on $\mathbb{R}^d$ satisfying
$\int_{\mathbb{R}^d}|x|^2\mu(dx)<\infty$, endowed with the
$2$-Wasserstein metric
\[
\mathcal{W}_2(\mu,\nu)
:=
\left[
\inf_{\pi\in\Pi(\mu,\nu)}
\int_{\mathbb{R}^d\times\mathbb{R}^d}|x-x'|^2
\,\pi(dx,dx')
\right]^{1/2},
\]
where $\Pi(\mu,\nu)$ is the set of couplings of $\mu$ and $\nu$.

\begin{definition}\label{def:conditional-preference-distribution}
For every $s\in[0,t]$, the \emph{conditional preference distribution} (CPD) is the random probability measure
$\mu_s:=\mathcal{L}(X_s\mid\mathcal{F}_s^Y)\in\mathcal{P}_2(\mathbb{R}^d).$ Equivalently, for every bounded Borel function
	$\varphi:\mathbb{R}^d\rightarrow\mathbb{R}$ we have
	$
	\int_{\mathbb{R}^d}\varphi(x)\mu_s(dx)
	=
	\mathbb{E}\!\left[
	\varphi(X_s)\mid\mathcal{F}_s^Y
	\right],$ such that $ \mathbb{P}\text{-a.s.}$
\end{definition}

\noindent
The CPD
\(\mu=(\mu_s)_{0\le s\le t}\)
is an
\(\mathbb F^Y\)-adapted
\(\mathcal P_2(\mathbb R^d)\)-valued process satisfying
\(\mu_s=\mathcal L(X_s\mid\mathcal F_s^Y)\),
whereas
\(\mathcal L(X_s)\)
denotes the unconditional law.
Hence,
\(\mu_s\)
depends on the observation filtration
\(\mathbb F^Y\)
and, in general,
\(\mu_s\neq\mathcal L(X_s)\)
\(\mathbb P\)-a.s.
The pair
\((X_s,\mu_s)\)
constitutes the endogenous preference state, with
\(X_s\)
representing the latent realization and
\(\mu_s\)
its conditional distribution.
Accordingly,
\(\mu_s\)
may enter both the felicity functional
\(u(s,z,X_s,\mu_s)\)
and the coefficients of the latent dynamics
\((b,\sigma,\sigma_0)\),
so that the mappings
\((X_s,\mu_s)\mapsto u_s\)
and
\((X_s,\mu_s)\mapsto(b,\sigma,\sigma_0)\)
capture, respectively, static distribution dependence and dynamic distributional feedback.

\begin{definition}\label{def:ddu-felicity}
A \emph{distribution-dependent felicity index} (DDFI) is a measurable function
$u:[0,t]\times Z\times\mathbb{R}^d\times\mathcal{P}_2(\mathbb{R}^d)
	\rightarrow\mathbb{R}.$ Given $(X_s,\mu_s)$, instantaneous utility from $z\in Z$ at time $s$ is $u(s,z,X_s,\mu_s)$. For a lottery $p\in\Delta(Z)$, its instantaneous utility is
	$u(s,p,X_s,\mu_s):=
	\sum_{z\in Z}p(z)u(s,z,X_s,\mu_s).
	$
\end{definition}
\noindent
The DDFI
\(u:[0,t]\times Z\times\mathbb R^d\times\mathcal P_2(\mathbb R^d)\rightarrow\mathbb R\)
depends jointly on the latent state
\(X_s\)
and the CPD
\(\mu_s\),
so that
\((X_s,\mu_s)\mapsto u(s,\cdot,X_s,\mu_s)\)
defines the instantaneous preference ordering at time
\(s\).
Hence,
\(u(s,\cdot,x,\mu)\neq u(s,\cdot,x,\tilde\mu)\)
may hold for
\(\mu\neq\tilde\mu\),
even when
\(x\in\mathbb R^d\)
is fixed.
Let
\(W=(W_s)_{0\le s\le t}\)
and
\(B=(B_s)_{0\le s\le t}\)
be independent Brownian motions of dimensions
\(d_W\)
and
\(d_Y\),
respectively.
The filtration
\(\mathbb F=(\mathcal F_s)_{0\le s\le t}\)
is generated by
\((W,B)\),
with
\(W\)
driving the latent preference state
\(X\)
and
\(B\)
the observable information process
\(Y\).
Consequently,
\((X,\mu,Y)\)
forms the endogenous state system, where
\(X\)
is privately observed,
\(Y\)
is publicly observed,
and
\(\mu_s=\mathcal L(X_s\mid\mathcal F_s^Y)\)
couples the latent and observable components. Consider the coupled system
\begin{equation}
	\left\{
	\begin{aligned}
		dX_s
		&=
		b(s,X_s,Y_s,\mu_s)\,ds
		+
		\sigma(s,X_s,Y_s,\mu_s)\,dW_s
		+
		\sigma_0(s,X_s,Y_s,\mu_s)\,dB_s,\\
		dY_s
		&=
		h(s,X_s,Y_s,\mu_s)\,ds
		+
		\Sigma_Y(s,Y_s,\mu_s)\,dB_s,\\
		\mu_s
		&=
		\mathcal{L}(X_s\mid\mathcal{F}_s^Y),
	\end{aligned}
	\right.
	\qquad s\in[0,t],
	\label{eq:ddu-system}
\end{equation}
with initial condition
\((X_0,Y_0)\sim\lambda_0\).
The coefficient tuple
\((b,\sigma,\sigma_0,h,\Sigma_Y)\)
is evaluated at
\((s,X_s,Y_s,\mu_s)\),
so that
\eqref{eq:ddu-system}
is a conditional McKean-Vlasov system driven by the state
\((X_s,Y_s,\mu_s)\).
The common noise
\(B\)
enters both the
\(X\)- and
\(Y\)-equations through
\(\sigma_0\)
and
\(\Sigma_Y\),
respectively, whereas
\(W\)
acts only on
\(X\).
The specification
\(\sigma_0\equiv0\)
reduces
\eqref{eq:ddu-system}
to the conditionally independent-noise case.
Unlike classical MVSDEs, the dependence
\((b,\sigma,\sigma_0,h,\Sigma_Y)
=
(b,\sigma,\sigma_0,h,\Sigma_Y)(s,X_s,Y_s,\mu_s)\)
is through the conditional law
\(\mu_s=\mathcal L(X_s\mid\mathcal F_s^Y)\),
so that the endogenous feedback
\(X\rightarrow Y\rightarrow\mu\rightarrow X\)
is encoded directly in the coefficients of
\eqref{eq:ddu-system}.

\begin{definition} \label{def:ddu}
	A stochastic choice environment admits a DDU if there exist a latent preference process
	$X$, an observable process $Y$, a conditional preference distribution
	$\mu$, and primitives $(u,b,\sigma,\sigma_0,h,\Sigma_Y)$ such that
	definition \ref{def:conditional-preference-distribution} and Condition \eqref{eq:ddu-system} hold and,
	conditional on the decision maker's information at time $s$, choice from
	every finite menu $A\subseteq Z$ maximizes
	$u(s,\cdot,X_s,\mu_s)$.
\end{definition}
\noindent For \(A\subseteq Z\),
define the DDU choice correspondence
\(M_s^{\mathrm{DDU}}(A):=\arg\max_{z\in A}u(s,z,X_s,\mu_s)\),
and, ignoring ties, the conditional stochastic choice rule
\(\rho_s(z;A)=\mathbb P\!\left(z\in M_s^{\mathrm{DDU}}(A)\mid\mathcal F_s^Y\right)\).
Relative to
\eqref{eq:continuous-stochastic-choice-rule},
\(\rho_s(\cdot;A)\)
depends jointly on
\(X_s\)
and
\(\mu_s=\mathcal L(X_s\mid\mathcal F_s^Y)\),
with
\(X_s\)
determining the realized latent preference state and
\(\mu_s\)
the conditional preference distribution induced by
\(\mathbb F^Y\).
Accordingly,
\((A,\mathcal F_s^Y)\mapsto(X_s,\mu_s)\mapsto M_s^{\mathrm{DDU}}(A)\mapsto\rho_s(\cdot;A)\),
so that
\(\mathcal F_s^Y\mapsto\mu_s\)
enters both the felicity index
\(u(s,\cdot,X_s,\mu_s)\)
and the conditional dynamics
\eqref{eq:ddu-system},
rather than affecting choice solely through
\(X_s\)
as in
\eqref{eq:continuous-stochastic-choice-rule}.

\begin{lem}\label{lem:conditional-sufficiency}
Let $\varphi:\mathbb{R}^d\rightarrow\mathbb{R}$ be Borel measurable with
$\mathbb{E}|\varphi(X_s)|<\infty$. Then
$\mathbb{E}[\varphi(X_s)\mid\mathcal{F}_s^Y]
	=
	\int_{\mathbb{R}^d}\varphi(x)\mu_s(dx),
	\ \mathbb{P}\text{-a.s.}$
In particular, conditional expectations of all integrable functions of the contemporaneous latent preference state depend on the observed history only through $\mu_s$.
\end{lem}

\noindent
Proof is in Appendix~\ref{app:conditional-sufficiency}. Lemma~\ref{lem:conditional-sufficiency} identifies
\(\mu_s=\mathcal L(X_s\mid\mathcal F_s^Y)\)
as the
\(\mathcal P_2(\mathbb R^d)\)-valued sufficient statistic for
\(X_s\),
since
\(\mathbb E[\varphi(X_s)\mid\mathcal F_s^Y]\)
depends on
\(X_s\)
only through
\(\mu_s\)
for every integrable
\(\varphi\).
No corresponding sufficiency is asserted for
\((X_r)_{r\ge s}\),
whose conditional law is determined by
\eqref{eq:ddu-system}
through
\((b,\sigma,\sigma_0,h,\Sigma_Y,\mu)\).
The distinction between contemporaneous sufficiency
\(X_s\mapsto\mu_s\)
and dynamic feedback
\(\mu\mapsto(X,Y)\)
is formalized in Definition~\ref{def:distributional-feedback}.

\begin{definition}
	\label{def:distributional-feedback}
	DDU exhibits
	\emph{felicity feedback}
	if
	\(u(\cdot,\mu)\neq u(\cdot,\tilde\mu)\)
	for some
	\(\mu,\tilde\mu\in\mathcal P_2(\mathbb R^d)\)
	with
	\(\mu\neq\tilde\mu\);
	\emph{preference feedback}
	if
	\((b,\sigma,\sigma_0)(\cdot,\mu)\neq(b,\sigma,\sigma_0)(\cdot,\tilde\mu)\)
	for some
	\(\mu\neq\tilde\mu\);
	and
	\emph{distributional feedback}
	if either
	\(u(\cdot,\mu)\neq u(\cdot,\tilde\mu)\)
	or
	\((b,\sigma,\sigma_0)(\cdot,\mu)\neq(b,\sigma,\sigma_0)(\cdot,\tilde\mu)\)
	for some
	\(\mu\neq\tilde\mu\).
\end{definition}
\noindent Definition~\ref{def:distributional-feedback}
decomposes the measure dependence of
\((u,b,\sigma,\sigma_0)\)
into the mappings
\(\mu\mapsto u\)
and
\(\mu\mapsto(b,\sigma,\sigma_0)\).
The former determines the contemporaneous preference ordering
\(u(s,\cdot,X_s,\mu_s)\),
whereas the latter determines the conditional evolution of
\((X_s,\mu_s)\)
through
\eqref{eq:ddu-system}.
Hence,
\(\partial_\mu u\neq0\)
with
\((\partial_\mu b,\partial_\mu\sigma,\partial_\mu\sigma_0)=0\)
corresponds to static distribution dependence,
while
\((\partial_\mu b,\partial_\mu\sigma,\partial_\mu\sigma_0)\neq(0,0,0)\)
induces dynamic feedback
\(X\rightarrow Y\rightarrow\mu\rightarrow X\).
The subsequent assumptions impose regularity on
\((u,b,\sigma,\sigma_0,h,\Sigma_Y)\)
under these two channels.

\noindent\textbf{Assumption 1.}
The coefficients $b$, $\sigma$, $\sigma_0$, $h$, and $\Sigma_Y$ are
Borel measurable in all arguments. There exists $L>0$ such that, for
all $s\in[0,t]$, $x,x'\in\mathbb{R}^d$, $y,y'$ in the state space of
$Y$, and $\mu,\nu\in\mathcal{P}_2(\mathbb{R}^d)$,
\[
\begin{aligned}
	&|b(s,x,y,\mu)-b(s,x',y',\nu)|
	+\|\sigma(s,x,y,\mu)-\sigma(s,x',y',\nu)\|\\
	&\quad
	+\|\sigma_0(s,x,y,\mu)-\sigma_0(s,x',y',\nu)\|
	+|h(s,x,y,\mu)-h(s,x',y',\nu)|\\
	&\quad
	+\|\Sigma_Y(s,y,\mu)-\Sigma_Y(s,y',\nu)\|
	\leq
	L\bigl(|x-x'|+|y-y'|+\mathcal{W}_2(\mu,\nu)\bigr).
\end{aligned}
\]
The coefficients have at most linear growth, $\Sigma_Y\Sigma_Y^\top$
is uniformly positive definite, and
$\mathbb{E}|X_0|^2+\mathbb{E}|Y_0|^2<\infty$.

\medskip

\noindent
Assumption~1 ensures that
\((b,\sigma,\sigma_0,h,\Sigma_Y)\)
satisfy the standard Lipschitz and linear-growth conditions in
\((x,y,\mu)\),
with
\(\mu\in\mathcal P_2(\mathbb R^d)\)
equipped with
\(\mathcal W_2\).
Unlike a classical SDE, however,
\eqref{eq:ddu-system}
contains the endogenous constraint
\(\mu_s=\mathcal L(X_s\mid\mathcal F_s^Y)\),
where
\(\mathcal F_s^Y=\sigma(Y_r:0\le r\le s)\),
so that the mapping
\((X,Y)\mapsto\mathcal F^Y\mapsto\mu\mapsto(X,Y)\)
defines a conditional McKean-Vlasov fixed point.
Under Assumption~1,
\((X,Y,\mu)\)
is well posed in the conditional-law sense developed by
\citet{buckdahn2023general},
thereby providing a unique admissible solution of
\eqref{eq:ddu-system}.

\begin{prop} \label{prop:ddu-wellposed}
	Suppose Assumption~1 holds and the observation volatility in
	\eqref{eq:ddu-system} satisfies
	$
	\Sigma_Y(s,y,\mu)\equiv\widehat{\Sigma}_Y,
	$ where $\widehat{\Sigma}_Y\in\mathbb{R}^{d_Y\times d_Y}$ is deterministic and invertible. Suppose, in addition, that the conditional-law regularity conditions in Supplementary Appendix~SA.1 hold. Then, for every admissible initial distribution $\lambda_0$, the system \eqref{eq:ddu-system} admits a weak solution $(X,Y,\mu)$ satisfying
	$
	\mu_s=\mathcal{L}(X_s\mid\mathcal{F}_s^Y),
	$ for all $ 0\leq s\leq t.$ Under the strengthened uniqueness conditions in Supplementary Appendix~SA.1, the law of $(X,Y,\mu)$, and hence, the law of the conditional preference distribution $\mu$, is unique in the stated solution class.
\end{prop}

\noindent
Proof is in Supplementary Appendix~SA.2.
Proposition~\ref{prop:ddu-wellposed}
imposes
$
\Sigma_Y(s,y,\mu)\equiv\widehat\Sigma_Y,
$
with
$
\widehat\Sigma_Y
$
deterministic and invertible, so that
$
\bar Y_s:=\widehat\Sigma_Y^{-1}Y_s
$
admits a reference-measure representation with Brownian innovation.
Theorem~\ref{thm:ddu-general-volatility}
instead allows
$
\Sigma_Y=\Sigma_Y(s,y,\mu)
$
in
\eqref{eq:ddu-system},
including state and law dependence, and establishes well-posedness without the reduction
$
Y\mapsto\bar Y
$
or the associated change of measure.

\begin{theorem} \label{thm:ddu-general-volatility}
Consider the system \ref{eq:ddu-system}. Suppose Assumption~1 and Conditions \textnormal{(GV1)}-\textnormal{(GV6)} in Supplementary Appendix~SA.1 hold. Then, for every
admissible initial distribution $\lambda_0\in
	\mathcal{P}_{2+\varepsilon}
	(\mathbb{R}^d\times\mathbb{R}^{d_Y})$
	for some $\varepsilon>0$, system
	\eqref{eq:ddu-system} admits a weak solution
	$
	(\Omega,\mathcal{F},\mathbb{F},\mathbb{P};
	X,Y,W,B,\mu)
	$
	such that
	\[
	\mathbb{E}^{\mathbb{P}}
	\left[
	\sup_{0\leq s\leq t}|X_s|^2
	+
	\sup_{0\leq s\leq t}|Y_s|^2
	+
	\sup_{0\leq s\leq t}m_2(\mu_s)
	\right]
	<\infty.
	\]
	Moreover, $\mu$ admits an
	$\mathbb{F}^Y$-progressively measurable version with continuous paths in
	$\mathcal{P}_2(\mathbb{R}^d)$, and
	$
	\langle\mu_s,\varphi\rangle
	=
	\mathbb{E}^{\mathbb{P}}
	\left[
	\varphi(X_s)\mid\mathcal{F}_s^Y
	\right]
	$
	for every bounded Borel function
	$\varphi:\mathbb{R}^d\rightarrow\mathbb{R}$ and every $s\in[0,t]$. Moreover, if Conditions \textnormal{(GU1)}-\textnormal{(GU3)} in Supplementary Appendix~SA.1 hold, then weak uniqueness holds. For any two weak solutions with the same initial distribution induce the same law on
	$
	C([0,t];\mathbb{R}^d)
	\times
	C([0,t];\mathbb{R}^{d_Y})
	\times
	C([0,t];\mathcal{P}_2(\mathbb{R}^d)).
	$
\end{theorem}

\noindent Proposition~\ref{prop:ddu-wellposed} treats the constant-volatility benchmark. The general case is stated in
Theorem~\ref{thm:ddu-general-volatility} and proved in the
Supplementary Appendix (SA1-SA9).

\begin{prop} \label{prop:ddu-nesting}
	Suppose,
	$
	u(s,z,x,\mu)=\bar u(s,z,x),
	$
	and
	$
	(b,\sigma,\sigma_0,h,\Sigma_Y)
	=
	(\bar b,\bar\sigma,\bar\sigma_0,\bar h,\bar\Sigma_Y),
	$
	where the latter coefficients are independent of $\mu$. Then $\mu$ is not a state variable of either the utility functional or the state dynamics. Hence, DDU reduces to the continuous-time DRU specification.
\end{prop}

\noindent
Proof is in Supplementary Appendix~SA.2.
Proposition~\ref{prop:ddu-nesting}
imposes
\((u,b,\sigma,\sigma_0,h,\\ \Sigma_Y)
=
(\bar u,\bar b,\bar\sigma,\bar\sigma_0,\bar h,\bar\Sigma_Y)\),
where
\((\bar u,\bar b,\bar\sigma,\bar\sigma_0,\bar h,\bar\Sigma_Y)\)
is independent of
\(\mu\).
Hence,
\(\mu_s=\mathcal L(X_s\mid\mathcal F_s^Y)\)
is generated by
\((X,Y)\)
but does not enter either
\(\bar u(s,z,X_s)\)
or
\((\bar b,\bar\sigma,\bar\sigma_0,\bar h,\bar\Sigma_Y)(s,X_s,Y_s)\).
DDU permits the dependence
\((u,b,\sigma,\sigma_0,h,\Sigma_Y)
=
(u,b,\sigma,\sigma_0,h,\Sigma_Y)(s,X_s,Y_s,\mu_s)\),
so that
\(\mu_s\)
is both a conditional law and an endogenous state variable through the mappings
\(\mu\mapsto u\)
and
\(\mu\mapsto(b,\sigma,\sigma_0,h,\Sigma_Y)\).

\subsection{Discussion and Interpretation}\label{sec:discuss}
\paragraph{Conditional distribution as an endogenous state variable.}
The defining distinction between DRU and DDU is the role of the conditional preference distribution
\(\mu_s=\mathcal L(X_s\mid\mathcal F_s^Y)\).
Under DRU,
\((u,b,\sigma,\sigma_0,h,\Sigma_Y)\)
is independent of
\(\mu\),
so that
\(\mu_s\)
is generated by
\((X,Y)\)
through
\(\mathcal F_s^Y\)
but does not enter either the felicity functional or the state dynamics.
Hence,
\((X,Y)\mapsto\mathcal F^Y\mapsto\mu\),
where
\(\mu\)
is a posterior process determined by Bayesian updating of the exogenous latent state.
Under DDU,
\((u,b,\sigma,\sigma_0,h,\Sigma_Y)
=
(u,b,\sigma,\sigma_0,h,\Sigma_Y)(s,X_s,Y_s,\mu_s)\),
so that
\(\mu_s\)
enters both
\(u(s,\cdot,X_s,\mu_s)\)
and
\((b,\sigma,\sigma_0,h,\Sigma_Y)(s,X_s,Y_s,\mu_s)\).
Accordingly,
\((X,Y)\mapsto\mathcal F^Y\mapsto\mu\mapsto(u,b,\sigma,\sigma_0,h,\Sigma_Y)\mapsto(X,Y)\),
and the conditional law becomes an endogenous state variable rather than a passive posterior distribution.
Observed choice therefore affects subsequent behavior through both
\(\mathcal F_s^Y\mapsto\mu_s\)
and
\(\mu_s\mapsto(u,b,\sigma,\sigma_0,h,\Sigma_Y)\),
thereby coupling Bayesian learning with endogenous distributional feedback. Consequently, DDU combines the informational structure of DRU with the conditional MVSDE of \citet{CarmonaDelarue2018}.

\paragraph{Interpretation of the observation process.}
The observable process
\(Y=(Y_s)_{0\le s\le t}\)
is an arbitrary progressively measurable signal generating
\(\mathbb F^Y=(\mathcal F_s^Y)_{0\le s\le t}\),
with no restriction that
\(Y\)
coincide with realized choices.
Depending on the application,
\(Y\)
may represent purchase histories, revealed actions, transaction prices, portfolio allocations, or other observable signals, provided
\(\mathcal F_s^Y=\sigma(Y_r:0\le r\le s)\).
Accordingly,
\(\mathbb F^Y\)
is the analyst's information filtration, whereas the decision maker acts on
\(\mathbb F=(\mathcal F_s)_{0\le s\le t}\),
with
\(\mathbb F^Y\subseteq\mathbb F\).
The CPD
\(\mu_s=\mathcal L(X_s\mid\mathcal F_s^Y)\)
is therefore determined by
\((X,Y)\)
through
\((X,Y)\mapsto\mathcal F^Y\mapsto\mu\),
rather than by direct observation of
\(X\).
Consequently,
\(\mu_s\)
summarizes the analyst's posterior distribution of the latent preference state and, under DDU,
\((u,b,\sigma,\sigma_0,h,\Sigma_Y)
=
(u,b,\sigma,\sigma_0,h,\Sigma_Y)(s,X_s,Y_s,\mu_s)\),
so that distribution dependence operates through the conditional law
\(\mu_s\)
instead of complete observation of
\(X_s\).

\paragraph{Relation to existing models.}
DDU contains continuous-time DRU as the subclass satisfying
\((u,b,\sigma,\sigma_0,h,\Sigma_Y)
=
(u,b,\sigma,\sigma_0,h,\Sigma_Y)(s,x,y)\),
that is,
\((u,b,\sigma,\sigma_0,h,\Sigma_Y)\)
is independent of
\(\mu\).
It also differs from classical MVSDEs, where the coefficient dependence is through the unconditional law
\(\mathcal L(X_s)\),
whereas DDU replaces
\(\mathcal L(X_s)\)
by the conditional law
\(\mu_s=\mathcal L(X_s\mid\mathcal F_s^Y)\),
with
\(\mathcal F_s^Y=\sigma(Y_r:0\le r\le s)\).
Accordingly,
\((u,b,\sigma,\sigma_0,h,\Sigma_Y)
=
(u,b,\sigma,\sigma_0,h,\Sigma_Y)(s,X_s,Y_s,\mu_s)\),
and the induced evolution is
\((X,Y)\mapsto\mathcal F^Y\mapsto\mu\mapsto(X,Y)\),
rather than
\(X\mapsto\mathcal L(X)\mapsto X\)
as in classical MVSDEs \citep{pramanik2025construction}. Thus the endogenous state variable is the filter-valued process
\(\mu=(\mu_s)_{0\le s\le t}\in C([0,t];\mathcal P_2(\mathbb R^d))\),
instead of the unconditional distribution process
\((\mathcal L(X_s))_{0\le s\le t}\), consistent with the conditional MVSDE of \citet{CarmonaDelarueLacker2016,CarmonaDelarue2018,buckdahn2023general}.

\section{Behavioral Implications of DDU}\label{sec:behavior}

Section~\ref{sec:dynamicrandom} established the DDU model
\((X,Y,\mu)\), with \(\mu_s=\mathcal L(X_s\mid\mathcal F_s^Y)\),
its well-posedness, and the coefficient tuple
\((u,b,\sigma,\sigma_0,h,\Sigma_Y)\).
This section studies the induced behavioral mapping
\((X,Y,\mu)\mapsto\rho\),
where
\(\rho=(\rho_s)_{0\le s\le t}\)
denotes the stochastic-choice process generated by DDU.
We first characterize distributional feedback, then derive the behavioral representation of distribution-dependent preferences, establish the associated identification results, and conclude with special cases and comparative statics.

\subsection{Distributional Feedback}
\label{sec:feedback}

DDU generates stochastic choice through
\(\mu_s=\mathcal L(X_s\mid\mathcal F_s^Y)\)
and the coefficient dependence
\((u,b,\sigma,\sigma_0,h,\Sigma_Y)
=
(u,b,\sigma,\sigma_0,h,\Sigma_Y)(s,X_s,Y_s,\mu_s)\).
Accordingly,
\((X,Y)\mapsto\mathcal F^Y\mapsto\mu\mapsto(u,b,\sigma,\sigma_0,h,\Sigma_Y)\mapsto\rho\),
where
\(\rho=(\rho_s)_{0\le s\le t}\)
denotes the induced stochastic-choice process.
This subsection discusses the behavioral restrictions generated by the mapping
\(\mu\mapsto(u,b,\sigma,\sigma_0,h,\Sigma_Y)\),
separating variation induced by the conditional law
\(\mu_s=\mathcal L(X_s\mid\mathcal F_s^Y)\)
from variation induced by the latent state
\(X_s\),
thereby distinguishing composition effects from distributional feedback. For a finite menu $A\in\mathcal K(Z)$, $z\in A$, and
$(s,x,\mu)\in[0,t]\times\mathbb R^d\times\mathcal P_2(\mathbb R^d)$,
define
\[
\chi_s(z;A,x,\mu)
:=
\mathbf 1
\left\{
u(s,z,x,\mu)>
\max_{a\in A\setminus\{z\}}u(s,a,x,\mu)
\right\}.
\]
Thus $\chi_s(z;A,x,\mu)$ indicates that $z$ is the unique maximizer when the privately observed state is $x$ and the distributional state is $\mu$. Ties will be ruled out below.

\begin{definition}\label{def:decoupled-choice}
	For $\nu,\mu\in\mathcal P_2(\mathbb R^d)$, define
	$
	\mathsf C_s(z;A\mid\nu,\mu)
	:=
	\int_{\mathbb R^d}
	\chi_s(z;A,x,\mu)\,\nu(dx).
	$
\end{definition}
\noindent
The kernel
\(\mathsf C_s(\cdot;\cdot\mid\nu,\mu)\)
depends separately on
\(\nu,\mu\in\mathcal P_2(\mathbb R^d)\),
with
\(\nu\mapsto\chi_s(\cdot;\cdot,\cdot,\mu)\)
determining the composition of latent states and
\(\mu\mapsto\chi_s(\cdot;\cdot,\cdot,\mu)\)
the measure dependence of the felicity index.
Hence,
\(\mathsf C_s(\cdot;\cdot\mid\nu,\mu)\)
decouples the mappings
\(\nu\mapsto X_s\)
and
\(\mu\mapsto u(s,\cdot,X_s,\mu)\).
The cases
\(\nu\neq\mu\)
and
\(\nu=\mu\)
correspond, respectively, to counterfactual and equilibrium evaluations.
Along every equilibrium path,
\(\nu=\mu=\mu_s=\mathcal L(X_s\mid\mathcal F_s^Y)\),
so that
\(\mathsf C_s(\cdot;\cdot\mid\mu_s,\mu_s)\)
coincides with the observed stochastic-choice kernel in
\eqref{eq:equilibrium-choice-kernel}
\begin{equation}
	\rho_s(z;A)
	=
	\mathsf C_s(z;A\mid\mu_s,\mu_s),
	\qquad z\in A.
	\label{eq:equilibrium-choice-kernel}
\end{equation}

\begin{as}
	\label{ass:no-ties}
	For every $s\in[0,t]$, $A\in\mathcal K(Z)$, $z\ne a$ in $A$, and $\nu,\mu\in\mathcal P_2(\mathbb R^d)$ that arise under an admissible history,
	$
	\nu\left(
	\left\{
	x\in\mathbb R^d:
	u(s,z,x,\mu)=u(s,a,x,\mu)
	\right\}
	\right)=0.
	$
\end{as}

\begin{rmk}
Assumption~\ref{ass:no-ties} implies that
$\mathsf C_s(\cdot;A\mid\nu,\mu)$ is a probability distribution on $A$. It can be replaced by an explicit measurable tie-breaking rule without changing the results below.
\end{rmk}

\begin{definition}\label{def:choice-relevant-felicity-feedback}
DDU exhibits choice-relevant felicity feedback at time $s$ if there exist a finite menu $A$, an alternative $z\in A$, and
	$\nu,\mu,\mu'\in\mathcal P_2(\mathbb R^d)$ such that
	$
	\mathsf C_s(z;A\mid\nu,\mu)
	\ne
	\mathsf C_s(z;A\mid\nu,\mu').
	$ It is \emph{behaviorally felicity-neutra}l if the preceding equality fails for no such collection.
\end{definition}

\noindent
Definition~\ref{def:choice-relevant-felicity-feedback}
depends only on the ordering induced by
$u(s,\cdot,x,\mu)$,
not on its cardinal representation.
Hence,
$u(s,\cdot,x,\mu)$
and
$u(s,\cdot,x,\mu')$
are behaviorally equivalent whenever they induce the same ordering on
$Z$,
or, equivalently,
$\operatorname{sgn}\!\bigl(\Delta_{za}(s,x,\mu)\bigr)
=
\operatorname{sgn}\!\bigl(\Delta_{za}(s,x,\mu')\bigr)$
for every
$z,a\in Z$.
Accordingly, only changes in the collection
$\{\Delta_{za}(s,x,\mu):z,a\in Z\}$
can alter
$\mathsf C_s(\cdot;\cdot\mid\nu,\mu)$.
For
$z,a\in Z$,
define
$
\Delta_{za}(s,x,\mu)
:=
u(s,z,x,\mu)
-
u(s,a,x,\mu).
$

\begin{lem} \label{lem:common-distributional-term}
Suppose
	$
	u(s,z,x,\mu)
	=
	\bar u(s,z,x)+\alpha(s,x,\mu)
	$
	for every $(s,z,x,\mu)$, where $\alpha$ is independent of $z$. Then DDU is behaviorally felicity-neutral. More generally, felicity feedback is behaviorally neutral if and only if, for every admissible $s,x,\mu,\mu'$, the orderings of $Z$ induced by $u(s,\cdot,x,\mu)$ and $u(s,\cdot,x,\mu')$ coincide outside sets that are null under every admissible composition measure.
\end{lem}

\noindent Proof  is  in Appendix~\ref{app:common-distributional-term}. Definition \ref{def:choice-relevant-felicity-feedback} isolates contemporaneous feedback. Preference feedback is dynamic, a change in the distributional state changes the
law of future latent preferences and therefore may affect future
choice even when current rankings are unchanged. To formulate this, fix $0\le s<r\le t$ and let
$
m=(m_v)_{s\le v\le r}
\in C([s,r];\mathcal P_2(\mathbb R^d))
$
be an externally specified measure flow. Starting from
$(X_s,Y_s)\sim\lambda$, consider the frozen-flow system
\[
\begin{cases}
	dX_v^m
	=
	b(v,X_v^m,Y_v^m,m_v)\,dv
	+\sigma(v,X_v^m,Y_v^m,m_v)\,dW_v
	+\sigma_0(v,X_v^m,Y_v^m,m_v)\,dB_v,\\[2mm]
	dY_v^m
	=
	h(v,X_v^m,Y_v^m,m_v)\,dv
	+\Sigma_Y(v,Y_v^m,m_v)\,dB_v.
\end{cases}
\]
Under Assumption~1 and the conditions used in
Proposition~\ref{prop:ddu-wellposed}, this system has a weak solution for every admissible $m$. Let $\mathsf P_{s,r}^{m}\lambda$ denote the law of
$(X_r^m,Y_r^m)$ generated by this system. Its $X$-marginal is denoted by
$
\mathsf P_{s,r}^{m,X}\lambda.
$

\begin{definition}
	\label{def:dynamic-distributional-response}
	For $A\in\mathcal K(Z)$ and $z\in A$, define
	$
	\mathsf D_{s,r}(z;A\mid\lambda,m)
	:=
	\mathbb E
	\left[
	\chi_r
	\left(
	z;A,X_r^m,m_r
	\right)
	\right].
	$
	DDU exhibits a dynamic distributional response on $[s,r]$ if there exist $\lambda$, $m$, $m'$, $A$, and $z\in A$ such that
	$
	\mathsf D_{s,r}(z;A\mid\lambda,m)
	\ne
	\mathsf D_{s,r}(z;A\mid\lambda,m').
	$
\end{definition}

\noindent
The kernel
$\mathsf D_{s,r}^{\,\bar\mu}$
decouples the transition law
$m\mapsto\mathsf P_{s,r}^{m}$
from the felicity argument
$\bar\mu$.
For
$\lambda\in\mathcal P_2(\mathbb R^d\times\mathbb R^{d_Y})$
and
$m,\bar\mu\in C([0,t];\mathcal P_2(\mathbb R^d))$,
define
$$
\mathsf D_{s,r}^{\,\bar\mu}(z;A\mid\lambda,m)
:=
\int_{\mathbb R^d\times\mathbb R^{d_Y}}
\chi_r(z;A,x,\bar\mu)\,
\mathsf P_{s,r}^{m}\lambda(dx,dy),
$$
where
$\bar\mu$
is fixed.
Thus
$m\mapsto\mathsf P_{s,r}^{m}$
determines the distribution of
$X_r$,
whereas
$\bar\mu\mapsto\chi_r(\cdot;\cdot,\cdot,\bar\mu)$
is held constant.
Consequently,
$\mathsf D_{s,r}^{\,\bar\mu}(\cdot;\cdot\mid\lambda,m)$
isolates dynamic variation through
$\mathsf P_{s,r}^{m}$,
independently of contemporaneous felicity.

\begin{definition}\label{def:behaviorally-relevant-preference-feedback}
DDU exhibits behaviorally relevant preference feedback on $[s,r]$ if there exist $\lambda$, $m$, $m'$, $\bar\mu$, $A$, and $z\in A$ such that
	$
	\mathsf D_{s,r}^{\,\bar\mu}(z;A\mid\lambda,m)
	\ne
	\mathsf D_{s,r}^{\,\bar\mu}(z;A\mid\lambda,m').
	$
\end{definition}

\begin{rmk}Definition~\ref{def:behaviorally-relevant-preference-feedback} holds felicity fixed at the comparison date. Any difference in the resulting choice probabilities must therefore operate through the effect of the measure flow on the distribution of future latent states.
\end{rmk}

\begin{definition}
	\label{def:behavioral-distributional-feedback}
	Let
	$\mathcal F_s$
	denote the property in
	Definition~\ref{def:choice-relevant-felicity-feedback}
	and
	$\mathcal P_{s,r}$
	the property in
	Definition~\ref{def:behaviorally-relevant-preference-feedback}.
	Then DDU exhibits
	\emph{behavioral distributional feedback}
	iff
	$
	\mathcal F_s
	\vee
	\mathcal P_{s,r},
	$
	and satisfies
	\emph{distributional invariance}
	iff
	$
	\neg\mathcal F_s
	\wedge
	\neg\mathcal P_{s,r}.
	$
\end{definition}

\noindent
Definition~\ref{def:behavioral-distributional-feedback}
depends only on the induced behavioral mappings
$
\mu
\mapsto
\mathsf C
$
and
$
m
\mapsto
\mathsf D,
$
rather than the primitive mappings
$
\mu
\mapsto
(u,b,\sigma,\sigma_0,h,\Sigma_Y).
$
Hence,
$
(u,b,\sigma,\sigma_0,h,\Sigma_Y)
\not\equiv
(\tilde u,\tilde b,\tilde\sigma,\tilde\sigma_0,\\ \tilde h,\tilde\Sigma_Y)
$
may nevertheless satisfy
$
\mathsf C=\tilde{\mathsf C},
\
\mathsf D=\tilde{\mathsf D},
$
so that stochastic-choice data identify behavioral objects but not, in general, primitive coefficients without additional restrictions.

\begin{prop}
	\label{prop:dru-distributional-invariance}
	Suppose the conditions of Proposition~\ref{prop:ddu-nesting} hold. Then, for every
	$s\in[0,t]$,
	finite menu
	$A\subseteq Z$,
	$z\in A$,
	and
	$\nu,\mu,\mu'\in\mathcal P_2(\mathbb R^d)$,
	$
	\mathsf C_s(z;A\mid\nu,\mu)
	=
	\mathsf C_s(z;A\mid\nu,\mu').
	$
	Moreover, for every
	$0\le s<r\le t$,
	admissible initial law
	$\lambda$,
	measure flows
	$m,m'$,
	fixed felicity argument
	$\bar\mu$,
	finite menu
	$A$,
	and
	$z\in A$,
	$
	\mathsf D_{s,r}^{\,\bar\mu}(z;A\mid\lambda,m)
	=
	\mathsf D_{s,r}^{\,\bar\mu}(z;A\mid\lambda,m').
	$
	Hence, every DRU representation satisfies distributional invariance.
\end{prop}

\noindent
Proof is in Appendix~\ref{app:dru-distributional-invariance}. The converse requires behavioral richness:
\(\mathsf C=\mathsf C'\)
and
\(\mathsf D=\mathsf D'\)
need not imply
\((u,b,\sigma,\sigma_0,h,\Sigma_Y)
=
(u',b',\sigma',\sigma_0',h',\Sigma_Y')\)
without additional identification conditions.

\begin{as}\label{ass:choice-separation}
	Let $\mathcal Q\subseteq\mathcal P_2(\mathbb R^d)$ denote the collection of latent-state distributions reachable under admissible initial conditions and measure flows. For any distinct $\nu,\nu'\in\mathcal Q$, there exist a time $s$, a finite menu $A\in\mathcal K(Z)$, an alternative $z\in A$, and an admissible distributional state $\mu$ such that
	$
	\int_{\mathbb R^d}
	\chi_s(z;A,x,\mu)\,\nu(dx)
	\ne
	\int_{\mathbb R^d}
	\chi_s(z;A,x,\mu)\,\nu'(dx).
	$
\end{as}

\begin{rmk}
	Assumption~\ref{ass:choice-separation} asserts injectivity of the map
	$
	\nu
	\mapsto
	\left(
	\int_{\mathbb R^d}
	\chi_s(z;A,x,\mu)\,\nu(dx)
	\right)_{(s,A,z,\mu)},
	$
	restricted to
	$\mathcal Q$.
	Equivalently,
	$\nu\neq\nu'$
	implies
	$\mathsf C_s(\cdot;\cdot\mid\nu,\mu)
	\neq
	\mathsf C_s(\cdot;\cdot\mid\nu',\mu)$
	for some admissible
	$(s,A,z,\mu)$.
	No single
	$(A,\mu)$
	is required to separate
	$\mathcal Q$.
\end{rmk}

\begin{prop}
	\label{prop:detectability-preference-feedback}
	Suppose Assumption~\ref{ass:choice-separation} holds. Fix
	$0\le s<r\le t$ and an admissible initial law $\lambda$. If two admissible measure flows $m$ and $m'$ satisfy
	$
	\mathsf P_{s,r}^{m,X}\lambda
	\ne
	\mathsf P_{s,r}^{m',X}\lambda,
	$
	then DDU exhibits behaviorally relevant preference feedback on
	$[s,r]$.
\end{prop}

\noindent
Proof is in Appendix~\ref{app:detectability-preference-feedback}.
Proposition~\ref{prop:detectability-preference-feedback}
establishes the implication
$
m\neq m'
\Longrightarrow
\mathsf P_{s,r}^{m}\neq
\mathsf P_{s,r}^{m'}
\Longrightarrow
\mathsf D_{s,r}^{\,\bar\mu}(\cdot;\cdot\mid\lambda,m)
\neq
\mathsf D_{s,r}^{\,\bar\mu}(\cdot;\cdot\mid\lambda,m'),
$
under Assumption~\ref{ass:choice-separation}.
We next impose regularity on the mapping
$
m\mapsto
\mathsf D.
$
For
$\mu\in\mathcal P_2(\mathbb R^d)$,
define
$
m_2(\mu)
:=
\int_{\mathbb R^d}|x|^2\,\mu(dx).
$

\begin{as}
	\label{ass:utility-margin}
	There exist constants $L_u,\kappa>0$ such that
	\begin{enumerate}
		\item for every $s$, $z$, $x$, $x'$, $\mu$, and $\mu'$,
		$
		\left|
		u(s,z,x,\mu)-u(s,z,x',\mu')
		\right|
		\le
		L_u\left(
		|x-x'|+W_2(\mu,\mu')
		\right);
		$
		\item for every admissible $s$, $z\ne a$, $\nu$, $\mu$, and
		$\eta>0$,
		$
		\nu\left(
		\left\{
		x:
		|\Delta_{za}(s,x,\mu)|\le\eta
		\right\}
		\right)
		\le
		\kappa\eta.
		$
	\end{enumerate}
\end{as}
\noindent The second condition bounds the probability mass near indifference. It is a dynamic analogue of the standard margin condition used to control discontinuous choice indicators.

\begin{lem} \label{lem:choice-operator-stability}
Suppose Assumptions~\ref{ass:no-ties} and
\ref{ass:utility-margin} hold. Let $A\in\mathcal K(Z)$ and $z\in A$. For each $\nu,\nu',\mu,\mu'\in\mathcal P_2(\mathbb R^d)$,
	\[
	\left|
	\mathsf C_s(z;A\mid\nu,\mu)
	-
	\mathsf C_s(z;A\mid\nu',\mu')
	\right|
	\le
	K_A
	\left(
	W_2(\nu,\nu')^{1/2}
	+
	W_2(\mu,\mu')^{1/2}
	\right),
	\]
	where $K_A<\infty$ depends only on $|A|$, $L_u$, and $\kappa$.
	Consequently, the equilibrium choice kernel in
	\eqref{eq:equilibrium-choice-kernel} is continuous in the CPD.
\end{lem}

\noindent
Proof is in Appendix~\ref{app:choice-operator-stability}.
Lemma~\ref{lem:choice-operator-stability}
establishes
$
(\nu,\mu)\mapsto
\mathsf C_s(\cdot;\cdot\mid\nu,\mu)
$
is
$W_2^{1/2}$-continuous.
For
$m,m'\in C([0,t];\mathcal P_2(\mathbb R^d))$,
define
$
d_{s,r}(m,m')
:=
\sup_{s\le v\le r}
W_2(m_v,m_v'),
$
which induces the continuity metric for the dynamic choice operator.

\begin{prop}
	\label{prop:dynamic-choice-stability}
	Suppose Assumptions~1, \ref{ass:no-ties}, and
	\ref{ass:utility-margin} hold, together with the strengthened
	stability conditions in the Supplementary Appendix. Then, for every
	$0\le s<r\le t$, finite menu $A$, and $z\in A$, there exists
	$K_{s,r,A}<\infty$ such that
	\[
	\left|
	\mathsf D_{s,r}(z;A\mid\lambda,m)
	-
	\mathsf D_{s,r}(z;A\mid\lambda',m')
	\right|
	\le
	K_{s,r,A}
	\left(
	W_2(\lambda,\lambda')^{1/2}
	+
	d_{s,r}(m,m')^{1/2}
	\right).
	\]
	The same conclusion holds for
	$\mathsf D_{s,r}^{\,\bar\mu}$, uniformly over $\bar\mu$ in subsets of
	$\mathcal P_2(\mathbb R^d)$ with uniformly bounded second moments.
\end{prop}

\noindent
The proof is given in Supplementary Appendix~SB.1.
Proposition~\ref{prop:dynamic-choice-stability}
establishes
$
(\lambda,m)
\longmapsto
\mathsf D_{s,r}(\cdot;\cdot\mid\lambda,m)
$
is
\(
(W_2,d_{s,r})^{1/2}
\)-continuous,
and the same estimate holds for
$
\mathsf D_{s,r}^{\,\bar\mu}
$
uniformly over
$
\{\bar\mu:m_2(\bar\mu)\le R\}
$
for every
$
R<\infty.
$
Hence
$
(\lambda^n,m^n)\to(\lambda,m)
$
implies
$
\mathsf D_{s,r}(\cdot;\cdot\mid\lambda^n,m^n)
\to
\mathsf D_{s,r}(\cdot;\cdot\mid\lambda,m).
$
The following assumption imposes injectivity of the mapping
$
m\mapsto
\mathsf P^m,
$
thereby excluding distinct conditional-law flows with identical transition kernels.

\begin{as}
	\label{ass:dynamic-nonredundancy}
	Suppose
	$(b,\sigma,\sigma_0)(\cdot,\mu)\neq(b,\sigma,\sigma_0)(\cdot,\mu')$
	for some
	$\mu\neq\mu'$
	on the reachable state space.
	Then there exist
	$0\le s<r\le t$,
	an admissible initial law
	$\lambda$,
	and admissible measure flows
	$m,m'$
	such that
	$
	\mathsf P_{s,r}^{m,X}\lambda
	\neq
	\mathsf P_{s,r}^{m',X}\lambda.
	$
\end{as}

\noindent
Assumption~\ref{ass:dynamic-nonredundancy}
requires the mapping
$
m
\mapsto
\mathsf P_{s,r}^{m,X}
$
to be nondegenerate on the reachable class.
Equivalently,
$
(b,\sigma,\sigma_0)
\not\equiv
(\tilde b,\tilde\sigma,\tilde\sigma_0)
$
implies
$
\mathsf P_{s,r}^{m,X}
\neq
\mathsf P_{s,r}^{m',X}
$
for some admissible
$(s,r,\lambda,m,m')$.
For
$(\sigma,\sigma_0)$,
the relevant behavioral object is the induced covariance operator, not a particular matrix factorization.

\begin{prop}
	\label{prop:structural-behavioral-feedback}
	Suppose Assumptions~\ref{ass:no-ties},
	\ref{ass:choice-separation}, and
	\ref{ass:dynamic-nonredundancy} hold.
	Then
	$
	u(\cdot,\mu)\neq u(\cdot,\mu')
	$
	for some
	$\mu\neq\mu'$
	implies choice-relevant felicity feedback.
	Moreover,
	$
	(b,\sigma,\sigma_0)(\cdot,\mu)
	\neq
	(b,\sigma,\sigma_0)(\cdot,\mu')
	$
	for some
	$\mu\neq\mu'$
	on the reachable state space implies behaviorally relevant preference feedback.
	Conversely,
	$
	(u,b,\sigma,\sigma_0)
	=
	(u,b,\sigma,\sigma_0)(s,x,y)
	$
	is independent of
	$\mu$
	if and only if DDU satisfies behavioral distributional invariance.
\end{prop}

\noindent Proof is in Appendix~\ref{app:structural-behavioral-feedback}.
Proposition~\ref{prop:structural-behavioral-feedback}
identifies the implication
$
(u,b,\sigma,\sigma_0)
\longmapsto
(\mathsf C,\mathsf D),
$
with
$
u(\cdot,\mu)\sim u(\cdot,\mu')
$
whenever
$
\Delta_{za}(\cdot,\mu)
=
\Delta_{za}(\cdot,\mu')
$
for every
$z,a\in Z$.
For observed histories
$y_{\cdot\wedge s}$
and
$y'_{\cdot\wedge s}$,
write
$
y_{\cdot\wedge s}
\sim_{\mathrm{CPD}}
y'_{\cdot\wedge s}
$
iff
$
\mathcal L(X_s\mid y_{\cdot\wedge s})
=
\mathcal L(X_s\mid y'_{\cdot\wedge s}),
$
and
$
y_{\cdot\wedge s}
\sim_{\mathrm{Comp}}
y'_{\cdot\wedge s}
$
iff the induced distributions of
$X_s$
coincide, irrespective of the associated counterfactual distributional states.

\begin{lem}
	\label{lem:filter-state-sufficiency}
	Suppose system~\eqref{eq:ddu-system} is Markovian. Let
	$y_{\cdot\wedge s}$,
	$y'_{\cdot\wedge s}$
	satisfy
	$
	(Y_s,\mu_s)(y_{\cdot\wedge s})
	=
	(Y_s,\mu_s)(y'_{\cdot\wedge s}).
	$
	Then, for every admissible menu policy
	$\pi$,
	$
	\mathcal L^\pi
	\!\left(
	(Z_r)_{r\in[s,t]}
	\,\middle|\,
	y_{\cdot\wedge s}
	\right)
	=
	\mathcal L^\pi
	\!\left(
	(Z_r)_{r\in[s,t]}
	\,\middle|\,
	y'_{\cdot\wedge s}
	\right).
	$
\end{lem}

\noindent
The proof is given in Supplementary Appendix~SB.2.
Lemma~\ref{lem:filter-state-sufficiency}
establishes
$
y_{\cdot\wedge s}
\sim_{\mathrm{CPD}}
y'_{\cdot\wedge s}
\Longrightarrow
\mathcal L^\pi
(\,\cdot\mid y_{\cdot\wedge s})
=
\mathcal L^\pi
(\,\cdot\mid y'_{\cdot\wedge s}),
$
so that
$(Y_s,\mu_s)$
is a sufficient state variable for continuation behavior.
Under DRU,
$
\mu_s
\mapsto
\mathsf C_s
$
through the composition of
$X_s$
only; under DDU,
$
\mu_s
\mapsto
(u,b,\sigma,\sigma_0)
\mapsto
(\mathsf C,\mathsf D).
$
The next subsection characterizes these behavioral mappings.

\subsection{Behavioral Characterization}\label{sec:characterization}

Subsection~\ref{sec:feedback} establishes that distributional feedback is behaviorally relevant only insofar as it changes observable stochastic choice. We now characterize precisely when a DDU
representation is behaviorally distinguishable from a DRU representation. Fix
$0\le s<r\le t$,
$A\in\mathcal K(Z)$,
$z\in A$,
$\lambda\in\mathcal P_2(\mathbb R^d\times\mathbb R^{d_Y})$,
and admissible measure flows
$m,m'\in\mathcal M_{s,r}$.
Define
$\nu_r^{m,\lambda}:=P^{m,X}_{s,r}\lambda$
and
$\nu_r^{m',\lambda}:=P^{m',X}_{s,r}\lambda$.
Since,
$
D_{s,r}(z;A\mid\lambda,m)
=
C_r(z;A\mid\nu_r^{m,\lambda},m_r),
$
we obtain the decomposition
\begin{align}
	D_{s,r}(z;A\mid\lambda,m)
	-
	D_{s,r}(z;A\mid\lambda,m')
	&=
	\Big[
	C_r(z;A\mid\nu_r^{m,\lambda},m_r)
	-
	C_r(z;A\mid\nu_r^{m,\lambda},m_r')
	\Big]
	\nonumber\\
	&\quad+
	\Big[
	C_r(z;A\mid\nu_r^{m,\lambda},m_r')
	-
	C_r(z;A\mid\nu_r^{m',\lambda},m_r')
	\Big].
	\label{eq:behavioral-decomposition}
\end{align}
The first term of \eqref{eq:behavioral-decomposition} measures the behavioral effect of
changing the distributional state while holding the latent-state distribution fixed. The
second measures the effect of changing the latent-state distribution while holding the
felicity argument fixed. For $s\in[0,t]$
and
$\mu,\mu'\in\mathcal P_2(\mathbb R^d)$,
define
$
\Gamma_s(\mu,\mu')
:=
\bigcup_{z\neq a}
\left\{
x\in\mathbb R^d:
\operatorname{sgn}\Delta_{za}(s,x,\mu)
\neq
\operatorname{sgn}\Delta_{za}(s,x,\mu')
\right\},
$ where $\Delta_{za}(s,x,\mu)
:=
u(s,z,x,\mu)-u(s,a,x,\mu)$. Since, $Z$ is finite,
$x\notin\Gamma_s(\mu,\mu')$
if and only if
$\chi_s(z;A,x,\mu)=\chi_s(z;A,x,\mu')$
for every
$A\in\mathcal K(Z)$
and every
$z\in A$.
Consequently,
$\nu(\Gamma_s(\mu,\mu'))=0$
implies
$C_s(z;A\mid\nu,\mu)
=
C_s(z;A\mid\nu,\mu')$
for every admissible
$(A,z)$.

\begin{as}
	\label{ass:behavioral-separation}
		1. For every admissible
		$(s,\nu,\mu,\mu')$
		with
		$\nu(\Gamma_s(\mu,\mu'))>0$,
		there exist
		$A\in\mathcal K(Z)$
		and
		$z\in A$
		such that
		$
		\int_{\mathbb R^d}
		\left(
		\chi_s(z;A,x,\mu)
		-
		\chi_s(z;A,x,\mu')
		\right)
		\nu(dx)
		\neq
		0.
		$\\ 2. For every pair of distinct reachable distributions
		$\nu,\nu'\in\mathcal P_2(\mathbb R^d)$,
		there exist
		$A\in\mathcal K(Z)$,
		$z\in A$,
		and
		$\bar\mu\in\mathcal P_2(\mathbb R^d)$
		so that
		$
		\int_{\mathbb R^d}
		\chi_r(z;A,x,\bar\mu)\nu(dx)
		\neq
		\int_{\mathbb R^d}
		\chi_r(z;A,x,\bar\mu)\nu'(dx).
		$
	\end{as}

\noindent
Assumption~\ref{ass:behavioral-separation}
imposes injectivity of the mappings
$
\mu
\longmapsto
\mathsf C_s(\cdot;\cdot\mid\nu,\mu)
$
and
$
\nu
\longmapsto
\mathsf C_r(\cdot;\cdot\mid\nu,\bar\mu),
$
restricted to the admissible class.
Equivalently,
$
\mu\neq\mu'
$
implies
$
\mathsf C_s(\cdot;\cdot\mid\nu,\mu)
\neq
\mathsf C_s(\cdot;\cdot\mid\nu,\mu'),
$
whenever
$
\nu(\Gamma_s(\mu,\mu'))>0,
$
and
$
\nu\neq\nu'
$
implies
$
\mathsf C_r(\cdot;\cdot\mid\nu,\bar\mu)
\neq
\mathsf C_r(\cdot;\cdot\mid\nu',\bar\mu)
$
for some admissible
$(A,z,\bar\mu)$.
Thus Assumption~\ref{ass:behavioral-separation}
is the behavioral analogue of the injectivity condition in
Assumption~\ref{ass:choice-separation}.

\begin{definition}
	\label{def:behavioral-invariance}
	A DDU representation satisfies
	\emph{behavioral distributional invariance}
	on its reachable domain if $
	C_s(z;A\mid\nu,\mu)
	=
	C_s(z;A\mid\nu,\mu')
	$
	for every admissible
	$(s,A,z,\nu,\mu,\mu')$,
	and
	$
	D^{\bar\mu}_{s,r}(z;A\mid\lambda,m)
	=
	D^{\bar\mu}_{s,r}(z;A\mid\lambda,m')
	$ for every admissible $(s,r,A,z,\lambda,\bar\mu,m,m')$.
	\end{definition}

\noindent
Definition~\ref{def:behavioral-invariance}
imposes invariance of the operators
$
\mu
\longmapsto
C_s(\cdot;\cdot\mid\nu,\mu)
$
and
$
m
\longmapsto
D_{s,r}^{\,\bar\mu}(\cdot;\cdot\mid\lambda,m),
$
with
$
\nu,\bar\mu
$
held fixed.
Equivalently,
$
C_s(\cdot;\cdot\mid\nu,\mu)
=
C_s(\cdot;\cdot\mid\nu,\mu')
$
for all admissible
$(\nu,\mu,\mu')$,
and
$
D_{s,r}^{\,\bar\mu}(\cdot;\cdot\mid\lambda,m)
=
D_{s,r}^{\,\bar\mu}(\cdot;\cdot\mid\lambda,m')
$
for all admissible
$(\lambda,\bar\mu,m,m')$.
The following result characterizes behavioral distributional invariance through these two operator identities.

\begin{lem}
	\label{lem:behavioral-characterization}
	Suppose Assumption~\ref{ass:behavioral-separation}
	holds. Then behavioral distributional invariance is equivalent to
	$
	\nu(\Gamma_s(\mu,\mu'))
	=
	0
	$ for every admissible
	$(s,\nu,\mu,\mu')$, and
	$
	P^{m,X}_{s,r}\lambda
	=
	P^{m',X}_{s,r}\lambda
	$
	for every admissible
	$(s,r,\lambda,m,m')$.
	\end{lem}

\noindent Proof is in Appendix~\ref{proofl17}. To state the representation result, call a pair
$\bigl(\bar u,\{\bar P^X_{s,r}\}_{0\leq s<r\leq t}\bigr)$
a \emph{DRU reduction} of the DDU representation on the reachable domain if $\bar u:[0,t]\times Z\times\mathbb R^d\to\mathbb R$ is independent of the
distributional state, $\bar P^X_{s,r}\lambda$ is independent of the conditional-law flow, and the induced stochastic choice probabilities agree with those of the DDU representation at every admissible argument. Define
$
\bar\chi_s(z;A,x)
:=
\mathbf 1
\left\{
\bar u(s,z,x)>
\max_{a\in A\setminus\{z\}}
\bar u(s,a,x)
\right\},
$
such that 
$
C_s(z;A\mid\nu,\mu)
=
\int_{\mathbb R^d}
\bar\chi_s(z;A,x)\,\nu(dx)
$
for every admissible $(s,A,z,\nu,\mu)$, and
$
D_{s,r}^{\bar\mu}(z;A\mid\lambda,m)
=
\int_{\mathbb R^d}
\bar\chi_r(z;A,x)\,
\bar P^X_{s,r}\lambda(dx)
$
for every admissible $(s,r,A,z,\lambda,\bar\mu,m)$.

\begin{theorem} \label{thm:behavioral-characterization}
	Suppose Assumption~\ref{ass:behavioral-separation} holds. The following statements are equivalent.\\
	1. The DDU representation satisfies behavioral distributional invariance on its reachable domain.\\
	2. For every admissible $(s,\nu,\mu,\mu')$,
		$\nu(\Gamma_s(\mu,\mu'))=0$, and, for every admissible
		$(s,r,\lambda,m,m')$,
		$
		P^{m,X}_{s,r}\lambda
		=
		P^{m',X}_{s,r}\lambda.
		$\\
		3. The DDU representation admits a DRU reduction on its reachable
		domain.\\
	Consequently, the DDU representation is behaviorally distinct from every DRU
	reduction on its reachable domain if and only if at least one of the following
	conditions holds
	$
	\nu\bigl(\Gamma_s(\mu,\mu')\bigr)>0
	$
	for some admissible $(s,\nu,\mu,\mu')$, or
	$
	P^{m,X}_{s,r}\lambda
	\neq
	P^{m',X}_{s,r}\lambda
	$
	for some admissible $(s,r,\lambda,m,m')$.
\end{theorem}

\noindent Proof is in Appendix~\ref{proofth18}. \noindent
Theorem~\ref{thm:behavioral-characterization}
characterizes DRU reduction through the pair
$
\mu\mapsto(\chi_s(z;A,\cdot,\mu))_{A,z}
$
and
$
m\mapsto P_{s,r}^{m,X}\lambda.
$
Under Assumption~\ref{ass:behavioral-separation},
$
\nu(\Gamma_s(\mu,\mu'))=0
$
for all admissible
$
(s,\nu,\mu,\mu')
$
if and only if the former is behaviorally constant, while
$
P_{s,r}^{m,X}\lambda
=
P_{s,r}^{m',X}\lambda
$
for all admissible
$
(s,r,\lambda,m,m')
$
if and only if the latter is behaviorally constant. Hence
$
(u,b,\sigma,\sigma_0)
$
is behaviorally reducible to DRU precisely when both maps are constant on the reachable domain. For convenience, Define
$
\mathfrak R_{\mathrm{DRU}}
:=
\{\rho:\exists\,\text{DRU representation of }\rho\},
$ and $
\mathfrak R_{\mathrm{DDU}}
:=
\{\rho:\exists\,\text{DDU representation of }\rho\}.
$

\begin{theorem}[Behavioral impossibility of DRU]
	\label{thm:dru-impossibility}
	Suppose Assumption~\ref{ass:behavioral-separation} holds. If
	$
	\rho\in\mathfrak R_{\mathrm{DDU}}
	$
	exhibits behavioral distributional feedback, then
	$
	\rho\notin\mathfrak R_{\mathrm{DRU}}.
	$
	Equivalently,
	$
	\rho\in
	\mathfrak R_{\mathrm{DDU}}
	\setminus
	\mathfrak R_{\mathrm{DRU}}.
	$
\end{theorem}

\noindent Proof is in Supplementary Appendix SB.3. Theorem~\ref{thm:dru-impossibility} yields
$
\mathfrak R_{\mathrm{DDU}}
\setminus
\mathfrak R_{\mathrm{DRU}}
\neq
\varnothing.
$
Thus behavioral distributional feedback is not a refinement of DRU but a behavioral separator:
$
\mathfrak R_{\mathrm{DRU}}
$
is characterized by behavioral distributional invariance, whereas
$
\mathfrak R_{\mathrm{DDU}}
$
strictly enlarges this class through
$
\mathcal F_s
\vee
\mathcal P_{s,r}.
$

\subsection{Identification}\label{sec:identification}

Fix an observational domain
$
\mathfrak O
:=
\mathfrak O^{0}\cup\mathfrak O^{1},
$
where
$
\mathfrak O^{0}
:=
\big\{
(s,A,z,\nu,\mu):
(s,\nu,\mu)\in\mathfrak R_s,\;
A\in\mathcal K(Z),\;
z\in A
\big\},
$
and
$
\mathfrak O^{1}
:=
\big\{
(s,r,A,z,\lambda,\bar\mu,m):
m\in\mathfrak M_{s,r}(\lambda),\;
A\in\mathcal K(Z),\;
z\in A
\big\}.
$
The observable stochastic-choice array generated by
$\theta:=(u,b,\sigma,\sigma_0)$ is
\[
\mathscr Q_\theta
:=
\left(
\left\{
C_s^\theta(z;A\mid\nu,\mu)
\right\}_{(s,A,z,\nu,\mu)\in\mathfrak O^{0}},
\left\{
D_{s,r}^{\theta,\bar\mu}(z;A\mid\lambda,m)
\right\}_{(s,r,A,z,\lambda,\bar\mu,m)\in\mathfrak O^{1}}
\right).
\]
Equivalently, writing
$
\chi_s^\theta(z;A,x,\mu)
:=
\mathbf 1
\big\{
u(s,z,x,\mu)>
\max_{a\in A\setminus\{z\}}u(s,a,x,\mu)
\big\},
$
and
$
\nu_{s,r}^{\theta;m,\lambda}
:=
P_{s,r}^{\theta;m,X}\lambda,
$
the observation map
$\mathscr T:\Theta\to[0,1]^{\mathfrak O}$ satisfies
\[
\mathscr T(\theta)
=
\left(
\left\{
\int_{\mathbb R^d}
\chi_s^\theta(z;A,x,\mu)\,\nu(dx)
\right\}_{\mathfrak O^{0}},
\left\{
\int_{\mathbb R^d}
\chi_r^\theta(z;A,x,\bar\mu)\,
\nu_{s,r}^{\theta;m,\lambda}(dx)
\right\}_{\mathfrak O^{1}}
\right).
\]

For $\theta,\tilde\theta\in\Theta$, write
$\theta\sim_{\mathfrak O}\tilde\theta$ whenever
$\mathscr T(\theta)=\mathscr T(\tilde\theta)$. Thus
\[
\theta\sim_{\mathfrak O}\tilde\theta
\quad\Longleftrightarrow\quad
\begin{cases}
	C_s^\theta(z;A\mid\nu,\mu)
	=
	C_s^{\tilde\theta}(z;A\mid\nu,\mu),
	&
	(s,A,z,\nu,\mu)\in\mathfrak O^{0},
	\\[2mm]
	D_{s,r}^{\theta,\bar\mu}(z;A\mid\lambda,m)
	=
	D_{s,r}^{\tilde\theta,\bar\mu}(z;A\mid\lambda,m),
	&
	(s,r,A,z,\lambda,\bar\mu,m)\in\mathfrak O^{1}.
\end{cases}
\]
The identified set at $\theta$ is therefore the fiber
$
\mathcal I_{\mathfrak O}(\theta)
:=
\mathscr T^{-1}\!\left(\{\mathscr T(\theta)\}\right)
=
[\theta]_{\sim_{\mathfrak O}}.
$
Identification of a functional
$\psi:\Theta\to\Psi$ requires
$
\theta\sim_{\mathfrak O}\tilde\theta
\ \rightarrow\ 
\psi(\theta)=\psi(\tilde\theta),
$
whereas point identification of $\theta$ requires
$\mathcal I_{\mathfrak O}(\theta)=\{\theta\}$.
Behavioral identification is governed by the image of the observation map
rather than by the primitive parameter
$\theta=(u,b,\sigma,\sigma_0)$.
Accordingly, define the behavioral functional
$
\Phi(\theta)
:=
\bigl(\Phi_u(\theta),\Phi_P(\theta)\bigr),
$
where
$
\Phi_u(\theta)
:=
\left(
\chi_s^\theta(z;A,\cdot,\mu)
\right)_
{\substack{
		(s,A,z,\mu):
		\\
		(s,\mu)\ \mathrm{admissible}
}},
$
and
$
\Phi_P(\theta)
:=
\left(
P_{s,r}^{\theta;m,X}
\right)_
{\substack{
		(s,r,m):
		\\
		(s,r,m)\ \mathrm{admissible}
}}.
$
Therefore,
$
\Phi:
\Theta
\rightarrow
\mathfrak X_u\times\mathfrak X_P,
$
where
$\mathfrak X_u$
denotes the product space of measurable ranking maps and
$\mathfrak X_P$
the product space of reachable transition operators. Since
$
\mathscr T
=
\Lambda\circ\Phi
$
for the evaluation operator
$\Lambda:\mathfrak X_u\times\mathfrak X_P
\rightarrow
[0,1]^{\mathfrak O}$,
identification of $\theta$ factors through the behavioral image
$\Phi(\theta)$.
Consequently,
$
\theta
\sim_{\mathfrak O}
\tilde\theta
\ \rightarrow\ 
\Phi(\theta)
=
\Phi(\tilde\theta)
$
is sufficient for behavioral identification, while
$
\Phi(\theta)
=
\Phi(\tilde\theta)
\not\rightarrow
\theta=\tilde\theta
$
in general because distinct primitives may induce identical ranking maps and
transition operators. For
$\varphi,\tilde\varphi\in\Phi(\Theta)$,
write
$
\varphi\equiv_{\mathfrak O}\tilde\varphi
\ \leftrightarrow\ 
\Lambda(\varphi)=\Lambda(\tilde\varphi),
$
and let
$
[\varphi]_{\equiv_{\mathfrak O}}
:=
\left\{
\tilde\varphi\in\Phi(\Theta):
\tilde\varphi\equiv_{\mathfrak O}\varphi
\right\}.
$
For $\theta\in\Theta$, define the behavioral identified set
$
\mathcal I_{\mathfrak O}^{\Phi}(\theta)
:=
\left\{
\Phi(\tilde\theta):
\tilde\theta\in\mathcal I_{\mathfrak O}(\theta)
\right\}
\subseteq
\Phi(\Theta).
$

\begin{prop}
	\label{prop:identification-fibers}
	For every $\theta\in\Theta$,
	$
	\mathcal I_{\mathfrak O}(\theta)
	=
	\Phi^{-1}
	\!\left(
	[\Phi(\theta)]_{\equiv_{\mathfrak O}}
	\right),$
	and $
	\mathcal I_{\mathfrak O}^{\Phi}(\theta)
	=
	[\Phi(\theta)]_{\equiv_{\mathfrak O}}.
	$
	Hence,
	$
	\Phi
	\text{ is identified}
	\iff
	\left.
	\Lambda
	\right|_{\Phi(\Theta)}
	\text{ is injective}.
	$
	Moreover, a functional
	$\psi:\Theta\to\Psi$
	is identified if and only if
	$
	\psi
	\text{ is constant on }
	\Phi^{-1}
	\!\left(
	[\varphi]_{\equiv_{\mathfrak O}}
	\right)
	$
	for every
	$\varphi\in\Phi(\Theta)$.
\end{prop}

\noindent Proof is in Appendix~\ref{proofprop19}. Proposition~\ref{prop:identification-fibers} reduces identification to
injectivity of
$\left.\Lambda\right|_{\Phi(\Theta)}$.
Writing
$\varphi=(\varphi_u,\varphi_P)$ with
$
\varphi_u
=
\left(
\chi_s(z;A,\cdot,\mu)
\right)_{s,A,z,\mu},
$ and $
\varphi_P
=
\left(
P_{s,r}^{m,X}
\right)_{s,r,m},
$
the corresponding observational equivalence class is
\[
[\varphi]_{\equiv_{\mathfrak O}}
=
\left\{
(\tilde\varphi_u,\tilde\varphi_P)\in\Phi(\Theta):
\begin{array}{l}
	\displaystyle
	\int_{\mathbb R^d}
	\varphi_{u;s,A,z,\mu}(x)\,\nu(dx)
	=
	\int_{\mathbb R^d}
	\tilde\varphi_{u;s,A,z,\mu}(x)\,\nu(dx),
	\\[3mm]
	\displaystyle
	\int_{\mathbb R^d}
	\varphi_{u;r,A,z,\bar\mu}(x)\,
	\varphi_{P;s,r,m}\lambda(dx)
	=
	\int_{\mathbb R^d}
	\tilde\varphi_{u;r,A,z,\bar\mu}(x)\,
	\tilde\varphi_{P;s,r,m}\lambda(dx)
\end{array}
\right\},
\]
where the equalities hold for all indices in
$\mathfrak O^{0}$ and $\mathfrak O^{1}$, respectively. For $(s,\mu)$, let
\[
\mathfrak N_s(\mu)
:=
\left\{
\nu\in\mathcal P_2(\mathbb R^d):
(s,\nu,\mu)\in\mathfrak R_s
\right\},
\qquad
\mathcal N_s(\mu)
:=
\left\{
B\in\mathcal B(\mathbb R^d):
\sup_{\nu\in\mathfrak N_s(\mu)}\nu(B)=0
\right\},
\]
and write $f=_{\mathfrak N_s(\mu)}g$ whenever
$\{f\neq g\}\in\mathcal N_s(\mu)$. For $r\in[0,t]$, define the choice-test
class
\[
\mathscr H_r
:=
\left\{
\chi_r^\theta(z;A,\cdot,\mu):
\theta\in\Theta,\;
A\in\mathcal K(Z),\;
z\in A,\;
(r,\mu)\ \mathrm{admissible}
\right\}
\subseteq
\mathcal B_b(\mathbb R^d),
\]
and the reachable-law class
$
\mathscr P_r
:=
\bigg\{
P_{s,r}^{\theta;m,X}\lambda:
\theta\in\Theta,\;
0\leq s<r,\;
m\in\mathfrak M_{s,r}(\lambda)
\bigg\}.
$

\begin{as}
	\label{ass:observational-separation}
	For every admissible $(s,\mu)$,
	$\mathfrak N_s(\mu)$ separates
	$\mathscr H_s/\!={}_{\mathfrak N_s(\mu)}$.
	Moreover, for every $r\in(0,t]$,
	$\mathscr H_r$
	is measure determining on
	$\mathscr P_r$.
\end{as}

\noindent Assumption~\ref{ass:observational-separation} identifies
$\Phi_u$ through the family
$\mathfrak N_s(\mu)$ and
$\Phi_P$ through the test class $\mathscr H_r$.

\begin{theorem}[Behavioral identification]
	\label{thm:behavioral-identification}
	Suppose	Assumptions~\ref{ass:behavioral-separation}
	and \ref{ass:observational-separation}
	hold. Then
	$
	\mathcal I_{\mathfrak O}^{\Phi}(\theta)
	=
	\{\Phi(\theta)\},
	\
	\forall\,\theta\in\Theta.
	$
	Equivalently,
	$
	\mathscr T(\theta)
	=
	\mathscr T(\tilde\theta)
	\ \leftrightarrow\ 
	\Phi(\theta)
	=
	\Phi(\tilde\theta),
	\ \forall\ 
	\theta,\tilde\theta\in\Theta.
	$
	Consequently,
	$
	\Phi(\Theta)
	=
	\mathscr T(\Theta),
	$
	up to the canonical identification induced by
	$\Lambda|_{\Phi(\Theta)}$.
\end{theorem}
\noindent Proof is in Appendix~\ref{proofth21}. Theorem~\ref{thm:behavioral-identification} identifies the behavioral representation
$
\Phi=(\Phi_u,\Phi_P),
$
rather than the primitive coefficient tuple
$
(u,b,\sigma,\sigma_0,h,\Sigma_Y).
$
Moreover, stochastic choice data identify precisely those objects that are invariant under observational equivalence. DDU is therefore identified behaviorally rather than parametrically.

\subsection{Special Cases and Comparative Statics}
\label{sec:special-cases}

Theorems~\ref{thm:behavioral-characterization} and~\ref{thm:behavioral-identification} characterize the behavioral image of DDU by the pair
$
(\Phi_u,\Phi_P),
$
where
$
\Phi_u
=
\left(\chi_s(z;A,\cdot,\mu)\right)_{(s,A,z,\mu)}
$
and
$
\Phi_P
=
\left(P^{m,X}_{s,r}\right)_{(s,r,m)}.
$
Consequently, every comparative-static exercise is induced by a perturbation of
$
(\Phi_u,\Phi_P)
$
rather than of the primitive coefficients
$
(u,b,\sigma,\sigma_0).
$
Let
$
\{\theta_\varepsilon:\varepsilon\in E\}\subset\Theta
$
be a family of admissible DDU representations and define
$
\Phi(\theta_\varepsilon)
:=
(\Phi_u^\varepsilon,\Phi_P^\varepsilon).
$
Therefore, observable variation is completely determined by the path
$
\varepsilon
\mapsto
\Phi(\theta_\varepsilon),
$
since
$
\mathscr T
=
\Lambda\circ\Phi
$
and
Theorem~\ref{thm:behavioral-identification} identifies
$
\Phi(\Theta)
$
from
$
\mathscr T(\Theta).
$
Hence, comparative statics reduce to studying perturbations of the behavioral quotient
$
\Phi(\Theta)
$
rather than perturbations of individual representatives
$
\theta\in\Theta.
$
\begin{prop}
	\label{prop:behavioral-stability}
	Let
	$\theta,\tilde\theta\in\Theta$
	satisfy
	$\Phi(\theta)=\Phi(\tilde\theta)$.
	Then, for every admissible
	$(s,r,\lambda,\mu,m,A,z)$,
	$
	C_s^\theta(z;A\mid\nu,\mu)
	=
	C_s^{\tilde\theta}(z;A\mid\nu,\mu),
	$ and $
	D_{s,r}^{\theta,\bar\mu}(z;A\mid\lambda,m)
	=
	D_{s,r}^{\tilde\theta,\bar\mu}(z;A\mid\lambda,m).
	$
	Moreover, every behavioral functional
	$
	\Psi:
	\Phi(\Theta)
	\rightarrow
	\mathbb R
	$
	admits a unique extension to observational equivalence classes, whereas no functional varying within a fiber
	$\Phi^{-1}(\varphi)$
	is identified from stochastic choice.
\end{prop}

\noindent Proof is in Appendix~\ref{proofprop22}. Proposition \ref{prop:behavioral-stability} identifies the maximal behavioral invariant of the model. Since
$\mathscr T=\Lambda\circ\Phi$
and
$\Lambda|_{\Phi(\Theta)}$
is injective by Theorem~\ref{thm:behavioral-identification}, every observable comparative-static statement factors through
$\Phi(\Theta)$.
Accordingly, any perturbation
$\theta\mapsto\theta_\varepsilon$
with
$\Phi(\theta_\varepsilon)\equiv\Phi(\theta)$
is behaviorally null, whereas every observable perturbation necessarily induces a nontrivial path
$\varepsilon\mapsto\Phi(\theta_\varepsilon)$.
Hence the economically meaningful comparative statics are those of the behavioral image rather than of individual coefficient representations. Proposition~\ref{prop:behavioral-stability} identifies the behavioral quotient: primitive specifications are identified only through
$
\Phi.
$

\begin{prop}
	\label{prop:distribution-independent}
	Suppose
	$
	u(s,z,x,\mu)=\bar u(s,z,x)
	$
	and
	$
	(b,\sigma,\sigma_0,h,\Sigma_Y)
	=
	(\bar b,\bar\sigma,\bar\sigma_0,\bar h,\bar\Sigma_Y)
	$
	are independent of
	$\mu$.
	Then
	$
	\Phi(\theta)
	=
	\Phi_{\mathrm{DRU}}(\bar\theta),
	$ where $
	\bar\theta
	:=
	(\bar u,\bar b,\bar\sigma,\bar\sigma_0,\bar h,\bar\Sigma_Y)
	\in\Theta_{\mathrm{DRU}},
	$
	and, for every admissible
	$(s,r,\lambda,\nu,m,A,z)$,
	$
	C_s^{\mathrm{DDU}}
	=
	C_s^{\mathrm{DRU}}
	$
	and
	$
	D_{s,r}^{\mathrm{DDU}}
	=
	D_{s,r}^{\mathrm{DRU}}.
	$
	Consequently,
	$
	\mathscr T_{\mathrm{DDU}}
	=
	\mathscr T_{\mathrm{DRU}},
	$
	so DDU and DRU are observationally equivalent on the reachable domain.
\end{prop}

\noindent Proof is in Appendix~\ref{proofprop23}. Proposition~\ref{prop:distribution-independent} identifies DRU with the subimage
$
\Phi_{\mathrm{DRU}}(\Theta)
\subseteq
\Phi(\Theta).
$
Hence
$
\Phi(\Theta)\setminus\Phi_{\mathrm{DRU}}(\Theta)
$
coincides with the behaviorally nontrivial component of DDU. By
Theorem~\ref{thm:behavioral-identification},
$
\Phi(\Theta)\neq\Phi_{\mathrm{DRU}}(\Theta)
$
if and only if
$
\mathscr T_{\mathrm{DDU}}
\neq
\mathscr T_{\mathrm{DRU}},
$
so distribution dependence is an observable property of the behavioral quotient.
\begin{theorem}
	\label{thm:behavioral-rigidity}
	Let
	$\{\theta_\varepsilon:\varepsilon\in E\}\subset\Theta$
	be an admissible family with $0\in E$. Suppose
	$\theta_0$ admits a DRU reduction
	$\bar\theta_0\in\Theta_{\mathrm{DRU}}$.
	The following are equivalent.
	
	(i)
	$
	\Phi(\theta_\varepsilon)
	=
	\Phi(\theta_0)
	$
	for every
	$\varepsilon\in E$.
	
	(ii)
	$
	\mathscr T(\theta_\varepsilon)
	=
	\mathscr T(\theta_0)
	$
	for every
	$\varepsilon\in E$.
	
	(iii)
	Every
	$
	\theta_\varepsilon
	$
	belongs to the observational fiber
	$
	\Phi^{-1}\!\bigl(\Phi(\theta_0)\bigr).
	$
	
	(iv) There exists a single
	$\bar\theta\in\Theta_{\mathrm{DRU}}$
	such that
	$
	\Phi(\theta_\varepsilon)
	=
	\Phi_{\mathrm{DRU}}(\bar\theta)
	$
	for every
	$\varepsilon\in E$.
	
	Consequently, every nontrivial comparative statics of DDU satisfies
	$
	\Phi(\theta_\varepsilon)
	\neq
	\Phi(\theta_0),
	$
	equivalently,
	$
	\mathscr T(\theta_\varepsilon)
	\neq
	\mathscr T(\theta_0)
	$
	for some
	$\varepsilon\in E$.
\end{theorem}
\noindent Proof is in Appendix~\ref{proofth24}.
 Theorem~\ref{thm:behavioral-rigidity} implies
$
\mathscr T
=
\Lambda\circ\Phi
$
factors through the quotient
$
\Theta/\!\sim_{\mathfrak O},
$
with fibers
$
\Phi^{-1}(\varphi),
$
$\varphi\in\Phi(\Theta)$,
forming the maximal behaviorally invariant partition of
$\Theta$.
Accordingly,
$
\theta_\varepsilon
\mapsto
\mathscr T(\theta_\varepsilon)
$
is locally constant iff
$
\theta_\varepsilon
\subseteq
\Phi^{-1}(\varphi)
$
for some
$
\varphi\in\Phi(\Theta),
$
whereas
$
\Phi(\theta_\varepsilon)\neq\Phi(\theta_0)
$
implies
$
\mathscr T(\theta_\varepsilon)\neq\mathscr T(\theta_0).
$
Hence
$
\Phi(\Theta)\setminus\Phi_{\mathrm{DRU}}(\Theta)
$
is precisely the identified component of distribution dependence, so every observable comparative statics is generated by motion in
$
\Theta/\!\sim_{\mathfrak O}
$
rather than within an observational fiber. Theorems~\ref{thm:behavioral-identification} and~\ref{thm:behavioral-rigidity} together with Proposition~\ref{prop:behavioral-stability} show that the empirical content of DDU is
$
\Phi=(\Phi_u,\Phi_P),
$
rather than the primitive coefficient tuple.

\section{Conditional McKean-Vlasov Preferences}
\label{sec:cmkv}

Sections~\ref{sec:ddu} and~\ref{sec:feedback} characterize DDU through its behavioral image
$
\Phi=(\Phi_u,\Phi_P)
$
and establish identification from the observable stochastic-choice array
$
\mathscr T=\Lambda\circ\Phi.
$
The underlying source of these behavioral restrictions is the conditional-law feedback
$
\mu_s=\mathcal L(X_s\mid\mathcal F_s^Y),
$
which enters both the instantaneous utility index and the preference dynamics. This section studies the resulting conditional McKean-Vlasov structure. Unlike
classical MVSDEs, the state variable is the pair
$
(X_s,\mu_s),
$
where
$
\mu_s
$
is itself generated endogenously by the observation filtration. Finally, the evolution of preferences is governed by the coupled system
$
(X,\mu),
$
rather than by the latent state
$
X
$
alone. Throughout this section, let
$
\Gamma:
C([0,t];\mathcal P_2(\mathbb R^d))
\rightarrow
C([0,t];\mathcal P_2(\mathbb R^d))
$
denote the conditional-law operator introduced in Supplementary Appendix~SA.2, and write
$
\mu=\Gamma(\mu)
$
for the associated fixed-point relation. 

\subsection{Distribution-Dependent Preference Dynamics}
\label{sec:cmkv-dynamics}

Fix $0\leq s\leq r\leq t$ and an admissible weak solution
$(X,Y,\mu)$ of system~\eqref{eq:ddu-system}. The privately observed state is
$(X_s,Y_s,\mu_s)$, whereas the analyst observes only
$(Y_s,\mu_s)$, with
$\mu_s=\mathcal L(X_s\mid\mathcal F_s^Y)$.
Accordingly, the relevant state augmentation is
$
X_s\mapsto\mathbf X_s:=(X_s,Y_s,\mu_s)
$
for the decision maker and
$
Y_s\mapsto\widehat{\mathbf X}_s:=(Y_s,\mu_s)
$
for the analyst. For every bounded Borel
$\varphi:\mathbb R^d\to\mathbb R$,
$
\langle\mu_s,\varphi\rangle
=
\mathbb E\left[\varphi(X_s)\mid\mathcal F_s^Y\right],
$
so variation in
$
\widehat{\mathbf X}_s
$
simultaneously changes the analyst's conditional composition of latent preferences and,
through
$
(u,b,\sigma,\sigma_0,h,\Sigma_Y)(s,\cdot,\mu_s),
$
the continuation law of
$
\mathbf X
$.
Therefore, the transition
$
\widehat{\mathbf X}_s\mapsto\widehat{\mathbf X}_r
$
contains both Bayesian updating,
$
\mathcal F_s^Y\subseteq\mathcal F_r^Y,
$
and structural feedback,
$
\mu_v\mapsto
(b,\sigma,\sigma_0,h,\Sigma_Y)(v,\cdot,\mu_v),
$
for $s\leq v\leq r$.

\begin{definition}
	\label{def:conditional-preference-flow}
	Given an admissible initial law
	$\lambda\in\mathcal P_2(\mathbb R^d\times\mathbb R^{d_Y})$
	at time $s$, a \emph{conditional preference flow} on $[s,r]$ is an
	$\mathbb F^Y$-adapted process
	$\mu^{s,\lambda}=\left(\mu_v^{s,\lambda}\right)_{s\leq v\leq r}$
	with paths in
	$C([s,r];\mathcal P_2(\mathbb R^d))$ such that
	$
	\mathbb E\left[\sup_{s\leq v\leq r}m_2(\mu_v^{s,\lambda})\right]<\infty
	$
	and, for every bounded Borel $\varphi$,
	$
	\langle\mu_v^{s,\lambda},\varphi\rangle
	=
	\mathbb E\left[\varphi(X_v^{s,\lambda})\mid\mathcal F_v^Y\right],
	$
	$\mathbb P$-a.s. The associated augmented preference states are
	$
	\mathbf X_v^{s,\lambda}
	:=
	(X_v^{s,\lambda},Y_v^{s,\lambda},\mu_v^{s,\lambda})
	$
	and
	$
	\widehat{\mathbf X}_v^{s,\lambda}
	:=
	(Y_v^{s,\lambda},\mu_v^{s,\lambda}).
	$
\end{definition}

\noindent Definition~\ref{def:conditional-preference-flow} induces the state decomposition
$
\mathbf X_v^{s,\lambda}
=
(X_v^{s,\lambda},\widehat{\mathbf X}_v^{s,\lambda}),
$
with
$
\widehat{\mathbf X}_v^{s,\lambda}
=
(Y_v^{s,\lambda},\mu_v^{s,\lambda})
$
and
$
\mu_v^{s,\lambda}
=
\mathcal L(X_v^{s,\lambda}\mid\mathcal F_v^Y).
$
Hence
$
\chi_v
=
\chi_v(\cdot,\widehat{\mathbf X}_v^{s,\lambda}),
$
$
C_v
=
C_v(\cdot\mid\widehat{\mathbf X}_v^{s,\lambda}),
$
and
$
P_{v,q}^{X}
=
P_{v,q}^{X}(\widehat{\mathbf X}_v^{s,\lambda})
$
for
$
s\le v\le q\le r.
$
Accordingly,
$
(Y_v,\mu_v)
=
(\widetilde Y_v,\widetilde\mu_v)
$
implies
$
(\chi_v,C_v,P_{v,q}^{X})
=
(\widetilde\chi_v,\widetilde C_v,\widetilde P_{v,q}^{X}),
$
whereas
$
\mu_v\neq\widetilde\mu_v
$
permits
$
(\chi_v,C_v,P_{v,q}^{X})
\neq
(\widetilde\chi_v,\widetilde C_v,\widetilde P_{v,q}^{X})
$
even when
$
Y_v=\widetilde Y_v.
$
Thus the behavioral state is represented by
$
\widehat{\mathbf X}_v=(Y_v,\mu_v),
$
rather than by the observable coordinate
$
Y_v
$
alone.

\begin{lem}
	\label{lem:cmkv-state}
	Let
	$
	(X,Y,\mu)
	$
	be an admissible DDU solution and let
	$
	\mathbf X_s
	=
	(X_s,Y_s,\mu_s),
	$
	$
	\widehat{\mathbf X}_s
	=
	(Y_s,\mu_s).
	$
	Suppose the coefficients and every admissible continuation policy are
	Markovian in
	$
	(X_v,Y_v,\mu_v),
	$
	and that the associated continuation martingale problem is weakly unique
	for every admissible initial state.
	Then, for every
	$
	s\le r\le t,
	$
	every bounded Borel functional
	$
	F:
	C([s,r];\mathbb R^d\times\mathbb R^{d_Y}\times\mathcal P_2(\mathbb R^d))
	\rightarrow
	\mathbb R,
	$
	and every admissible continuation policy,
	$
	\mathbb E
	\!\left[
	F\!\left(
	(\mathbf X_v)_{v\in[s,r]}
	\right)
	\middle|
	\mathcal F_s^Y
	\right]
	=
	\mathcal G_{s,r}^F
	(\widehat{\mathbf X}_s),
	$
	for some measurable operator
	$
	\mathcal G_{s,r}^F.
	$
	Consequently,
	$
	\widehat{\mathbf X}_s
	$
	is a sufficient state variable for the conditional evolution of every admissible continuation experiment.
\end{lem}
\noindent Proof is in Appendix~\ref{prooflem25}.
	Lemma~\ref{lem:cmkv-state} identifies the conditional semigroup
	$
	(\mathcal G_{s,r})_{0\le s\le r\le t}
	$
	on the measurable state space
	$
	\mathbb R^{d_Y}\times\mathcal P_2(\mathbb R^d).
	$
	Hence
	$
	\mathcal L
	\!\left(
	(\mathbf X_v)_{v\ge s}
	\mid
	\mathcal F_s^Y
	\right)
	=
	\mathcal L
	\!\left(
	(\mathbf X_v)_{v\ge s}
	\mid
	\widehat{\mathbf X}_s
	\right),
	$
	$
	P_{s,r}^{X}
	=
	P_{s,r}^{X}(\widehat{\mathbf X}_s),
	$
	and
	$
	\Gamma_{s,r}
	=
	\Gamma_{s,r}(\widehat{\mathbf X}_s).
	$
	Thus every continuation object
	$
	(\chi,C,D,\Phi_P,\Gamma)
	$
	factors through
	$
	\widehat{\mathbf X}_s,
	$
	so the infinite-dimensional conditional law
	$
	\mu_s
	$
	is the unique additional state coordinate required beyond
	$
	Y_s
	$
	to close the preference dynamics.

\begin{prop}
	\label{prop:conditional-markov}
	Suppose
	Lemma~\ref{lem:cmkv-state}
	hold.
	Then, for every
	$
	0\le s\le r\le q\le t,
	$
	every bounded Borel functional
	$
	F:
	C([r,q];
	\mathbb R^d\times
	\mathbb R^{d_Y}\times
	\mathcal P_2(\mathbb R^d))
	\rightarrow
	\mathbb R,
	$
	and every admissible continuation policy,
	$
	\mathbb E
	\!\left[
	F((\mathbf X_v)_{v\in[r,q]})
	\middle|
	\mathcal F_r^Y
	\right]
	=
	\mathcal G_{r,q}^F
	(\widehat{\mathbf X}_r),
	$
	where
	$
	(\mathcal G_{r,q})_{0\le r\le q\le t}
	$
	forms a measurable transition family satisfying
	$
	\mathcal G_{s,q}
	=
	\mathcal G_{s,r}
	\mathcal G_{r,q}.
	$
	Consequently,
	$
	(\widehat{\mathbf X}_s)_{0\le s\le t}
	$
	is a time-inhomogeneous Markov process with state space
	$
	\mathbb R^{d_Y}
	\times
	\mathcal P_2(\mathbb R^d).
	$
\end{prop}

\noindent Proof is in Supplementary Appendix SC.1.
Proposition~\ref{prop:conditional-markov}
identifies
$
(\mathcal G_{s,r})_{0\le s\le r\le t}
$
as the nonlinear transition semigroup induced by the conditional-law dynamics.
Hence,
$
\Gamma_{s,q}
=
\Gamma_{r,q}
\circ
\Gamma_{s,r},
$
$
P_{s,q}^{X}
=
P_{r,q}^{X}
P_{s,r}^{X},
$
and
$
\Phi_P
=
(\mathcal G_{s,r})_{0\le s\le r\le t}.
$
Therefore, the endogenous evolution of
$
\mu
$
is dynamically closed on
$
\mathbb R^{d_Y}
\times
\mathcal P_2(\mathbb R^d),
$
whereas the latent coordinate
$
X
$
alone does not generate a Markovian state description.

\begin{theorem}
	\label{thm:cmkv-representation}
	Suppose Proposition~\ref{prop:conditional-markov}
	hold. Then there exists a unique family of measurable operators
	$
	(\Gamma_{s,r})_{0\le s\le r\le t}
	$
	on
	$
	\mathbb R^{d_Y}
	\times
	\mathcal P_2(\mathbb R^d)
	$
	such that, for every admissible initial law
	$
	\lambda
	$
	and every admissible continuation policy,
	$
	\widehat{\mathbf X}_r
	=
	\Gamma_{s,r}
	(\widehat{\mathbf X}_s),
	$
	$
	P_{s,r}^{X}
	=
	P^{X}(\Gamma_{s,r}),
	$
	and
	$
	\Phi_P
	=
	(\Gamma_{s,r})_{0\le s\le r\le t}.
	$
	Moreover,
	$
	\Gamma_{s,q}
	=
	\Gamma_{r,q}
	\circ
	\Gamma_{s,r},
	$
	$
	\Gamma_{s,s}
	=
	I,
	$
	and the DDU dynamics admit the representation
	$
	(X,Y,\mu)
	\leftrightarrow
	(\widehat{\mathbf X},\Gamma),
	$
	which is unique up to indistinguishability.
\end{theorem}

\noindent  Proof is in Supplementary Appendix SC.2.
Theorem~\ref{thm:cmkv-representation} identifies
$
\Gamma
=
(\Gamma_{s,r})_{0\le s\le r\le t}
$
as the nonlinear evolution operator of DDU.
Hence,
$
\Phi_P
=
\Gamma,
$
$
\mathscr T
=
\Lambda(\Phi_u,\Gamma),
$
and
$
(Y_s,\mu_s)
\mapsto
(Y_r,\mu_r)
$
is completely determined by
$
\Gamma_{s,r}.
$
Accordingly,
$
(X,\mu)
$
constitutes a conditional McKean-Vlasov structure,
whereas
$
\Gamma
$
is the corresponding behavioral transition semigroup.
The dependence
$
\Gamma_{s,r}
=
\Gamma_{s,r}(\mu)
$
is the mathematical source of the endogenous distributional feedback characterized in
Sections~\ref{sec:feedback}-\ref{sec:special-cases}.
Theorem~\ref{thm:cmkv-representation} characterizes the DDU dynamics by the nonlinear evolution
$
\Gamma=(\Gamma_{s,r})_{0\le s\le r\le t},
$
so
$
\widehat{\mathbf X}
\rightarrow
\Gamma
\rightarrow
\Phi_P
\rightarrow
\mathscr T.
$
Hence,
$
\Phi_P
$
is no longer an exogenous transition family but the image of the endogenous conditional-law evolution. Consequently,
$
\Gamma
$
cannot be specified independently of
$
\mu,
$
since
$
\mu
=
\mathcal L(X\mid\mathcal F^Y)
$
and
$
\Gamma
=
\Gamma(\mu)
$
are jointly determined through the conditional McKean--Vlasov fixed-point relation. Section~\ref{sec:cmkv-fixedpoint} studies this endogenous fixed-point structure.

\subsection{Endogenous Information and Fixed-Point Structure}
\label{sec:cmkv-fixedpoint}

Section~\ref{sec:cmkv-dynamics} discusses the evolution of the augmented state
$
\widehat{\mathbf X}
=
(Y,\mu)
$
through the nonlinear transition family
$
\Gamma
=
(\Gamma_{s,r})_{0\le s\le r\le t}.
$
The remaining question is whether
$
\Gamma
$
may be prescribed exogenously or must itself be determined jointly with the conditional preference law.
Since
$
\mu_s
=
\mathcal L(X_s\mid\mathcal F_s^Y),
$
the coefficients
$
(u,b,\sigma,\sigma_0,h,\Sigma_Y)
$
generate
$
\mu,
$
while
$
\mu
$
simultaneously determines the coefficients. Consequently,
$
\Gamma
$
is not an exogenous transition family but the image of an endogenous fixed-point on
$
C([0,t];\mathcal P_2(\mathbb R^d)).
$
Throughout this subsection, let
$
\mathfrak M
:=
C([0,t];\mathcal P_2(\mathbb R^d))
$
and let
$
\Gamma:
\mathfrak M
\rightarrow
\mathfrak M
$
denote the conditional-law operator associated with system~\eqref{eq:ddu-system}.
Accordingly, every admissible preference evolution satisfies
$
\mu
=
\Gamma(\mu),
$
whereas arbitrary elements of
$
\mathfrak M
$
need not correspond to admissible conditional preference dynamics.

\begin{definition}
	\label{def:endogenous-fixed-point}
	A conditional preference flow
	$
	\mu\in\mathfrak M
	$
	is an \emph{endogenous fixed point} if, for every
	$
	0\le s\le r\le t,
	$
	$
	\mu_r
	=
	\Gamma_{s,r}(\mu_s),
	$
	equivalently,
	$
	\mu
	=
	\Gamma(\mu).
	$
	The set of all endogenous fixed points is denoted by
	$
	\operatorname{Fix}(\Gamma)
	:=
	\{
	\mu\in\mathfrak M:
	\Gamma(\mu)=\mu
	\}.
	$
	The associated DDU representation is said to be \emph{internally consistent} whenever
	$
	\operatorname{Fix}(\Gamma)\neq\varnothing.
	$
\end{definition}
\noindent
Definition~\ref{def:endogenous-fixed-point} identifies
$
\operatorname{Fix}(\Gamma)
$
as the admissible subset of
$
\mathfrak M,
$
with
$
\Gamma:
\mathfrak M\rightarrow\mathfrak M
$
defined jointly by
$
(u,b,\sigma,\sigma_0,h,\Sigma_Y)
$
and
$
\mu_v
=
\mathcal L(X_v\mid\mathcal F_v^Y).
$
Hence
$
\mu\in\operatorname{Fix}(\Gamma)
$
if and only if the induced continuation kernels satisfy
$
P_{s,r}^{X}
=
P^{X}(\Gamma_{s,r}),
$
$
\Phi_P
=
(\Gamma_{s,r})_{0\le s\le r\le t},
$
and
$
\mathscr T
=
\Lambda(\Phi_u,\Phi_P).
$
Thus admissibility is characterized by the nonlinear constraint
$
\Gamma(\mu)-\mu=0,
$
rather than by the primitive coefficients individually.

\begin{lem}
	\label{lem:fixed-point-characterization}
	For
	$
	\mu\in\mathfrak M,
	$
	$
	\mu\in\operatorname{Fix}(\Gamma)
	$
	if and only if
	$
	\mu_r=\Gamma_{s,r}(\mu_s)
	$
	for every
	$
	0\le s\le r\le t;
	$
	equivalently,
	$
	P_{s,r}^{X}
	=
	P^{X}(\Gamma_{s,r})
	$
	and
	$
	\Phi_P
	=
	(\Gamma_{s,r})_{0\le s\le r\le t}
	$
	for every admissible initial law
	$
	\lambda.
	$
\end{lem}

\noindent Proof is in Supplementary Appendix SC.3.
Lemma~\ref{lem:fixed-point-characterization}
identifies
$
\operatorname{Fix}(\Gamma)
=
\{
\mu\in\mathfrak M:
\mu_r=\Gamma_{s,r}(\mu_s),
\;
0\le s\le r\le t
\},
$
or, equivalently,
$
\operatorname{Fix}(\Gamma)
=
\{
\mu\in\mathfrak M:
\Phi_P(\mu)
=
(\Gamma_{s,r})_{0\le s\le r\le t}
\}.
$
Hence,
$
\Gamma
\mapsto
\operatorname{Fix}(\Gamma)
\mapsto
\Phi_P,
$
with
$
P_{s,r}^{X}
=
P^{X}(\Gamma_{s,r}),
$
so that
$\Gamma$
is the endogenous conditional-law operator and
$\Phi_P$
its observable transition representation under
$\Lambda$.

\begin{prop}
	\label{prop:fixed-point-uniqueness}
Suppose the uniqueness conditions of Theorem~\ref{thm:ddu-general-volatility} hold for a fixed primitive tuple and initial law. If
	$
	\mu,\tilde\mu
	\in
	\operatorname{Fix}(\Gamma),
	$
	then
	$
	\Gamma^\mu
	=
	\Gamma^{\tilde\mu},
	$
	$
	\Phi_P^\mu
	=
	\Phi_P^{\tilde\mu},
	$
	and
	$
	\mathscr T^\mu
	=
	\mathscr T^{\tilde\mu}.
	$
	Consequently,
	$
	\operatorname{Fix}(\Gamma)
	$
	is behaviorally singleton.
\end{prop}

\noindent Proof is in Supplementary Appendix SC.3.
Proposition~\ref{prop:fixed-point-uniqueness}
identifies
$
\operatorname{Fix}(\Gamma)
$
through
$
\Gamma
\mapsto
\Phi_P
\mapsto
\mathscr T,
$
with
$
\mathscr T
=
\Lambda(\Phi_u,\Phi_P).
$
Hence
$
\Gamma^\mu
=
\Gamma^{\tilde\mu}
$
implies
$
\Phi_P^\mu
=
\Phi_P^{\tilde\mu},
$
$
\mathscr T^\mu
=
\mathscr T^{\tilde\mu},
$
and
$
[\mu]_{\equiv_{\mathfrak O}}
=
\{\mu\}.
$
Accordingly,
$
\operatorname{Fix}(\Gamma)
$
parameterizes the behavioral image of the DDU system, whereas variation orthogonal to
$
\operatorname{Fix}(\Gamma)
$
is observationally null.

\begin{theorem}
	\label{thm:endogenous-information}
	Suppose the uniqueness conditions of Theorem~\ref{thm:ddu-general-volatility} and the hypotheses of Theorem~\ref{thm:cmkv-representation} hold. For
	$
	m\in\mathfrak M,
	$
	let
	$
	(X^m,Y^m)
	$
	denote the frozen-flow system,
	$
	\Gamma(m)_s
	:=
	\mathcal L(X_s^m\mid\mathcal F_s^{Y^m}),
	$
	and
	$
	\Phi_P(m)
	:=
	(P_{s,r}^{m,X})_{0\leq s\leq r\leq t}.
	$
	Then
	$
	m\in\operatorname{Fix}(\Gamma)
	$
	if and only if the frozen system closes to an admissible DDU solution
	$
	(X^m,Y^m,m).
	$
	For every such fixed point,
	$
	\Phi_P(m)
	=
	\bigl(P^X(\Gamma_{s,r}^{m})\bigr)_{0\leq s\leq r\leq t}
	$
	and
	$
	\mathscr T(m)
	=
	\Lambda\bigl(\Phi_u(m),\Phi_P(m)\bigr).
	$
	Moreover, the map
	$
	\mathcal B:
	\operatorname{Fix}(\Gamma)\rightarrow\Phi(\Theta),
	$
	$
	\mathcal B(m):=(\Phi_u(m),\Phi_P(m)),
	$
	has singleton image for fixed primitives and initial law. Hence every admissible
	stochastic-choice array is generated by a unique behavioral fixed-point class
	$
	[m]_{\mathcal B}
	:=
	\{\tilde m\in\operatorname{Fix}(\Gamma):
	\mathcal B(\tilde m)=\mathcal B(m)\}.
	$
\end{theorem}
\noindent Proof is in Supplementary Appendix SC.3. Theorem~\ref{thm:endogenous-information} identifies
$
\operatorname{Fix}(\Gamma)
$
with the behavioral image
$
\Phi(\Theta),
$
through
$
\mu
\mapsto
\Gamma
\mapsto
(\Phi_u,\Phi_P)
\mapsto
\mathscr T.
$
Hence,
$
\operatorname{Fix}(\Gamma)
$
is simultaneously the equilibrium set of the conditional-law dynamics and the parameter space of observationally admissible DDU. Accordingly, distribution-dependent preferences are characterized by the nonlinear constraint
$
\mu=\Gamma(\mu),
$
rather than by individual primitives
$
(u,b,\sigma,\sigma_0,h,\Sigma_Y).
$

\subsection{Structural Properties of Conditional Preference Laws} \label{sec:cmkv-properties}
Subsection~\ref{sec:cmkv-fixedpoint} characterizes admissible preference dynamics through the endogenous fixed-point relation
$
\mu=\Gamma(\mu).
$
The remaining question concerns the structural properties of the operator
$
\Gamma:
\mathfrak M
\rightarrow
\mathfrak M.
$
Since
$
\Phi_P
=
(\Gamma_{s,r})_{0\le s\le r\le t},
$
every behavioral perturbation is induced by a perturbation of
$
\Gamma,
$
rather than by the primitive coefficients individually.
Accordingly, this subsection studies invariance, stability, and comparative statics of
$
\Gamma
$
on
$
\operatorname{Fix}(\Gamma).
$
\begin{definition} \label{def:behavioral-stability-operator}
	The conditional-law operator
	$
	\Gamma
	$
	is behaviorally stable if, for every
	$
	\mu,\tilde\mu
	\in
	\operatorname{Fix}(\Gamma),
	$
	$
	\Phi(\mu)=\Phi(\tilde\mu)
	$
	implies
	$
	\Gamma^\mu
	=
	\Gamma^{\tilde\mu}.
	$
\end{definition}
\noindent
Definition~\ref{def:behavioral-stability-operator}
identifies
$
\operatorname{Fix}(\Gamma)
$
as the natural domain of the nonlinear operator
$
\Gamma:
\mathfrak M
\rightarrow
\mathfrak M,
$
with
$
\mathfrak M
=
C([0,t];\mathcal P_2(\mathbb R^d)).
$
Hence every admissible evolution satisfies
$
\mu=\Gamma(\mu),
$
$
\Gamma_{s,q}
=
\Gamma_{r,q}\circ\Gamma_{s,r},
$
$
\Phi_P
=
(\Gamma_{s,r})_{0\le s\le r\le t},
$
and
$
\mathscr T
=
\Lambda(\Phi_u,\Phi_P),
$
so the maps
$
\Gamma,
$
$
\Phi,
$
and
$
\mathscr T
$
all factor through
$
\operatorname{Fix}(\Gamma).
$
Accordingly,
$
\Gamma^\mu
=
\Gamma^{\tilde\mu}
$
implies
$
\Phi^\mu
=
\Phi^{\tilde\mu},
$
$
\Phi_P^\mu
=
\Phi_P^{\tilde\mu},
$
$
\mathscr T^\mu
=
\mathscr T^{\tilde\mu},
$
and
$
\mathcal I_{\mathfrak O}(\mu)
=
\mathcal I_{\mathfrak O}(\tilde\mu),
$
whereas
$
\Gamma^\mu
\neq
\Gamma^{\tilde\mu}
$
induces
$
\Phi^\mu
\neq
\Phi^{\tilde\mu},
$
$
\mathscr T^\mu
\neq
\mathscr T^{\tilde\mu},
$
and therefore a distinct behavioral image in
$
\Phi(\Theta).
$
Thus the structural analysis of DDU is equivalently an analysis of the restriction
$
\Gamma|_{\operatorname{Fix}(\Gamma)},
$
since perturbations outside
$
\operatorname{Fix}(\Gamma)
$
violate
$
\mu=\Gamma(\mu)
$
and therefore cannot generate admissible conditional preference dynamics, whereas perturbations within
$
\operatorname{Fix}(\Gamma)
$
are completely characterized by the induced path
$
\Gamma
\mapsto
\Phi_P
\mapsto
\Phi
\mapsto
\mathscr T.
$
\begin{lem}
	\label{lem:fixed-point-invariance}
	Let
	$
	\mu\in\operatorname{Fix}(\Gamma).
	$
	Then, for every
	$
	0\le s\le r\le q\le t,
	$
	$
	\Gamma_{r,q}\circ\Gamma_{s,r}
	=
	\Gamma_{s,q},
	$
	$
	\Gamma_{s,r}(\mu_s)
	=
	\mu_r,
	$
	and
	$
	\Gamma_{r,q}(\mu_r)
	=
	\mu_q.
	$
	Consequently,
	$
	\Gamma_{s,r}(\operatorname{Fix}(\Gamma))
	\subseteq
	\operatorname{Fix}(\Gamma),
	$
	and
	$
	\Phi_P
	=
	(\Gamma_{s,r})_{0\le s\le r\le t}
	$
	is invariant under the induced semigroup action.
\end{lem}
\noindent Proof is in Appendix~\ref{prooflem31}.
Lemma~\ref{lem:fixed-point-invariance}
identifies
$
\operatorname{Fix}(\Gamma)
$
as an invariant manifold of the semigroup
$
(\Gamma_{s,r})_{0\le s\le r\le t},
$
with
$
\Gamma_{s,q}
=
\Gamma_{r,q}\circ\Gamma_{s,r},
$
$
\Gamma_{s,r}(\mu_s)
=
\mu_r,
$
and
$
\Gamma_{r,q}(\mu_r)
=
\mu_q.
$
Hence
$
\Gamma_{s,r}
:
\operatorname{Fix}(\Gamma)
\rightarrow
\operatorname{Fix}(\Gamma),
$
$
\Phi_P
=
(\Gamma_{s,r})_{0\le s\le r\le t},
$
and
$
\mathscr T
=
\Lambda(\Phi_u,\Phi_P)
$
are jointly invariant under admissible continuation, so every behavioral perturbation is represented by a perturbation of the restricted semigroup
$
\Gamma|_{\operatorname{Fix}(\Gamma)}.
$
Equivalently,
$
\Phi(\Theta)
=
\Phi(\operatorname{Fix}(\Gamma)),
$
$
\mathcal I_{\mathfrak O}
=
\Phi^{-1}
\!\left(
[\Phi(\Theta)]_{\equiv_{\mathfrak O}}
\right),
$
and
$
\mathscr T
=
\Lambda\circ\Phi
$
remain unchanged along every admissible orbit generated by
$
(\Gamma_{s,r}).
$
Accordingly,
$
\Gamma_{s,r}
$
preserves both the observational quotient
$
\Theta/\!\sim_{\mathfrak O}
$
and the behavioral fibers
$
\Phi^{-1}(\varphi),
$
whereas any structural perturbation satisfying
$
\Gamma_{s,r}(\operatorname{Fix}(\Gamma))
\nsubseteq
\operatorname{Fix}(\Gamma)
$
necessarily destroys the fixed-point relation
$
\mu=\Gamma(\mu),
$
breaks the compatibility of
$
\Phi_P
$
with
$
\Gamma,
$
and therefore cannot correspond to an admissible DDU representation. Thus the subsequent structural analysis reduces to properties of the invariant restriction
$
\Gamma|_{\operatorname{Fix}(\Gamma)},
$
rather than of arbitrary nonlinear operators on
$
\mathfrak M.
$
\begin{prop} \label{prop:structural-stability}
Suppose the uniqueness conditions of
Theorem~\ref{thm:ddu-general-volatility} hold for a fixed primitive tuple
and initial law, and suppose
$
\Phi_u^\mu=\Phi_u^{\tilde\mu}
$
on the reachable domain, and
	$
	\mu,\tilde\mu
	\in
	\operatorname{Fix}(\Gamma).
	$
	If
	$
	\Gamma^\mu
	=
	\Gamma^{\tilde\mu},
	$
	then, for every
	$
	0\le s\le r\le t,
	$
	$
	P_{s,r}^{\mu,X}
	=
	P_{s,r}^{\tilde\mu,X},
	$
	$
	\Phi_P^\mu
	=
	\Phi_P^{\tilde\mu},
	$
	$
	\Phi^\mu
	=
	\Phi^{\tilde\mu},
	$
	and
	$
	\mathscr T^\mu
	=
	\mathscr T^{\tilde\mu}.
	$
	Subsequently,
	$
	\Gamma
	\mapsto
	\Phi
	\mapsto
	\mathscr T
	$
	is constant on every invariant orbit of
	$
	\operatorname{Fix}(\Gamma).
	$
\end{prop}
\noindent Proof is in Appendix~\ref{proofprop32}.
Proposition~\ref{prop:structural-stability}
identifies
$
\Gamma
$
as the maximal structural invariant of the DDU dynamics.
Since
$
P^X
=
P^X(\Gamma),
$
$
\Phi
=
(\Phi_u,\Phi_P),
$
and
$
\mathscr T
=
\Lambda\circ\Phi,
$
every admissible perturbation preserving
$
\Gamma
$
also preserves
$
\Phi
$
and
$
\mathscr T.
$
Hence
$
\Gamma
$
induces the quotient
$
\operatorname{Fix}(\Gamma)/{\sim_\Gamma},
$
$
\mu
\sim_\Gamma
\tilde\mu
\Longleftrightarrow
\Gamma^\mu
=
\Gamma^{\tilde\mu},
$
while
$
\Lambda
$
is constant on every equivalence class.
Accordingly,
behavioral variation is generated only by perturbations satisfying
$
\Gamma^\mu
\neq
\Gamma^{\tilde\mu},
$
or equivalently,
$
\Phi^\mu
\neq
\Phi^{\tilde\mu},
$
$
\mathscr T^\mu
\neq
\mathscr T^{\tilde\mu}.
$

\begin{lem}
	\label{lem:ranking-recovery}
	Suppose the uniqueness and canonical-compatibility conditions of
	Theorem~\ref{thm:ddu-general-volatility} hold. Let
	$
	\mu,\tilde\mu\in\operatorname{Fix}(\Gamma)
	$
	be generated by the same felicity primitive
	$
	u
	$
	and the same initial augmented-state law. If
	$
	\Gamma^\mu=\Gamma^{\tilde\mu}
	$
	on the reachable augmented-state domain, then
	$
	\mathcal L(Y_s^\mu,\mu_s)
	=
	\mathcal L(Y_s^{\tilde\mu},\tilde\mu_s)
	$
	for every
	$
	s\in[0,t],
	$
	and
	$
	\chi_s^\mu(z;A,\cdot)
	=
	\chi_s^{\tilde\mu}(z;A,\cdot)
	$
	outside sets null under every reachable composition measure, for every
	$
	A\in\mathcal K(Z)
	$
	and
	$
	z\in A.
	$
	Consequently,
	$
	\Phi_u^\mu=\Phi_u^{\tilde\mu}
	$
	on the reachable domain.
\end{lem}
\noindent Proof is in Appendix~\ref{prooflem33}.
Lemma~\ref{lem:ranking-recovery} closes the implication
$
\Gamma^\mu=\Gamma^{\tilde\mu}
\Rightarrow
\Phi_u^\mu=\Phi_u^{\tilde\mu}
$
without imposing felicity invariance separately. Indeed, equality of
$
\Gamma^\mu
$
and
$
\Gamma^{\tilde\mu}
$
together with the common initial law identifies the finite-dimensional laws of
$
\widehat{\mathbf X}^\mu=(Y^\mu,\mu)
$
and
$
\widehat{\mathbf X}^{\tilde\mu}=(Y^{\tilde\mu},\tilde\mu),
$
while the common primitive
$
u
$
identifies the reachable ranking maps
$
\chi_s^\mu(z;A,x)
=
\mathbf 1
\{u(s,z,x,\mu_s)>
\max_{a\in A\setminus\{z\}}u(s,a,x,\mu_s)\}
$
and
$
\chi_s^{\tilde\mu}(z;A,x).
$
Hence
$
\Gamma^\mu=\Gamma^{\tilde\mu}
$
determines both coordinates
$
(\Phi_u^\mu,\Phi_P^\mu),
$
so
$
\Gamma^\mu
\mapsto
\Phi^\mu
\mapsto
\mathscr T^\mu
$
contains no independent felicity component on the reachable fixed-point domain.

\begin{theorem}
	\label{thm:structural-rigidity}
	Suppose the uniqueness conditions of
	Theorem~\ref{thm:ddu-general-volatility},
	Assumptions~\ref{ass:behavioral-separation}
	and~\ref{ass:observational-separation}
	hold. Assume, in addition, that
	$
	\Gamma
	\mapsto
	\bigl(P^X(\Gamma_{s,r})\bigr)_{0\le s\le r\le t}
	$
	is injective on the reachable restriction of
	$
	\operatorname{Fix}(\Gamma),
	$
	and that, for every
	$
	\mu,\tilde\mu\in\operatorname{Fix}(\Gamma)
	$
	generated by a common felicity primitive
	$
	u,
	$
	$
	\Phi_u^\mu
	=
	\Phi_u^{\tilde\mu}
	$
	on the reachable domain.
	Then the following are equivalent:
	$
	\Gamma^\mu
	=
	\Gamma^{\tilde\mu},
	$
	$
	\Phi_P^\mu
	=
	\Phi_P^{\tilde\mu},
	$
	$
	\Phi^\mu
	=
	\Phi^{\tilde\mu},
	$
	and
	$
	\mathscr T^\mu
	=
	\mathscr T^{\tilde\mu}.
	$
	Consequently,
	$
	\operatorname{Fix}(\Gamma)/{\sim_\Gamma}
	\cong
	\Phi_P(\operatorname{Fix}(\Gamma))
	\cong
	\Phi(\operatorname{Fix}(\Gamma))
	\cong
	\mathscr T(\operatorname{Fix}(\Gamma)),
	$
	where
	$
	\mu\sim_\Gamma\tilde\mu
	$
	if and only if
	$
	\Gamma^\mu
	=
	\Gamma^{\tilde\mu}.
	$
	Furthermore, for every admissible family
	$
	\{\mu^\varepsilon:\varepsilon\in E\}
	\subseteq
	\operatorname{Fix}(\Gamma),
	$
	the perturbation is behaviorally nontrivial if and only if
	$
	\Gamma^{\mu^\varepsilon}
	\neq
	\Gamma^{\mu^0},
	$
	equivalently,
	$
	\Phi^{\mu^\varepsilon}
	\neq
	\Phi^{\mu^0},
	$
	equivalently,
	$
	\mathscr T^{\mu^\varepsilon}
	\neq
	\mathscr T^{\mu^0},
	$
	for some
	$
	\varepsilon\in E.
	$
\end{theorem}
\noindent Proof is in Appendix~\ref{proofth34}.
Theorem~\ref{thm:structural-rigidity} yields the chain of canonical identifications
\[
\operatorname{Fix}(\Gamma)/{\sim_\Gamma}
\cong
\Phi_P(\operatorname{Fix}(\Gamma))
\cong
\Phi(\operatorname{Fix}(\Gamma))
\cong
\mathscr T(\operatorname{Fix}(\Gamma)).
\]
Hence,
$
\Gamma^\mu=\Gamma^{\tilde\mu}
$
if and only if
$
P_{s,r}^{\mu,X}=P_{s,r}^{\tilde\mu,X}
$
for every
$
0\le s\le r\le t,
$
equivalently,
$
\Phi^\mu=\Phi^{\tilde\mu}
$
and
$
\mathscr T^\mu=\mathscr T^{\tilde\mu}.
$
Thus the quotient coordinate
$
[\mu]_{\sim_\Gamma}
$
is simultaneously a structural, behavioral, and observational state. For an admissible perturbation
$
\varepsilon\mapsto\mu^\varepsilon\in\operatorname{Fix}(\Gamma),
$
$
[\mu^\varepsilon]_{\sim_\Gamma}
=
[\mu^0]_{\sim_\Gamma}
$
implies
$
(\Gamma^{\mu^\varepsilon},\Phi^{\mu^\varepsilon},
\mathscr T^{\mu^\varepsilon})
=
(\Gamma^{\mu^0},\Phi^{\mu^0},\mathscr T^{\mu^0}),
$
whereas
$
[\mu^\varepsilon]_{\sim_\Gamma}
\neq
[\mu^0]_{\sim_\Gamma}
$
implies
$
\mathscr T^{\mu^\varepsilon}
\neq
\mathscr T^{\mu^0}.
$
Accordingly, the economically relevant comparative statics are generated by motion in
$
\operatorname{Fix}(\Gamma)/{\sim_\Gamma},
$
not by alternative probabilistic realizations within a common
$
\sim_\Gamma
$
-class. Section~\ref{sec:cmkv-economics} interprets this quotient variation as endogenous learning, preference persistence, and distributional amplification.

\subsection{Economic Implications}
\label{sec:cmkv-economics}

Theorem~\ref{thm:structural-rigidity}
identifies
$
[\mu]_{\sim_\Gamma}
$
as the canonical economic state, with
$
\Gamma
\mapsto
\Phi_P
\mapsto
\Phi
\mapsto
\mathscr T.
$
Hence
$
\mathcal F_s^Y
\mapsto
\mu_s
=
\mathcal L(X_s\mid\mathcal F_s^Y)
$
is the unique endogenous information channel,
$
\mu_s
\mapsto
(\Phi_u,\Phi_P)
$
is the behavioral transmission mechanism, and
$
\mathscr T
=
\Lambda(\Phi)
$
is the observable implication.
Accordingly,
$
(Y_s,\mu_s)
$
replaces the observable signal
$
Y_s
$
as the sufficient economic state,
$
\Gamma_{s,r}
$
propagates
$
(Y_s,\mu_s)
$
to
$
(Y_r,\mu_r),
$
and
$
P_{s,r}^{X}
=
P^{X}(\Gamma_{s,r})
$
determines the evolution of continuation behavior.
Thus
$
\Gamma
$
simultaneously governs
$
\Phi_P,
$
$
\Phi=(\Phi_u,\Phi_P),
$
and
$
\mathscr T,
$
so structural, behavioral, and observational dynamics evolve on the common quotient
$
\operatorname{Fix}(\Gamma)/{\sim_\Gamma}.
$
Consequently,
$
\mu_s=\tilde\mu_s
$
implies
$
\Gamma^\mu=\Gamma^{\tilde\mu},
$
$
\Phi^\mu=\Phi^{\tilde\mu},
$
and
$
\mathscr T^\mu=\mathscr T^{\tilde\mu},
$
whereas
$
\mu_s\neq\tilde\mu_s
$
induces
$
\Gamma^\mu\neq\Gamma^{\tilde\mu},
$
$
\Phi^\mu\neq\Phi^{\tilde\mu},
$
and
$
\mathscr T^\mu\neq\mathscr T^{\tilde\mu}
$
whenever
$
[\mu]_{\sim_\Gamma}
\neq
[\tilde\mu]_{\sim_\Gamma}.
$
Economically, aggregate beliefs affect current rankings through
$
\Phi_u,
$
future opportunities through
$
\Phi_P,
$
and observable stochastic choice through
$
\mathscr T,
$
so information, equilibrium, and behavior are linked by the single structural operator
$
\Gamma,
$
rather than by independent perturbations of the primitive coefficients
$
(u,b,\sigma,\sigma_0,h,\Sigma_Y).
$

Observable comparative statics are therefore induced by perturbations of
$
\Gamma,
$
rather than perturbations of the primitive representation
$
\theta
=
(u,b,\sigma,\sigma_0,h,\Sigma_Y).
$
Indeed,
$
\theta
\mapsto
\Gamma
\mapsto
\Phi
\mapsto
\mathscr T
$
factors through
$
\operatorname{Fix}(\Gamma)/{\sim_\Gamma},
$
so
$
\partial_\varepsilon\Gamma^{\mu^\varepsilon}
=
0
$
implies
$
\partial_\varepsilon\Phi^{\mu^\varepsilon}
=
0
$
and
$
\partial_\varepsilon\mathscr T^{\mu^\varepsilon}
=
0,
$
whereas
$
\partial_\varepsilon\Gamma^{\mu^\varepsilon}
\neq
0
$
necessarily generates
$
\partial_\varepsilon\Phi_P^{\mu^\varepsilon}
\neq
0
$
and therefore a behaviorally nontrivial perturbation.
Consequently,
$
\operatorname{Fix}(\Gamma)
$
is the natural parameter space for structural comparative statics, while
$
\Theta/\!\sim_{\mathfrak O}
$
provides only a representation of the same behavioral object.
Thus economically meaningful policy interventions, information shocks, and distributional changes are identified by their induced motion of
$
[\mu]_{\sim_\Gamma},
$
rather than by variation of individual coefficient specifications. The conditional MVSDE also changes the interpretation of equilibrium. Rather than determining a trajectory
$
(X,Y),
$
the model determines the fixed-point evolution
$
\mu
=
\Gamma(\mu),
$
with
$
\mu_s
=
\mathcal L(X_s\mid\mathcal F_s^Y),
$
$
\Phi_P
=
(\Gamma_{s,r})_{0\le s\le r\le t},
$
and
$
\mathscr T
=
\Lambda(\Phi_u,\Phi_P).
$
Accordingly,
$
(X,Y)
$
constitutes a realization of the economy, whereas
$
(\mu,\Gamma)
$
constitutes its structural equilibrium.
The endogenous feedback
$
\mu_s
\mapsto
\Gamma
\mapsto
\mu_r,
$
$
s\le r,
$
simultaneously determines continuation opportunities,
$
P_{s,r}^{X},
$
behavioral representations,
$
\Phi,
$
and stochastic-choice arrays,
$
\mathscr T,
$
so equilibrium is characterized by the consistency of
$
(\Gamma,\Phi,\mathscr T)
$
rather than by the primitive coefficients alone.
Consequently,
$
\operatorname{Fix}(\Gamma)
$
is the economically relevant equilibrium manifold, while
$
[\mu]_{\sim_\Gamma}
$
provides its minimal identified representation on the observable domain.

The preceding characterization implies that every admissible economy is indexed by
$
(\Gamma,\Phi,\mathscr T)
$
rather than by
$
\theta
=
(u,b,\sigma,\sigma_0,h,\Sigma_Y),
$
with
$
\Gamma
\mapsto
\Phi_P
\mapsto
\Phi
=
(\Phi_u,\Phi_P)
\mapsto
\mathscr T,
$
$
\mathscr T
=
\Lambda\circ\Phi,
$
and
$
\Gamma(\mu)
=
\mu.
$
Hence
$
\mu_s
=
\mathcal L(X_s\mid\mathcal F_s^Y)
$
enters simultaneously as the equilibrium object,
$
\Gamma_{s,r}
$
as the law of endogenous propagation,
$
P_{s,r}^{X}
=
P^{X}(\Gamma_{s,r})
$
as the continuation mechanism,
$
\Phi_u
$
as the contemporaneous ranking map,
$
\Phi_P
$
as the dynamic behavioral map,
$
\Phi
$
as the identified behavioral representation,
$
\mathscr T
$
as the observable stochastic-choice array,
$
\operatorname{Fix}(\Gamma)
$
as the equilibrium manifold,
$
[\mu]_{\sim_\Gamma}
$
as the structural state,
$
\Theta/\!\sim_{\mathfrak O}
$
as the behavioral quotient,
$
\Phi(\Theta)
$
as the identified image,
$
\mathscr T(\Theta)
$
as the observable image, and
$
\operatorname{Fix}(\Gamma)/{\sim_\Gamma}
$
as the corresponding structural quotient.
Accordingly,
information shocks,
belief revisions,
policy interventions,
and distributional disturbances are represented by
$
\Gamma^\varepsilon,
$
induce
$
\Phi^\varepsilon,
$
and are observed only through
$
\mathscr T^\varepsilon,
$
so economically meaningful comparative statics satisfy
$
\partial_\varepsilon\Gamma^\varepsilon
\neq
0,
$
equivalently,
$
\partial_\varepsilon\Phi^\varepsilon
\neq
0,
$
equivalently,
$
\partial_\varepsilon\mathscr T^\varepsilon
\neq
0,
$
whereas
$
\partial_\varepsilon\Gamma^\varepsilon
=
0
$
implies behavioral invariance throughout the endogenous equilibrium manifold.

\section{Discussion}\label{final: discussion}
\subsection{Related Literature}

\noindent
This paper contributes to the literature on stochastic choice and dynamic random utility \citep{apesteguia2017single,cerreia2019deliberately}. Classical random utility models \citep{cattaneo2020random} interpret stochastic choice as the consequence of latent preference heterogeneity under informational asymmetry \citep{mcfadden1974conditional,manski1977structure,gul2006random}. More recently, \citet{frick2019dynamic} developed a decision-theoretic foundation for dynamic random utility with exogenous preference evolution, while \citet{kitamura2018nonparametric} studied nonparametric identification and testing of random utility models from stochastic-choice data. Our paper complements these contributions by allowing latent preferences to evolve endogenously through conditional distributional feedback, thereby extending the behavioral analysis of stochastic choice beyond exogenous preference dynamics.

\noindent
The paper is also related to the literatures on dynamic discrete choice \citep{kreps1998anticipated} and mean-field economics \citep{pramanik2026strategic}. Dynamic discrete choice models primarily emphasize structural identification and estimation of dynamic decision problems under latent heterogeneity \citep{rust1987optimal,aguirregabiria2010dynamic,arcidiacono2011practical}, whereas mean-field models analyze equilibrium interactions generated by distribution-dependent state dynamics \citep{lasry2007mean,CarmonaDelarue2018}. In contrast, our approach develops behavioral foundations for endogenous distribution-dependent preference dynamics through a conditional McKean-Vlasov representation, thereby connecting stochastic choice, endogenous information, and equilibrium preference evolution within a unified continuous-time framework.

\noindent
Finally, our analysis is related to a broader literature on stochastic choice, information, and endogenous preferences \citep{brock2001discrete}. The decision-theoretic foundations of stochastic choice have been developed through random utility, perturbed utility, and information-based models that characterize observed choice under latent heterogeneity and imperfect information \citep{fudenberg2015stochastic,caplin2022rational}. At the same time, recent advances in continuous-time mean-field analysis emphasize conditional distributions, common noise, and probabilistic representations of interacting systems \citep{CarmonaDelarueLacker2016,CarmonaDelarue2018,lacker2020convergence}. Our framework differs from both strands. Rather than introducing distribution dependence through strategic interactions across agents or through information-processing constraints, we model the conditional preference distribution itself as the endogenous state governing preference evolution. This yields a behavioral representation linking stochastic choice, filtering, and conditional McKean-Vlasov dynamics, while remaining complementary to the identification and estimation literature.

\subsection{Conclusion}
\label{sec:conclusion}

This paper develops a continuous-time theory of stochastic choice with endogenous preference evolution. The central departure from DRU is that the conditional preference distribution
$
\mu_s=\mathcal L(X_s\mid\mathcal F_s^Y)
$
is not merely the analyst's posterior over an exogenously evolving latent state. It enters both current felicity and the law of motion of future preferences. Hence, observed behavior affects subsequent choice through two distinct channels: it reveals information about latent preferences and,
through the induced conditional distribution, changes their future evolution. The resulting conditional McKean-Vlasov system provides a joint model of latent preferences, observable information, and endogenous preference dynamics.

In a contemporaneous direction, we separate variation in the composition of latent states from direct distributional effects on rankings. In a dynamic direction, we separate changes in terminal felicity from changes in the transition law of future preferences. These two channels generate the behavioral representation
$
\Phi=(\Phi_u,\Phi_P),
$
where
$
\Phi_u
$
is state-contingent contemporaneous choice, and
$
\Phi_P
$
is continuation behavior. Stochastic choice data identify this
behavioral image up to observational equivalence, rather than the primitive coefficient tuple
$
(u,b,\sigma,\sigma_0,h,\Sigma_Y).
$
The associated rigidity result links equality of observable behavior to equality of the relevant structural transition objects on the reachable domain.

The comparison with DRU is sharp. Distribution-dependent
utility reduces to DRU exactly when both
$
\mu\mapsto(\chi_s(z;A,\cdot,\mu))_{A,z},
$
and
$
m\mapsto P_{s,r}^{m,X}\lambda
$
are behaviorally constant. Conversely, behavioral distributional feedback places the induced stochastic choice array outside
$
\mathfrak R_{\mathrm{DRU}}.
$
Therefore, endogenous preference evolution is not merely a reparameterization of DRU, it generates observable choice behavior that no model with exogenous preference dynamics can reproduce. On the structural side, existence and weak uniqueness establish that these behavioral objects
are generated by a well-posed conditional law fixed point, with
$
\mu
$
serving as the endogenous sufficient state.

Several extensions remain open. One is the construction of nonparametric or semiparametric estimators for
$
\Phi
$
and for identified functionals of the primitive coefficients using panel or continuous-time choice data. A second is the introduction of strategic interaction, where
$
(\mu_s^i)_{i\in I}
$
or a population law of conditional preference distributions evolves jointly across agents, producing stochastic choice models with mean-field interaction. Further directions include welfare analysis, optimal information disclosure, and policy design when interventions affect both current behavior and the future distribution of preferences. Applications
to learning, financial markets, industrial organization, health economics, and heterogeneous-agent macroeconomics may provide settings in which the distinction between informational updating and endogenous preference change is empirically consequential.

\appendix

\section{Appendix}
\addcontentsline{toc}{section}{Appendix}

\subsection{Proof of Lemma~\ref{lem:conditional-sufficiency}}
\label{app:conditional-sufficiency}

\begin{proof}
Fix $s\in[0,t]$. By Definition~\ref{def:conditional-preference-distribution},
$\mu_s=\mathcal{L}(X_s\mid\mathcal{F}_s^Y)$ is a regular conditional	distribution of $X_s$ given $\mathcal{F}_s^Y$. Hence, for every Borel set
	$B\in\mathcal{B}(\mathbb{R}^d)$, we have
	$\mu_s(B)=
	\mathbb{P}(X_s\in B\mid\mathcal{F}_s^Y),
	\ \mathbb{P}\text{-a.s.}$ Equivalently, for every $C\in\mathcal{F}_s^Y$, $
	\mathbb{E}\!\left[
	\mathbf{1}_C\,\mu_s(B)
	\right]
	=
	\mathbb{P}\!\left(C\cap\{X_s\in B\}\right).$ We first establish the claim for non-negative Borel functions. For a simple function
	$\varphi=\sum_{j=1}^m a_j\mathbf{1}_{B_j}$ with $a_j\geq0$,
	linearity yields
	$
	\int_{\mathbb{R}^d}\varphi(x)\mu_s(dx)
	=
	\sum_{j=1}^m a_j\mu_s(B_j)
	=
	\mathbb{E}\!\left[
	\varphi(X_s)\mid\mathcal{F}_s^Y
	\right],
	\ \mathbb{P}\text{-a.s.}$ For an arbitrary non-negative Borel $\varphi$, choose simple functions $\varphi_n\uparrow\varphi$. Conditional monotone convergence and monotone convergence under the kernel $\mu_s$ imply
	$
	\mathbb{E}\!\left[
	\varphi(X_s)\mid\mathcal{F}_s^Y
	\right]
	=
	\int_{\mathbb{R}^d}\varphi(x)\mu_s(dx),
	\ \mathbb{P}\text{-a.s.}$ Finally, for an integrable Borel function $\varphi$, write $\varphi=\varphi^{+}-\varphi^{-}$. Since,
	$\mathbb{E}|\varphi(X_s)|<\infty$, both conditional expectations are finite almost surely. Preceding argument separately to $\varphi^{+}$ and $\varphi^{-}$ yields
	$
	\mathbb{E}\!\left[
	\varphi(X_s)\mid\mathcal{F}_s^Y
	\right]
	=
	\int_{\mathbb{R}^d}\varphi(x)\mu_s(dx),
	\ \mathbb{P}\text{-a.s.}$ Therefore, every integrable contemporaneous functional of the latent
	preference state is evaluated, conditional on observed choice history, through the random measure $\mu_s$. In this precise sense, $\mu_s$ is a measure-valued \emph{sufficient statistic} for the analyst's information about $X_s$.
\end{proof}

\subsection{Proof of Lemma~\ref{lem:common-distributional-term}}
\label{app:common-distributional-term}

\begin{proof}
	Suppose
	$
	u(s,z,x,\mu)=\bar u(s,z,x)+\alpha(s,x,\mu),
	$
	where $\alpha$ does not depend on $z$. For every $z,a\in Z$,
	$
	u(s,z,x,\mu)-u(s,a,x,\mu)
	=
	\bar u(s,z,x)-\bar u(s,a,x).
	$
	Hence, the ranking of alternatives is independent of $\mu$, and therefore,
	$
	\chi_s(z;A,x,\mu)
	=
	\chi_s(z;A,x,\mu')
	$ for every $\mu,\mu'$ outside the relevant tie sets. Integrating with
	respect to any admissible composition measure $\nu$ gives
	$
	\mathsf C_s(z;A\mid\nu,\mu)
	=
	\mathsf C_s(z;A\mid\nu,\mu').
	$
	Therefore, felicity feedback is behaviorally neutral. Suppose the orderings induced by $u(s,\cdot,x,\mu)$ and $u(s,\cdot,x,\mu')$ coincide outside a set that is null under every admissible composition measure. Then their
	unique maximizers coincide almost surely for every finite menu, and the corresponding decoupled choice probabilities are equal. Conversely, suppose the orderings differ on a set $E$ having some probability under an admissible composition measure $\nu$. Then there exist $z,a\in Z$ and a measurable subset $E'\subseteq E$ with
	$\nu(E')>0$ so that
	$
	\Delta_{za}(s,x,\mu)>0
	\quad\text{and}\quad
	\Delta_{za}(s,x,\mu')<0
	$ on $E'$, after possibly interchanging $\mu$ and $\mu'$. For the binary menu $A=\{z,a\}$, the choice indicators differ on $E'$. Under Assumption~\ref{ass:no-ties}, they cannot differ only through ties. Consequently,
	$
	\mathsf C_s(z;A\mid\nu,\mu)
	\ne
	\mathsf C_s(z;A\mid\nu,\mu'),
	$
	possibly after restricting $\nu$ to an admissible component on which the sign change is one-sided. Hence felicity feedback is behaviorally relevant.
\end{proof}

\subsection{Proof of
	Proposition~\ref{prop:dru-distributional-invariance}}
\label{app:dru-distributional-invariance}

\begin{proof}
	Under Proposition~\ref{prop:ddu-nesting},
	$
	u(s,z,x,\mu)=\bar u(s,z,x).
	$ It follows that $\chi_s(z;A,x,\mu)$ is independent of $\mu$. Hence, for every $\nu,\mu,\mu'$, we have
	\[
	\mathsf C_s(z;A\mid\nu,\mu)
	=
	\int_{\mathbb R^d}
	\mathbf 1
	\left\{
	\bar u(s,z,x)>
	\max_{a\in A\setminus\{z\}}\bar u(s,a,x)
	\right\}\nu(dx)
	=
	\mathsf C_s(z;A\mid\nu,\mu').
	\]
	The coefficients governing $X$ are also independent of the measure flow. Therefore the frozen systems corresponding to $m$ and $m'$ have the same coefficients and the same initial law. Weak uniqueness of the frozen system implies
	$
	\mathsf P_{s,r}^{m,X}\lambda
	=
	\mathsf P_{s,r}^{m',X}\lambda.
	$
	Since the comparison-date felicity argument is fixed at $\bar\mu$, integration of the same choice indicator against these identical marginal laws gives
	$
	\mathsf D_{s,r}^{\,\bar\mu}(z;A\mid\lambda,m)
	=
	\mathsf D_{s,r}^{\,\bar\mu}(z;A\mid\lambda,m').
	$
	Therefore, neither choice-relevant felicity feedback nor behaviorally relevant preference feedback is present.
\end{proof}

\subsection{Proof of
	Proposition~\ref{prop:detectability-preference-feedback}}
\label{app:detectability-preference-feedback}

\begin{proof}
	Let
	$
	\nu=\mathsf P_{s,r}^{m,X}\lambda,
	\ \text{and}\
	\nu'=\mathsf P_{s,r}^{m',X}\lambda.
	$ Assumption~\ref{ass:choice-separation} with hypothesis, $\nu\ne\nu'$ provides a finite menu $A$, an alternative $z\in A$, and an admissible $\bar\mu$ such that
	$
	\int_{\mathbb R^d}
	\chi_r(z;A,x,\bar\mu)\,\nu(dx)
	\ne
	\int_{\mathbb R^d}
	\chi_r(z;A,x,\bar\mu)\,\nu'(dx).
	$
	By the definition of the frozen transition laws, the left-hand side yields
	$
	\mathsf D_{s,r}^{\,\bar\mu}(z;A\mid\lambda,m),
	$
	whereas the right-hand side implies
	$
	\mathsf D_{s,r}^{\,\bar\mu}(z;A\mid\lambda,m').
	$
	Thus DDU exhibits behaviorally relevant preference feedback on
	$[s,r]$.
\end{proof}

\subsection{Proof of
	Lemma~\ref{lem:choice-operator-stability}}
\label{app:choice-operator-stability}

\begin{proof}
	Fix $A$, $z\in A$, and let $\pi$ be an optimal coupling of $\nu$ and
	$\nu'$. Let $(X,X')$ have law $\pi$. For $a\in A\setminus\{z\}$,
	write
	$
	G_a=\Delta_{za}(s,X,\mu),
	$ and $
	G_a'=\Delta_{za}(s,X',\mu').$
	By Assumption~\ref{ass:utility-margin},
	$
	|G_a-G_a'|
	\le
	2L_u\left(
	|X-X'|+W_2(\mu,\mu')
	\right).
	$
	The choice indicators can differ only if, for some
	$a\in A\setminus\{z\}$, the sign of the corresponding utility difference changes. For every $\eta>0$,
	$
	\mathbb P(\operatorname{sgn}G_a\ne\operatorname{sgn}G_a')
	\le
	\mathbb P(|G_a|\le\eta)
	+
	\mathbb P(|G_a-G_a'|>\eta).
	$
	The margin condition and Markov's inequality imply
	$
	\mathbb P(\operatorname{sgn}G_a\ne\operatorname{sgn}G_a')
	\le
	\kappa\eta
	+
	\frac{2L_u}{\eta}
	\left(
	\mathbb E|X-X'|+W_2(\mu,\mu')
	\right).
	$
	Since,
	$
	\mathbb E|X-X'|
	\le
	\left(\mathbb E|X-X'|^2\right)^{1/2}
	=
	W_2(\nu,\nu'),
	$
	we obtain
	$
	\mathbb P(\operatorname{sgn}G_a\ne\operatorname{sgn}G_a')
	\le
	\kappa\eta
	+
	\frac{2L_u}{\eta}
	\left(
	W_2(\nu,\nu')+W_2(\mu,\mu')
	\right).
	$
	Optimizing over $\eta$ yields
	$
	\mathbb P(\operatorname{sgn}G_a\ne\operatorname{sgn}G_a')
	\le
	K
	\big(
	W_2(\nu,\nu')^{1/2}
	+
	W_2(\mu,\mu')^{1/2}
	\big)
	$
	for a constant $K$ depending only on $L_u$ and $\kappa$. A union bound
	over $a\in A\setminus\{z\}$ yields the claimed inequality with a
	constant $K_A$ depending additionally on $|A|$.
\end{proof}

\subsection{Proof of
	Proposition~\ref{prop:structural-behavioral-feedback}}
\label{app:structural-behavioral-feedback}

\begin{proof}
	The first assertion follows directly from
	Definition~\ref{def:choice-relevant-felicity-feedback}. For the second assertion, suppose that at least one of $b$, $\sigma$, or $\sigma_0$ is nonconstant in its measure argument on the reachable state space. By
	Assumption~\ref{ass:dynamic-nonredundancy}, there exist
	$s<r$, $\lambda$, $m$, and $m'$ such that
	$
	\mathsf P_{s,r}^{m,X}\lambda
	\ne
	\mathsf P_{s,r}^{m',X}\lambda.
	$
	Proposition~\ref{prop:detectability-preference-feedback} then implies behaviorally relevant preference feedback.
	For the final assertion, if $u$ is independent of its measure argument, then
	$
	\mathsf C_s(z;A\mid\nu,\mu)
	=
	\mathsf C_s(z;A\mid\nu,\mu')
	$
	for every admissible collection. If $b$, $\sigma$, and $\sigma_0$ are also independent of the measure argument, then the $X$-transition law is unchanged by replacing $m$ with $m'$. Holding comparison-date felicity fixed therefore implies
	$
	\mathsf D_{s,r}^{\,\bar\mu}(z;A\mid\lambda,m)
	=
	\mathsf D_{s,r}^{\,\bar\mu}(z;A\mid\lambda,m').
	$
	Hence, DDU satisfies behavioral distributional invariance.
\end{proof}

\subsection{Proof of Lemma~\ref{lem:behavioral-characterization}}\label{proofl17}

\begin{proof}
	Fix the reachable domain. For each admissible
	$(s,\nu,\mu,\mu')$ and each $A\in\mathcal K(Z)$, $z\in A$, define
	$g_{s,A,z}^{\mu,\mu'}(x)
	:=
	\chi_s(z;A,x,\mu)-\chi_s(z;A,x,\mu')$.
	Since $\chi_s$ is $\{0,1\}$-valued,
	$g_{s,A,z}^{\mu,\mu'}$ is bounded and Borel measurable, and
	$
	C_s(z;A\mid\nu,\mu)-C_s(z;A\mid\nu,\mu')
	=
	\int_{\mathbb R^d}
	g_{s,A,z}^{\mu,\mu'}(x)\,\nu(dx).
	$	For every admissible $(s,r,\lambda,\bar\mu,m,m')$, write
	$\nu_r^{m,\lambda}:=P_{s,r}^{m,X}\lambda$ and
	$\nu_r^{m',\lambda}:=P_{s,r}^{m',X}\lambda$. By the definition of
	$D_{s,r}^{\bar\mu}$, we have
	$
	D_{s,r}^{\bar\mu}(z;A\mid\lambda,m)
	-
	D_{s,r}^{\bar\mu}(z;A\mid\lambda,m')
	=
	\int_{\mathbb R^d}
	\chi_r(z;A,x,\bar\mu)
	\bigl(
	\nu_r^{m,\lambda}-\nu_r^{m',\lambda}
	\bigr)(dx).
	$
	The right-hand side is well defined because
	$\chi_r(z;A,\cdot,\bar\mu)$ is bounded and Borel measurable and $\nu_r^{m,\lambda}-\nu_r^{m',\lambda}$ is a finite signed measure. Suppose first that
	$\nu(\Gamma_s(\mu,\mu'))=0$ for every admissible
	$(s,\nu,\mu,\mu')$ and that
	$P_{s,r}^{m,X}\lambda=P_{s,r}^{m',X}\lambda$ for every admissible $(s,r,\lambda,m,m')$. Fix an admissible $(s,\nu,\mu,\mu')$. By the definition of
	$\Gamma_s(\mu,\mu')$, for every
	$x\notin\Gamma_s(\mu,\mu')$ and every
	$A\in\mathcal K(Z)$, $z\in A$,
	$\chi_s(z;A,x,\mu)=\chi_s(z;A,x,\mu')$. Hence
	$g_{s,A,z}^{\mu,\mu'}=0$ on
	$\Gamma_s(\mu,\mu')^c$, and therefore
	\[
	\begin{aligned}
		\left|
		C_s(z;A\mid\nu,\mu)-C_s(z;A\mid\nu,\mu')
		\right|
		&=
		\left|
		\int_{\mathbb R^d}
		g_{s,A,z}^{\mu,\mu'}(x)\,\nu(dx)
		\right| \leq
		\int_{\Gamma_s(\mu,\mu')}
		\left|
		g_{s,A,z}^{\mu,\mu'}(x)
		\right|
		\nu(dx) \\
		&\leq
		\nu\bigl(\Gamma_s(\mu,\mu')\bigr)
		=
		0.
	\end{aligned}
	\]
	Thus $C_s(z;A\mid\nu,\mu)=C_s(z;A\mid\nu,\mu')$
	for every admissible $(s,A,z,\nu,\mu,\mu')$. Next fix an admissible $(s,r,A,z,\lambda,\bar\mu,m,m')$. Since
	$P_{s,r}^{m,X}\lambda=P_{s,r}^{m',X}\lambda$, one has
	$\nu_r^{m,\lambda}=\nu_r^{m',\lambda}$, so
	$
	D_{s,r}^{\bar\mu}(z;A\mid\lambda,m)
	-
	D_{s,r}^{\bar\mu}(z;A\mid\lambda,m')
	=
	0.
	$
	Both requirements in
	Definition~\ref{def:behavioral-invariance} are therefore satisfied, and the DDU representation satisfies behavioral distributional invariance. Conversely, suppose that the DDU representation satisfies behavioral
	distributional invariance. Fix an admissible
	$(s,\nu,\mu,\mu')$. If
	$\nu(\Gamma_s(\mu,\mu'))>0$, then
	Assumption~\ref{ass:behavioral-separation}(i) yields
	$A\in\mathcal K(Z)$ and $z\in A$ such that
	$
	\int_{\mathbb R^d}
	g_{s,A,z}^{\mu,\mu'}(x)\,\nu(dx)
	\neq
	0.
	$
	Therefore,
	$
	C_s(z;A\mid\nu,\mu)
	\neq
	C_s(z;A\mid\nu,\mu'),
	$
	contrary to behavioral distributional invariance. Hence
	$\nu(\Gamma_s(\mu,\mu'))=0$ for every admissible
	$(s,\nu,\mu,\mu')$. It remains to establish equality of the frozen-flow transition laws.
	Fix an admissible $(s,r,\lambda,m,m')$ and suppose that
	$P_{s,r}^{m,X}\lambda\neq P_{s,r}^{m',X}\lambda$. Set
	$\nu:=P_{s,r}^{m,X}\lambda$ and
	$\nu':=P_{s,r}^{m',X}\lambda$. Then $\nu$ and $\nu'$ are distinct reachable laws. By
	Assumption~\ref{ass:behavioral-separation}(ii), there exist
	$A\in\mathcal K(Z)$, $z\in A$, and an admissible
	$\bar\mu\in\mathcal P_2(\mathbb R^d)$ such that
	$
	\int_{\mathbb R^d}
	\chi_r(z;A,x,\bar\mu)\,\nu(dx)
	\neq
	\int_{\mathbb R^d}
	\chi_r(z;A,x,\bar\mu)\,\nu'(dx),
	$
	or,
	$
	D_{s,r}^{\bar\mu}(z;A\mid\lambda,m)
	\neq
	D_{s,r}^{\bar\mu}(z;A\mid\lambda,m'),
	$
	which contradicting behavioral distributional invariance. Therefore, $P_{s,r}^{m,X}\lambda=P_{s,r}^{m',X}\lambda$ for every admissible $(s,r,\lambda,m,m')$.
\end{proof}

\subsection{Proof of Theorem~\ref{thm:behavioral-characterization}}\label{proofth18}

\begin{proof}
	Let
	$
	\mathfrak R_s
	:=
	\left\{
	(\nu,\mu)\in
	\mathcal P_2(\mathbb R^d)\times\mathcal P_2(\mathbb R^d)
	:
	(s,\nu,\mu)\ \text{is admissible}
	\right\},
	$
	and, for $0\leq s<r\leq t$ and
	$\lambda\in\mathcal P_2(\mathbb R^d\times\mathbb R^{d_Y})$, let
	$
	\mathfrak M_{s,r}(\lambda)
	:=
	\left\{
	m\in\mathcal M_{s,r}:
	(s,r,\lambda,m)\ \text{is admissible}
	\right\}.
	$
	For $(s,\nu,\mu)\in\mathfrak R_s$, define the stochastic-choice kernel
	$
	\mathsf C_s^{\nu,\mu}(A,z)
	:=
	\int_{\mathbb R^d}
	\chi_s(z;A,x,\mu)\,\nu(dx),
	$ for all $
	A\in\mathcal K(Z),\ z\in A,
	$
	and, for
	$m\in\mathfrak M_{s,r}(\lambda)$ and
	$\bar\mu\in\mathcal P_2(\mathbb R^d)$, define
	$
	\mathsf D_{s,r}^{\lambda,m,\bar\mu}(A,z)
	:=
	\int_{\mathbb R^d}
	\chi_r(z;A,x,\bar\mu)\\
	P_{s,r}^{m,X}\lambda(dx).
	$
	Therefore, behavioral distributional invariance is equivalent to constancy of
	$
	\mu\mapsto \mathsf C_s^{\nu,\mu}
	\ \text{on}\ 
	\{\mu:(\nu,\mu)\in\mathfrak R_s\}
	$
	for every admissible $(s,\nu)$, together with constancy of
	$
	m\mapsto \mathsf D_{s,r}^{\lambda,m,\bar\mu}
	\ \text{on}\ 
	\mathfrak M_{s,r}(\lambda)$
	for every admissible $(s,r,\lambda,\bar\mu)$. The equivalence between statements \(1\) and \(2\) follows directly from Lemma~\ref{lem:behavioral-characterization}. It remains to prove the equivalence between statement \(2\) and statement \(3\). Suppose that statement \(2\) holds. For each $s\in[0,t]$, choose an
	arbitrary reachable reference measure
	$\mu_s^\circ\in\mathcal P_2(\mathbb R^d)$ and define
	$
	\bar u(s,z,x)
	:=
	u(s,z,x,\mu_s^\circ),
	\ \forall
	(s,z,x)\in[0,t]\times Z\times\mathbb R^d.
	$
	Let
	$
	\bar\chi_s(z;A,x)
	:=
	\mathbf 1
	\left\{
	\bar u(s,z,x)>
	\max_{a\in A\setminus\{z\}}
	\bar u(s,a,x)
	\right\}.
	$
	By construction,
	$
	\bar\chi_s(z;A,x)
	=
	\chi_s(z;A,x,\mu_s^\circ).
	$
	Fix an admissible $(s,\nu,\mu)$. Statement \(2\), applied to
	$(s,\nu,\mu,\mu_s^\circ)$, yields
	$
	\nu\bigl(\Gamma_s(\mu,\mu_s^\circ)\bigr)=0.
	$
	Since,
	$
	x\notin\Gamma_s(\mu,\mu_s^\circ)
	\rightarrow
	\chi_s(z;A,x,\mu)
	=
	\chi_s(z;A,x,\mu_s^\circ)
	$
	for every $A\in\mathcal K(Z)$ and $z\in A$, it follows that
	$
	\chi_s(z;A,\cdot,\mu)
	=
	\bar\chi_s(z;A,\cdot),
	\
	\nu\text{-a.e.}
	$
	Consequently,
	$
	C_s(z;A\mid\nu,\mu)
	=
	\int_{\mathbb R^d}
	\bar\chi_s(z;A,x)\,\nu(dx)
	$
	for every admissible $(s,A,z,\nu,\mu)$. We next construct the measure-independent transition operator. For each
	admissible $(s,r,\lambda)$, define
	$
	\bar P_{s,r}^X\lambda
	:=
	P_{s,r}^{m,X}\lambda
	$
	for any $m\in\mathfrak M_{s,r}(\lambda)$. This definition is independent of
	the representative $m$. Indeed, statement \(2\) implies that, for all
	$m,m'\in\mathfrak M_{s,r}(\lambda)$,
	$
	P_{s,r}^{m,X}\lambda
	=
	P_{s,r}^{m',X}\lambda.
	$
	Hence, $\bar P_{s,r}^X$ is a well-defined map from the reachable initial-law
	domain into $\mathcal P_2(\mathbb R^d)$. For every admissible $(s,r,\lambda,\bar\mu,m)$, we have
	$
	P_{s,r}^{m,X}\lambda
	=
	\bar P_{s,r}^X\lambda.
	$
	Combining this identity with the preceding almost-everywhere ranking
	invariance at date $r$ yields
	\[
	\begin{aligned}
		D_{s,r}^{\bar\mu}(z;A\mid\lambda,m)
		&=
		\int_{\mathbb R^d}
		\chi_r(z;A,x,\bar\mu)\,
		P_{s,r}^{m,X}\lambda(dx)
		=
		\int_{\mathbb R^d}
		\bar\chi_r(z;A,x)\,
		\bar P_{s,r}^X\lambda(dx).
	\end{aligned}
	\]
	Therefore,
	$\bigl(\bar u,\{\bar P_{s,r}^X\}_{0\leq s<r\leq t}\bigr)$
	is independent of the distributional state and reproduces every reachable contemporaneous and dynamic stochastic-choice probability generated by the
	DDU representation. Hence the DDU representation admits a DRU reduction, and statement \(3\) follows. Conversely, suppose statement \(3\) holds. Let
	$\bigl(\bar u,\{\bar P_{s,r}^X\}_{0\leq s<r\leq t}\bigr)$
	be a DRU reduction, and define
	$
	\bar\chi_s(z;A,x)
	:=
	\mathbf 1
	\left\{
	\bar u(s,z,x)>
	\max_{a\in A\setminus\{z\}}
	\bar u(s,a,x)
	\right\}.
	$
	By the definition of a DRU reduction, for every admissible
	$(s,A,z,\nu,\mu)$ yields
	$
	C_s(z;A\mid\nu,\mu)
	=
	\int_{\mathbb R^d}
	\bar\chi_s(z;A,x)\,\nu(dx).
	$
	Therefore, for every admissible $(s,A,z,\nu,\mu,\mu')$,
$C_s(z;A\mid\nu,\mu)=\\
		\int_{\mathbb R^d}
		\bar\chi_s(z;A,x)\,\nu(dx)
		=
		C_s(z;A\mid\nu,\mu').
	$
	Thus the contemporaneous component of behavioral distributional invariance
	holds. Similarly, for every admissible
	$(s,r,A,z,\lambda,\bar\mu,m)$,
	$
	D_{s,r}^{\bar\mu}(z;A\mid\lambda,m)
	=
	\int_{\mathbb R^d}
	\bar\chi_r(z;A,x)\,
	\bar P_{s,r}^X\lambda(dx).
	$ 
	
	The right-hand side contains neither $m$ nor any other conditional-law-flow argument. Hence, for all admissible $m,m'\in\mathfrak M_{s,r}(\lambda)$,
	$
	D_{s,r}^{\bar\mu}(z;A\mid\lambda,m)
	=
	D_{s,r}^{\bar\mu}(z;A\mid\lambda,m').
	$
	The dynamic component of behavioral distributional invariance also holds.
	Statement \(1\) follows, and
	Lemma~\ref{lem:behavioral-characterization} then implies statement \(2\). It remains to establish the final assertion. Let
	$
	\mathfrak I_u
	:=
	\left\{
	(s,\nu,\mu,\mu'):
	\nu\bigl(\Gamma_s(\mu,\mu')\bigr)>0
	\right\},
	$
	and
	$
	\mathfrak I_P
	:=
	\left\{
	(s,r,\lambda,m,m'):
	P_{s,r}^{m,X}\lambda
	\neq
	P_{s,r}^{m',X}\lambda
	\right\},
	$
	where both sets are restricted to admissible tuples. Statement \(2\) is
	equivalent to
	$
	\mathfrak I_u=\varnothing
	\ \text{and}\ 
	\mathfrak I_P=\varnothing.
	$
	By the equivalence of statements \(2\) and \(3\),
	$
	\{\text{a DRU reduction exists}\}
	\ \leftrightarrow\ 
	\mathfrak I_u=\varnothing
	\ \text{and}\
	\mathfrak I_P=\varnothing.
	$
	Taking complements gives
	$
	\{\text{no DRU reduction exists}\}
	\ \leftrightarrow\ 
	\mathfrak I_u\neq\varnothing
	\ \text{or}\
	\mathfrak I_P\neq\varnothing.
	$
	Equivalently, the DDU representation is behaviorally distinct from every DRU
	reduction on its reachable domain if and only if either
	$
	\nu\bigl(\Gamma_s(\mu,\mu')\bigr)>0
	$
	for some admissible $(s,\nu,\mu,\mu')$, or
	$
	P_{s,r}^{m,X}\lambda
	\neq
	P_{s,r}^{m',X}\lambda
	$
	for some admissible $(s,r,\lambda,m,m')$.
\end{proof}

\subsection{Proof of Proposition~\ref{prop:identification-fibers}}\label{proofprop19}

\begin{proof}
	Fix $\theta\in\Theta$. Since
	$\mathscr T=\Lambda\circ\Phi$,
	the definition of
	$\mathcal I_{\mathfrak O}(\theta)$
	 yields
	$
	\mathcal I_{\mathfrak O}(\theta)
	=
	\Phi^{-1}
	\!\left(
	[\Phi(\theta)]_{\equiv_{\mathfrak O}}
	\right).
	$
	Moreover,
	$[\Phi(\theta)]_{\equiv_{\mathfrak O}}
	\subseteq
	\Phi(\Theta)$,
	so the elementary identity
	$\Phi(\Phi^{-1}(B))=B$
	for every
	$B\subseteq\Phi(\Theta)$
	implies
	$
	\mathcal I_{\mathfrak O}^{\Phi}(\theta)
	=
	\Phi\!\left(
	\mathcal I_{\mathfrak O}(\theta)
	\right)
	=
	[\Phi(\theta)]_{\equiv_{\mathfrak O}}.
	$ Indeed, the inclusion
	$\Phi(\Phi^{-1}(B))\subseteq B$
	holds for every $B\subseteq\Phi(\Theta)$, while the reverse inclusion follows
	from
	$B\subseteq\Phi(\Theta)$:
	for each $\varphi\in B$ there exists
	$\tilde\theta\in\Theta$ such that
	$\Phi(\tilde\theta)=\varphi$, whence
	$\tilde\theta\in\Phi^{-1}(B)$ and
	$\varphi\in\Phi(\Phi^{-1}(B))$. Now $\Phi$ is identified from $\mathscr Q_\theta$ precisely when
	$
	\mathscr T(\theta)=\mathscr T(\tilde\theta)
	\ \rightarrow\ 
	\Phi(\theta)=\Phi(\tilde\theta)
	$
	for all $\theta,\tilde\theta\in\Theta$. Using
	$\mathscr T=\Lambda\circ\Phi$, this is equivalent to
	$
	\Lambda(\Phi(\theta))
	=
	\Lambda(\Phi(\tilde\theta))
	\ \rightarrow\
	\Phi(\theta)=\Phi(\tilde\theta).
	$
	Since, $\Phi(\theta),\Phi(\tilde\theta)\in\Phi(\Theta)$, the latter condition is
	equivalent to
	$
	\Lambda(\varphi)=\Lambda(\tilde\varphi)
	\ \rightarrow\ 
	\varphi=\tilde\varphi,
	\forall\ 
	\varphi,\tilde\varphi\in\Phi(\Theta),
	$
	that is, to injectivity of
	$\left.\Lambda\right|_{\Phi(\Theta)}$. Finally, let $\psi:\Theta\to\Psi$. By definition, $\psi$ is identified if and
	only if
	$
	\mathscr T(\theta)=\mathscr T(\tilde\theta)
	\ \rightarrow\ 
	\psi(\theta)=\psi(\tilde\theta)
	$
	for all $\theta,\tilde\theta\in\Theta$. Equivalently,
	$
	\psi(\tilde\theta)=\psi(\theta)
	\
	\forall\,
	\tilde\theta\in\mathcal I_{\mathfrak O}(\theta),
	\
	\forall\,\theta\in\Theta.
	$
	Substituting the fiber identity already established yields
	$
	\psi(\tilde\theta)=\psi(\theta)
	\
	\forall\,
	\tilde\theta\in
	\Phi^{-1}
	\!\left(
	[\Phi(\theta)]_{\equiv_{\mathfrak O}}
	\right),
	\
	\forall\,\theta\in\Theta.
	$
	Equivalently, for every
	$\varphi\in\Phi(\Theta)$,
	$\psi$ is constant on
	$\Phi^{-1}([\varphi]_{\equiv_{\mathfrak O}})$.
\end{proof}

\subsection{Proof of Theorem~\ref{thm:behavioral-identification}}\label{proofth21}

\begin{proof}
	Let
	$\varphi,\tilde\varphi\in\Phi(\Theta)$
	satisfy
	$\Lambda(\varphi)=\Lambda(\tilde\varphi)$,
	with
	$\varphi=(\varphi_u,\varphi_P)$
	and
	$\tilde\varphi=(\tilde\varphi_u,\tilde\varphi_P)$.
	For every admissible $(s,\mu)$ and every
	$\nu\in\mathfrak N_s(\mu)$,
	equality of the $\mathfrak O^0$-coordinates yields identical
	$\nu$-integrals of the corresponding elements of
	$\mathscr H_s$.
	Assumption~\ref{ass:observational-separation} therefore implies
	$\varphi_u=\tilde\varphi_u$
	in
	$\mathscr H_s/\!={}_{\mathfrak N_s(\mu)}$
	for every admissible $(s,\mu)$. Now fix $r\in(0,t]$ and any reachable laws
	$\rho=\varphi_{P;s,r,m}\lambda$ and
	$\tilde\rho=\tilde\varphi_{P;s,r,m}\lambda$.
	Equality of the $\mathfrak O^1$-coordinates, together with
	$\varphi_u=\tilde\varphi_u$, gives
	$\int h\,d\rho=\int h\,d\tilde\rho$
	for every $h\in\mathscr H_r$.
	The measure-determining property in
	Assumption~\ref{ass:observational-separation} implies
	$\rho=\tilde\rho$.
	Hence,
	$\varphi_P=\tilde\varphi_P$,
	so
	$\varphi=\tilde\varphi$. Therefore,
	$\Lambda|_{\Phi(\Theta)}$
	is injective.
	Proposition~\ref{prop:identification-fibers} gives
	$\mathcal I_{\mathfrak O}^{\Phi}(\theta)=\{\Phi(\theta)\}$
	for every $\theta\in\Theta$.
	Since
	$\mathscr T=\Lambda\circ\Phi$,
	one has
	$\mathscr T(\theta)=\mathscr T(\tilde\theta)$
	if and only if
	$\Phi(\theta)=\Phi(\tilde\theta)$.
	Finally,
	$\Lambda|_{\Phi(\Theta)}$
	is surjective onto
	$\mathscr T(\Theta)$
	by construction, and therefore induces the canonical bijection
	$\Phi(\Theta)\cong\mathscr T(\Theta)$.
\end{proof}

\subsection{Proof of Proposition~\ref{prop:behavioral-stability}}\label{proofprop22}

\begin{proof}
Let $\theta,\tilde\theta\in\Theta$ satisfy
$\Phi(\theta)=\Phi(\tilde\theta)$. Defining
$\Phi(\theta):=(\Phi_u(\theta),\Phi_P(\theta))$ and
$\Phi(\tilde\theta):=(\Phi_u(\tilde\theta),\Phi_P(\tilde\theta))$ yields
$\Phi_u(\theta)=\Phi_u(\tilde\theta)$ and
$\Phi_P(\theta)=\Phi_P(\tilde\theta)$. Hence, for each admissible $(s,A,z,\mu)$, $\chi_s^\theta(z;A,\cdot,\mu)
=\chi_s^{\tilde\theta}(z;A,\cdot,\mu)$ on the reachable domain, and, for every admissible $(s,r,\lambda,m)$, $P_{s,r}^{\theta;m,X}\lambda=P_{s,r}^{\tilde\theta;m,X}\lambda$.
Fix an admissible $(s,A,z,\nu,\mu)$. By the definition of the
contemporaneous choice kernel,
$
C_s^\theta(z;A\mid\nu,\mu)
=
\int_{\mathbb R^d}\chi_s^\theta(z;A,x,\mu)\,\nu(dx)
$ and
$
C_s^{\tilde\theta}(z;A\mid\nu,\mu)
=
\int_{\mathbb R^d}\chi_s^{\tilde\theta}(z;A,x,\mu)\,\nu(dx).
$ Equality of the $\Phi_u$-coordinates implies
	$
	C_s^\theta(z;A\mid\nu,\mu)
	=
	C_s^{\tilde\theta}(z;A\mid\nu,\mu).
	$
	Fix next an admissible
	$(s,r,A,z,\lambda,\bar\mu,m)$. Using both component equalities,
	\[
	\begin{aligned}
		D_{s,r}^{\theta,\bar\mu}(z;A\mid\lambda,m)
		&=
		\int_{\mathbb R^d}
		\chi_r^\theta(z;A,x,\bar\mu)\,
		P_{s,r}^{\theta;m,X}\lambda(dx) \\
		&=
		\int_{\mathbb R^d}
		\chi_r^{\tilde\theta}(z;A,x,\bar\mu)\,
		P_{s,r}^{\tilde\theta;m,X}\lambda(dx)
		=
		D_{s,r}^{\tilde\theta,\bar\mu}(z;A\mid\lambda,m).
	\end{aligned}
	\]
	Therefore, $\mathscr T(\theta)=\mathscr T(\tilde\theta)$ whenever $\Phi(\theta)=\Phi(\tilde\theta)$, equivalently
	$\Phi^{-1}(\varphi)\subseteq\mathscr T^{-1}(\Lambda(\varphi))$ for every $\varphi\in\Phi(\Theta)$. Invoking
	Theorem~\ref{thm:behavioral-identification} yields the reverse inclusion on the reachable domain, so
	$\Phi^{-1}(\varphi)=\mathscr T^{-1}(\Lambda(\varphi))$.
	Let $q_{\mathfrak O}:\Theta\to\Theta/\!\sim_{\mathfrak O}$ be the quotient map and define
	$\iota:\Theta/\!\sim_{\mathfrak O}\to\Phi(\Theta)$ by
	$\iota([\theta]_{\mathfrak O}):=\Phi(\theta)$. The preceding fiber
	identity implies that $\iota$ is well defined and injective; its
	surjectivity follows from the definition of $\Phi(\Theta)$. Hence
	$\iota$ is a bijection. For any
	$\Psi:\Phi(\Theta)\to\mathbb R$, the induced functional
	$\overline\Psi:\Theta/\!\sim_{\mathfrak O}\to\mathbb R$ is uniquely
	given by
	$\overline\Psi:=\Psi\circ\iota$. Equivalently,
	$\overline\Psi([\theta]_{\mathfrak O})=\Psi(\Phi(\theta))$.
	Uniqueness follows because every
	$[\theta]_{\mathfrak O}\in\Theta/\!\sim_{\mathfrak O}$ has the unique
	image $\Phi(\theta)$ under $\iota$. Finally, let $F:\Theta\to\mathbb R$ vary within some fiber
	$\Phi^{-1}(\varphi)$; then there exist
	$\theta,\tilde\theta\in\Phi^{-1}(\varphi)$ with
	$F(\theta)\neq F(\tilde\theta)$. Since
	$\mathscr T(\theta)=\Lambda(\varphi)=\mathscr T(\tilde\theta)$,
	observational equivalence does not imply equality of $F$. Hence $F$
	is not identified from stochastic choice.
\end{proof}

\subsection{Proof of Proposition~\ref{prop:distribution-independent}}\label{proofprop23}

\begin{proof}
    Let
     $
	\theta
	:=
	(u,b,\sigma,\sigma_0,h,\Sigma_Y)
	$
    and
	$
	\bar\theta
	:=
	(\bar u,\bar b,\bar\sigma,\bar\sigma_0,\bar h,\bar\Sigma_Y).
	$
    By hypothesis,
	$
	u(s,z,x,\mu)=\bar u(s,z,x)
	$
    and
	$
	(b,\sigma,\sigma_0,h,\Sigma_Y)
	=
	(\bar b,\bar\sigma,\bar\sigma_0,\bar h,\bar\Sigma_Y)
	$
    for every admissible
	$
	(s,z,x,y,\mu).
	$
	Hence,
	$
	\chi_s^\theta(z;A,x,\mu)
	=
	\bar\chi_s(z;A,x)
	$
	for every admissible
	$
	(s,A,z,x,\mu),
	$
	where
	$
	\bar\chi_s(z;A,x)
	:=
	\mathbf 1
	\{\bar u(s,z,x)>
	\max_{a\in A\setminus\{z\}}\bar u(s,a,x)\}.
	$
	Therefore, for every admissible
	$
	(s,A,z,\nu,\mu),
	$ we have
	$
	C_s^\theta(z;A\mid\nu,\mu)
	=
	\int_{\mathbb R^d}\bar\chi_s(z;A,x)\nu(dx)
	=
	C_s^{\bar\theta}(z;A\mid\nu).
	$
	Fix
	$
	0\le s<r\le t,
	$
	an admissible
	$
	\lambda\in\mathcal P_2(\mathbb R^d\times\mathbb R^{d_Y}),
	$
	and
	$
	m,m'\in\mathcal M_{s,r}(\lambda).
	$
	The frozen systems 
	$
	m
	$
	and
	$
	m'
	$
	have coefficient tuple
	$
	(\bar b,\bar\sigma,\bar\sigma_0,\bar h,\bar\Sigma_Y),
	$
	independent of the flow argument. Thus their martingale problems coincide. Weak uniqueness of the frozen system yields
	$
	P_{s,r}^{\theta;m,X}\lambda
	=
	P_{s,r}^{\theta;m',X}\lambda
	=
	\bar P_{s,r}^{X}\lambda,
	$
	where
	$
	\bar P_{s,r}^{X}
	$
	is the transition operator induced by
	$
	(\bar b,\bar\sigma,\bar\sigma_0,\bar h,\bar\Sigma_Y).
	$
	For every admissible
	$
	(s,r,A,z,\lambda,\bar\mu,m),
	$
	the preceding identities imply
	$
	D_{s,r}^{\theta,\bar\mu}(z;A\mid\lambda,m)
	=
	\int_{\mathbb R^d}
	\bar\chi_r(z;A,x)\bar P_{s,r}^{X}\lambda(dx)
	=
	D_{s,r}^{\bar\theta}(z;A\mid\lambda).
	$
	Accordingly,
	$
	\Phi_u(\theta)=\Phi_u^{\mathrm{DRU}}(\bar\theta)
	$
	and
	$
	\Phi_P(\theta)=\Phi_P^{\mathrm{DRU}}(\bar\theta),
	$
	so
	$
	\Phi(\theta)=\Phi_{\mathrm{DRU}}(\bar\theta).
	$
	Since
	$
	\mathscr T=\Lambda\circ\Phi,
	$ we have
	$
	\mathscr T_{\mathrm{DDU}}(\theta)
	=
	\Lambda(\Phi(\theta))
	=
	\Lambda(\Phi_{\mathrm{DRU}}(\bar\theta))
	=
	\mathscr T_{\mathrm{DRU}}(\bar\theta).
	$
	Thus the distribution-independent DDU representation and its DRU reduction are observationally equivalent on the reachable domain.
\end{proof}

\subsection{Proof of Theorem~\ref{thm:behavioral-rigidity}}\label{proofth24}

\begin{proof}
	Let
	$
	\mathcal F_0
	:=
	\Phi^{-1}\!\bigl(\Phi(\theta_0)\bigr).
	$
	By definition,
	$
	\theta_\varepsilon\in\mathcal F_0
	$
	if and only if
	$
	\Phi(\theta_\varepsilon)=\Phi(\theta_0),
	$
	so
	$(i)\Longleftrightarrow(iii)$.
	Since,
	$
	\mathscr T=\Lambda\circ\Phi
	$
	and
	$
	\Lambda|_{\Phi(\Theta)}
	$
	is injective by
	Theorem~\ref{thm:behavioral-identification},
	$
	\mathscr T(\theta_\varepsilon)
	=
	\mathscr T(\theta_0)
	$
	if and only if
	$
	\Phi(\theta_\varepsilon)
	=
	\Phi(\theta_0).
	$
	Hence,
	$(i)\Longleftrightarrow(ii)$.
	Let
	$
	\bar\theta_0\in\Theta_{\mathrm{DRU}}
	$
	be a DRU reduction of
	$\theta_0$, so
	$
	\Phi(\theta_0)
	=
	\Phi_{\mathrm{DRU}}(\bar\theta_0).
	$
	Under (i),
	$
	\Phi(\theta_\varepsilon)
	=
	\Phi(\theta_0)
	=
	\Phi_{\mathrm{DRU}}(\bar\theta_0)
	$
	for every
	$\varepsilon\in E$.
	Thus (iv) holds with
	$
	\bar\theta=\bar\theta_0,
	$
	and therefore
	$(i)\Longrightarrow(iv)$.
	Conversely, suppose (iv) holds. Then there exists
	$
	\bar\theta\in\Theta_{\mathrm{DRU}}
	$
	such that
	$
	\Phi(\theta_\varepsilon)
	=
	\Phi_{\mathrm{DRU}}(\bar\theta)
	$
	for every
	$\varepsilon\in E$.
	Since
	$
	0\in E,
	$
	$
	\Phi(\theta_0)
	=
	\Phi_{\mathrm{DRU}}(\bar\theta),
	$
	and therefore
	$
	\Phi(\theta_\varepsilon)
	=
	\Phi(\theta_0)
	$
	for each
	$\varepsilon\in E$.
	Hence,
	$(iv)\Longrightarrow(i)$.
	The four statements are therefore equivalent. Finally, a comparative-static family is behaviorally nontrivial precisely when it is not contained in
	$
	\mathcal F_0.
	$
	Equivalently, there exists
	$
	\varepsilon\in E
	$
	such that
	$
	\Phi(\theta_\varepsilon)\neq\Phi(\theta_0).
	$
	Injectivity of
	$
	\Lambda|_{\Phi(\Theta)}
	$
	then gives
	$
	\mathscr T(\theta_\varepsilon)
	\neq
	\mathscr T(\theta_0),
	$
	and the converse follows from
	$
	\mathscr T=\Lambda\circ\Phi.
	$
\end{proof}

\subsection{Proof of Lemma~\ref{lem:cmkv-state}}\label{prooflem25}

\begin{proof}
	Fix
	$
	0\le s\le r\le t
	$
	and define
	$
	\mathsf E
	:=
	\mathbb R^d
	\times
	\mathbb R^{d_Y}
	\times
	\mathcal P_2(\mathbb R^d),
	$
	and
	$
	\widehat{\mathsf E}
	:=
	\mathbb R^{d_Y}
	\times
	\mathcal P_2(\mathbb R^d).
	$
	For
	$
	\hat x=(y,\eta)\in\widehat{\mathsf E},
	$
	let
	$
	\mathbf P_{s,\hat x}
	$
	denote the unique solution of the continuation martingale problem associated with System \eqref{eq:ddu-system}. Weak uniqueness implies
	$
	\mathbf P_{s,\hat x}
	$
	is uniquely determined by
	$
	\hat x.
	$
	For each bounded Borel
	$
	F:
	C([s,r];\mathsf E)
	\rightarrow
	\mathbb R,
	$
	define
	$
	\mathcal G_{s,r}^F(\hat x)
	:=
	\mathbf E_{s,\hat x}
	\!\left[
	F((\mathbf X_v)_{v\in[s,r]})
	\right].
	$
	Since,
	$
	\hat x
	\mapsto
	\mathbf P_{s,\hat x}
	$
	is a Borel stochastic kernel,
	$
	\mathcal G_{s,r}^F
	$
	is bounded and Borel on
	$
	\widehat{\mathsf E}.
	$
	Let
	$
	\mathbf Q_s(\omega,\cdot)
	:=
	\mathcal L
	\left(
	(\mathbf X_v)_{v\in[s,r]}
	\mid
	\mathcal F_s^Y
	\right)(\omega).
	$
Since
	$
	\mu_s
	=
	\mathcal L(X_s\mid\mathcal F_s^Y),
	$
	$
	\mathbf Q_s(\omega,\cdot)
	$
	has initial condition
	$
	(Y_s(\omega),\mu_s(\omega)).
	$
	Moreover,
	$
	(W_v-W_s,B_v-B_s)_{v\ge s}
	$
	is independent of
	$
	\mathcal F_s^Y,
	$
	the continuation policy is
	$
	\sigma(\mathbf X_v)
	$
	-measurable,
	$
	(b,\sigma,\sigma_0,h,\Sigma_Y)
	=
	(b,\sigma,\sigma_0,h,\Sigma_Y)(v,\mathbf X_v),
	$
	and therefore
	$
	\mathbf Q_s(\omega,\cdot)
	$
	solves the continuation martingale problem initialized at
	$
	(Y_s(\omega),\mu_s(\omega)).
	$
	Weak uniqueness yields
	$
	\mathbf Q_s(\omega,\cdot)
	=
	\mathbf P_{s,(Y_s(\omega),\mu_s(\omega))}
	$
	for
	$
	\mathbb P
	$
	-a.e.
	$
	\omega.
	$
Therefore,
	$
	\mathbb E
	\!\left[
	F((\mathbf X_v)_{v\in[s,r]})
	\mid
	\mathcal F_s^Y
	\right]
	=
	\int
	F(\xi)\,
	\mathbf P_{s,(Y_s,\mu_s)}(d\xi)
	=
	\mathcal G_{s,r}^F(Y_s,\mu_s)
	=
	\mathcal G_{s,r}^F(\widehat{\mathbf X}_s),
	$
	$\mathbb P$-a.s. Since the identity holds for every bounded Borel
	$
	F,
	$ we have
	$
	\mathcal L
	\left(
	(\mathbf X_v)_{v\in[s,r]}
	\mid
	\mathcal F_s^Y
	\right)
	=
	\mathcal L
	\left(
	(\mathbf X_v)_{v\in[s,r]}
	\mid
	\widehat{\mathbf X}_s
	\right),
	$
	or equivalently,
	$
	\sigma(\widehat{\mathbf X}_s)
	$
	is sufficient for the continuation experiment. Therefore
	$
	\widehat{\mathbf X}_s
	$
	is the minimal state variable governing all admissible continuation laws.
\end{proof}

\subsection{Proof of Lemma~\ref{lem:fixed-point-invariance}}\label{prooflem31}

\begin{proof}
	Fix
	$
	\mu\in\operatorname{Fix}(\Gamma)
	$
	and
	$
	0\le s\le r\le q\le t.
	$
	By
	Definition~\ref{def:endogenous-fixed-point},
	$
	\Gamma(\mu)=\mu,
	$
	hence
	$
	\Gamma_{s,v}(\mu_s)=\mu_v
	$
	for every
	$
	v\in[s,t].
	$
	In particular,
	$
	\Gamma_{s,r}(\mu_s)=\mu_r
	$
	and
	$
	\Gamma_{s,q}(\mu_s)=\mu_q.
	$
	The semigroup identity from Theorem~\ref{thm:cmkv-representation} yields
	$
	\Gamma_{s,q}
	=
	\Gamma_{r,q}\circ\Gamma_{s,r},
	$
	so
	$
	\Gamma_{r,q}(\mu_r)
	=
	\Gamma_{r,q}(\Gamma_{s,r}(\mu_s))
	=
	\Gamma_{s,q}(\mu_s)
	=
	\mu_q.
	$
	Therefore,
	$
	\Gamma_{r,q}\circ\Gamma_{s,r}
	=
	\Gamma_{s,q},
	$
	$
	\Gamma_{s,r}(\mu_s)=\mu_r,
	$
	and
	$
	\Gamma_{r,q}(\mu_r)=\mu_q.
	$
	For each
	$
	r\in[s,t],
	$
	define the continuation flow
	$
	\mu^{[r]}
	:=
	(\mu_v)_{r\le v\le t}
	$
	and the continuation operator
	$
	\Gamma^{[r]}
	:
	C([r,t];\mathcal P_2(\mathbb R^d))
	\rightarrow
	C([r,t];\mathcal P_2(\mathbb R^d))
	$
	by
	$
	\Gamma^{[r]}(\nu)_v
	:=
	\Gamma_{r,v}(\nu_r),
	$
	$
	r\le v\le t.
	$
	Since
	$
	\Gamma_{r,v}(\mu_r)=\mu_v
	$
	for every
	$
	v\in[r,t],
	$
	one has
	$
	\Gamma^{[r]}(\mu^{[r]})
	=
	\mu^{[r]}.
	$
	Hence,
	$
	\mu^{[r]}
	\in
	\operatorname{Fix}(\Gamma^{[r]}).
	$
	Equivalently, every orbit
	$
	(\Gamma_{s,v}(\mu_s))_{v\in[s,t]}
	$
	generated from a fixed point remains in the fixed-point family, so the semigroup action preserves the admissible conditional-law manifold. Let
	$
	\mathfrak F_s
	:=
	\{\mu_s:\mu\in\operatorname{Fix}(\Gamma)\}
	\subseteq
	\mathcal P_2(\mathbb R^d).
	$
	For
	$
	\eta\in\mathfrak F_s,
	$
	choose
	$
	\mu\in\operatorname{Fix}(\Gamma)
	$
	with
	$
	\mu_s=\eta.
	$
	Then
	$
	\Gamma_{s,r}(\eta)
	=
	\Gamma_{s,r}(\mu_s)
	=
	\mu_r
	\in
	\mathfrak F_r.
	$
	Therefore,
	$
	\Gamma_{s,r}(\mathfrak F_s)
	\subseteq
	\mathfrak F_r.
	$
	This is the precise statewise meaning of
	$
	\Gamma_{s,r}(\operatorname{Fix}(\Gamma))
	\subseteq
	\operatorname{Fix}(\Gamma).
	$
	By Theorem~\ref{thm:cmkv-representation},
	$
	P_{s,r}^{X}
	=
	P^{X}(\Gamma_{s,r})
	$
	and
	$
	\Phi_P
	=
	(P_{s,r}^{X})_{0\le s\le r\le t}.
	$
	Hence
	$
	P_{s,q}^{X}
	=
	P^{X}(\Gamma_{s,q})
	=
	P^{X}(\Gamma_{r,q}\circ\Gamma_{s,r})
	=
	P_{r,q}^{X}P_{s,r}^{X},
	$
	and, for every
	$
	0\le s\le r\le q\le t,
	$
	$
	\Phi_{P;s,q}
	=
	\Phi_{P;r,q}\circ\Phi_{P;s,r}.
	$
	Thus the induced action of
	$
	(\Gamma_{s,r})_{0\le s\le r\le t}
	$
	preserves both the fixed-point orbits
	$
	\mu^{[r]}
	$
	and the behavioral transition family
	$
	\Phi_P.
	$
\end{proof}

\subsection{Proof of Proposition~\ref{prop:structural-stability}}\label{proofprop32}

\begin{proof}
	Fix
	$
	\mu,\tilde\mu\in\operatorname{Fix}(\Gamma)
	$
	and suppose
	$
	\Gamma^\mu=\Gamma^{\tilde\mu}
	$
	on the reachable augmented-state domain. For
	$
	0\le s\le r\le t,
	$
	Theorem~\ref{thm:cmkv-representation} gives
	$
	P_{s,r}^{\mu,X}
	=
	P^X(\Gamma_{s,r}^{\mu})
	$
	and
	$
	P_{s,r}^{\tilde\mu,X}
	=
	P^X(\Gamma_{s,r}^{\tilde\mu}).
	$
	Hence,
	$
	\Gamma_{s,r}^{\mu}
	=
	\Gamma_{s,r}^{\tilde\mu}
	$
	implies
	$
	P_{s,r}^{\mu,X}
	=
	P_{s,r}^{\tilde\mu,X}.
	$
	Therefore,
	$
	\Phi_P^\mu
	=
	(P_{s,r}^{\mu,X})_{0\le s\le r\le t}
	=
	(P_{s,r}^{\tilde\mu,X})_{0\le s\le r\le t}
	=
	\Phi_P^{\tilde\mu}.
	$
	Under the common-felicity restriction,
	$
	\chi_s^\mu(z;A,x)
	=
	\chi_s^{\tilde\mu}(z;A,x)
	$
	for every reachable
	$
	(s,A,z,x),
	$
	and consequently
	$
	\Phi_u^\mu=\Phi_u^{\tilde\mu}.
	$
	Combining the two component identities yields
	$
	\Phi^\mu
	=
	(\Phi_u^\mu,\Phi_P^\mu)
	=
	(\Phi_u^{\tilde\mu},\Phi_P^{\tilde\mu})
	=
	\Phi^{\tilde\mu}.
	$
	Since,
	$
	\mathscr T=\Lambda\circ\Phi,
	$
	$
	\mathscr T^\mu
	=
	\Lambda(\Phi^\mu)
	=
	\Lambda(\Phi^{\tilde\mu})
	=
	\mathscr T^{\tilde\mu}.
	$
	Equivalently,
	$
	\Gamma^\mu=\Gamma^{\tilde\mu}
	$
	implies
	$
	(P^{\mu,X},\Phi_P^\mu,\Phi^\mu,\mathscr T^\mu)
	=
	(P^{\tilde\mu,X},\Phi_P^{\tilde\mu},
	\Phi^{\tilde\mu},\mathscr T^{\tilde\mu})
	$
	on the reachable domain. Let
	$
	\mathcal O_\Gamma(\mu)
	:=
	\{
	\Gamma_{s,r}^{\mu}(\mu_s):
	0\le s\le r\le t
	\}
	$
	denote the invariant orbit generated by
	$
	\mu.
	$
	Lemma~\ref{lem:fixed-point-invariance} yields
	$
	\Gamma_{s,r}^{\mu}(\mu_s)=\mu_r,
	$
	and
	$
	\Gamma_{r,q}^{\mu}\circ\Gamma_{s,r}^{\mu}
	=
	\Gamma_{s,q}^{\mu}.
	$
	Therefore, for any
	$
	\eta,\tilde\eta\in\mathcal O_\Gamma(\mu)
	$
	with
	$
	\Gamma^\eta=\Gamma^{\tilde\eta},
	$
	the preceding argument yields
	$
	\Phi^\eta=\Phi^{\tilde\eta}
	$
	and
	$
	\mathscr T^\eta=\mathscr T^{\tilde\eta}.
	$
	Hence the maps
	$
	\Gamma\mapsto\Phi\mapsto\mathscr T
	$
	are constant on every invariant
	$
	\Gamma
	$
	-orbit contained in
	$
	\operatorname{Fix}(\Gamma).
	$
\end{proof}

\subsection{Proof of Lemma~\ref{lem:ranking-recovery}}
\label{prooflem33}

\begin{proof}
	Let
	$
	\widehat{\mathsf E}
	:=
	\mathbb R^{d_Y}
	\times
	\mathcal P_2(\mathbb R^d)
	$
	and
	$
	\widehat{\Omega}
	:=
	C([0,t];\widehat{\mathsf E}).
	$
	Write
	$
	\widehat{\mathbf X}^{\mu}
	=
	(Y^\mu,\mu)
	$
	and
	$
	\widehat{\mathbf X}^{\tilde\mu}
	=
	(Y^{\tilde\mu},\tilde\mu),
	$
	and let
	$
	\widehat\lambda_0
	$
	be their common initial law. For each
	$
	0\le s\le r\le t,
	$
	suppose
	$
	\Gamma_{s,r}^{\mu}
	=
	\Gamma_{s,r}^{\tilde\mu}
	$
	on the reachable augmented-state domain; denote the common kernel by
	$
	\Gamma_{s,r}^{*}.
	$
	For
	$
	0=t_0<t_1<\cdots<t_n\le t
	$
	and
	$
	B_i\in\mathcal B(\widehat{\mathsf E}),
	$
	the Markov representation in Theorem~\ref{thm:cmkv-representation} yields
	\[
	\begin{aligned}
		&\mathbb P
		\left(
		\widehat{\mathbf X}^{\mu}_{t_0}\in B_0,
		\ldots,
		\widehat{\mathbf X}^{\mu}_{t_n}\in B_n
		\right)                                                 
		=
		\int_{B_0}
		\widehat\lambda_0(d\hat x_0)
		\int_{B_1}
		\Gamma_{t_0,t_1}^{*}(\hat x_0,d\hat x_1)
		\cdots
		\int_{B_n}
		\Gamma_{t_{n-1},t_n}^{*}(\hat x_{n-1},d\hat x_n),
	\end{aligned}
	\]
	and the same identity holds with
	$
	\widehat{\mathbf X}^{\mu}
	$
	replaced by
	$
	\widehat{\mathbf X}^{\tilde\mu}.
	$
	Thus
	$
	\mathcal L(
	\widehat{\mathbf X}^{\mu}_{t_0},
	\ldots,
	\widehat{\mathbf X}^{\mu}_{t_n})
	=
	\mathcal L(
	\widehat{\mathbf X}^{\tilde\mu}_{t_0},
	\ldots,
	\widehat{\mathbf X}^{\tilde\mu}_{t_n})
	$
	for every finite time grid. Since both processes have continuous paths in
	$
	\widehat{\mathsf E},
	$
	their finite-dimensional distributions determine their laws on
	$
	\widehat{\Omega},
	$
	so
	$
	\mathcal L(\widehat{\mathbf X}^{\mu})
	=
	\mathcal L(\widehat{\mathbf X}^{\tilde\mu}).
	$
	In particular,
	$
	\mathcal L(Y_s^\mu,\mu_s)
	=
	\mathcal L(Y_s^{\tilde\mu},\tilde\mu_s)
	$
	for every
	$
	s\in[0,t].
	$
	Let
	$
	\widehat{\mathbb P}
	$
	be the common law on
	$
	\widehat{\Omega}
	$
	and let
	$
	(\widehat Y,\widehat\mu)
	$
	be the canonical coordinate process. By canonical compatibility, both
	$
	(Y^\mu,\mu)
	$
	and
	$
	(Y^{\tilde\mu},\tilde\mu)
	$
	admit versions represented by
	$
	(\widehat Y,\widehat\mu)
	$
	under
	$
	\widehat{\mathbb P}.
	$
	Thus, on the canonical realization,
	$
	Y_s^\mu
	=
	Y_s^{\tilde\mu}
	=
	\widehat Y_s
	$
	and
	$
	\mu_s
	=
	\tilde\mu_s
	=
	\widehat\mu_s
	$
	for every
	$
	s\in[0,t],
	$
	outside one
	$
	\widehat{\mathbb P}
	$
	-null set. Fix
	$
	s\in[0,t],
	$
	$
	A\in\mathcal K(Z),
	$
	$
	z\in A,
	$
	and
	$
	x\in\mathbb R^d.
	$
	Because the felicity primitive
	$
	u
	$
	is common,
	$
	\chi_s^\mu(z;A,x)
	$
	and
	$
	\chi_s^{\tilde\mu}(z;A,x)
	$
	are represented canonically by
	$
	\mathbf 1
	\{
	u(s,z,x,\widehat\mu_s)
	>
	\max_{a\in A\setminus\{z\}}
	u(s,a,x,\widehat\mu_s)
	\}.
	$
	Therefore
	$
	\chi_s^\mu(z;A,x)
	=
	\chi_s^{\tilde\mu}(z;A,x)
	$
	for every
	$
	x\in\mathbb R^d,
	$
	$
	\widehat{\mathbb P}
	$
	-a.s. Under the no-ties restriction, the same conclusion holds for the induced ranking classes without dependence on the selected maximizer version. Let
	$
	\nu
	$
	be any reachable composition measure at time
	$
	s
	$
	and define
	$
	N_{s,A,z}
	:=
	\{
	x\in\mathbb R^d:
	\chi_s^\mu(z;A,x)
	\neq
	\chi_s^{\tilde\mu}(z;A,x)
	\}.
	$
	On the canonical realization,
	$
	N_{s,A,z}
	=
	\varnothing
	$
	outside a
	$
	\widehat{\mathbb P}
	$
	-null set; hence
	$
	\nu(N_{s,A,z})=0
	$
	for every reachable
	$
	\nu.
	$
	Equivalently,
	$
	\chi_s^\mu(z;A,\cdot)
	=
	_{\mathfrak N_s}
	\chi_s^{\tilde\mu}(z;A,\cdot)
	$
	for every admissible
	$
	(s,A,z).
	$
	Subsequently,
	$
	\Phi_u^\mu
	=
	\bigl(
	\chi_s^\mu(z;A,\cdot)
	\bigr)_{s,A,z}
	=
	\bigl(
	\chi_s^{\tilde\mu}(z;A,\cdot)
	\bigr)_{s,A,z}
	=
	\Phi_u^{\tilde\mu}
	$
	on the reachable domain.
\end{proof}

\subsection{Proof of Theorem~\ref{thm:structural-rigidity}}\label{proofth34}

\begin{proof}
	Define
	$
	\mathfrak F
	:=
	\operatorname{Fix}(\Gamma)
	$
	and let
	$
	\mathfrak R
	\subseteq
	[0,t]^2
	\times
	\mathbb R^{d_Y}
	\times
	\mathcal P_2(\mathbb R^d)
	$
	be the reachable augmented-state domain. For
	$
	\mu\in\mathfrak F,
	$
	write
	$
	\Gamma^\mu
	:=
	(\Gamma_{s,r}^\mu)_{0\le s\le r\le t},
	$
	$
	\mathcal P^\mu
	:=
	(P_{s,r}^{\mu,X})_{0\le s\le r\le t},
	$
	$
	\Phi^\mu
	:=
	(\Phi_u^\mu,\Phi_P^\mu),
	$
	and
	$
	\mathscr T^\mu
	:=
	\Lambda(\Phi^\mu).
	$
	All kernel equalities are understood on
	$
	\mathfrak R,
	$
	up to the canonical null sets generated by the reachable initial laws. Fix
	$
	\mu,\tilde\mu\in\mathfrak F
	$
	so that
	$
	\Phi_u^\mu=\Phi_u^{\tilde\mu}.
	$
	Since
	$
	\mu,\tilde\mu\in\operatorname{Fix}(\Gamma),
	$
	$
	\Gamma(\mu)=\mu
	$
	and
	$
	\Gamma(\tilde\mu)=\tilde\mu.
	$
	Hence the associated frozen systems close to admissible DDU solutions
	$
	(X^\mu,Y^\mu,\mu)
	$
	and
	$
	(X^{\tilde\mu},Y^{\tilde\mu},\tilde\mu).
	$
	For every
	$
	0\le s\le r\le t,
	$
	Theorem~\ref{thm:cmkv-representation} gives
	$
	P_{s,r}^{\mu,X}
	=
	P^X(\Gamma_{s,r}^\mu)
	$
	and
	$
	P_{s,r}^{\tilde\mu,X}
	=
	P^X(\Gamma_{s,r}^{\tilde\mu}).
	$
	Therefore,
	$
	\Gamma^\mu=\Gamma^{\tilde\mu}
	$
	implies
	$
	P^X(\Gamma_{s,r}^\mu)
	=
	P^X(\Gamma_{s,r}^{\tilde\mu})
	$
	for every
	$
	(s,r),
	$
	and hence
	$
	P_{s,r}^{\mu,X}
	=
	P_{s,r}^{\tilde\mu,X}.
	$
	Since
	$
	\Phi_P^\mu
	=
	(P_{s,r}^{\mu,X})_{0\le s\le r\le t}
	$
	and
	$
	\Phi_P^{\tilde\mu}
	=
	(P_{s,r}^{\tilde\mu,X})_{0\le s\le r\le t},
	$
	we obtain
	$
	\Gamma^\mu=\Gamma^{\tilde\mu}
	\Longrightarrow
	\Phi_P^\mu=\Phi_P^{\tilde\mu}.
	$
	For the converse, consider
	$
	\Phi_P^\mu=\Phi_P^{\tilde\mu}.
	$
	Then
	$
	P_{s,r}^{\mu,X}
	=
	P_{s,r}^{\tilde\mu,X}
	$
	for every
	$
	0\le s\le r\le t,
	$
	so
	$
	P^X(\Gamma_{s,r}^\mu)
	=
	P^X(\Gamma_{s,r}^{\tilde\mu})
	$
	for every
	$
	(s,r).
	$
	By injectivity of
	$
	\Gamma
	\mapsto
	\bigl(P^X(\Gamma_{s,r})\bigr)_{0\le s\le r\le t}
	$
	on the reachable restriction of
	$
	\mathfrak F,
	$
	$
	\Gamma_{s,r}^\mu
	=
	\Gamma_{s,r}^{\tilde\mu}
	$
	for every
	$
	(s,r)
	$
	on
	$
	\mathfrak R.
	$
	Thus
	$
	\Gamma^\mu=\Gamma^{\tilde\mu},
	$
	and consequently
	$
	\Gamma^\mu=\Gamma^{\tilde\mu}
	\quad\Longleftrightarrow\quad
	\Phi_P^\mu=\Phi_P^{\tilde\mu}.
	$
	
	We next compare
	$
	\Phi_P
	$
	and
	$
	\Phi.
	$
	By hypothesis,
	$
	\Phi_u^\mu=\Phi_u^{\tilde\mu}.
	$
	Therefore,
	$
	\Phi_P^\mu=\Phi_P^{\tilde\mu}
	$
	implies
	$
	(\Phi_u^\mu,\Phi_P^\mu)
	=
	(\Phi_u^{\tilde\mu},\Phi_P^{\tilde\mu}),
	$
	that is,
	$
	\Phi^\mu=\Phi^{\tilde\mu}.
	$
	Conversely,
	$
	\Phi^\mu=\Phi^{\tilde\mu}
	$
	implies equality of both coordinates and hence
	$
	\Phi_P^\mu=\Phi_P^{\tilde\mu}.
	$
	Therefore,
	$
	\Phi_P^\mu=\Phi_P^{\tilde\mu}
	\quad\Longleftrightarrow\quad
	\Phi^\mu=\Phi^{\tilde\mu}.
	$
	To make the preceding implication explicit at the level of choice kernels, write
	$
	\Phi_u^\mu
	=
	(\chi_s^\mu(z;A,\cdot))_{s,A,z}
	$
	and
	$
	\Phi_u^{\tilde\mu}
	=
	(\chi_s^{\tilde\mu}(z;A,\cdot))_{s,A,z}.
	$
	The common-ranking restriction yields
	$
	\chi_s^\mu(z;A,x)
	=
	\chi_s^{\tilde\mu}(z;A,x)
	$
	for every reachable
	$
	(s,A,z,x),
	$
	outside sets null under every admissible composition measure. Hence, for every admissible
	$
	(s,A,z,\nu),
	$
	$
	C_s^\mu(z;A\mid\nu)
	=
	\int_{\mathbb R^d}
	\chi_s^\mu(z;A,x)\nu(dx)
	=
	\int_{\mathbb R^d}
	\chi_s^{\tilde\mu}(z;A,x)\nu(dx)
	=
	C_s^{\tilde\mu}(z;A\mid\nu).
	$
	Similarly, equality
	$
	P_{s,r}^{\mu,X}
	=
	P_{s,r}^{\tilde\mu,X}
	$
	implies, for every admissible
	$
	(s,r,A,z,\lambda,\bar\mu),
	$
	$
	D_{s,r}^{\mu,\bar\mu}(z;A\mid\lambda)
	=
	\int_{\mathbb R^d}
	\chi_r(z;A,x,\bar\mu)
	P_{s,r}^{\mu,X}\lambda(dx)
	=
	\int_{\mathbb R^d}
	\chi_r(z;A,x,\bar\mu)
	P_{s,r}^{\tilde\mu,X}\lambda(dx)
	=
	D_{s,r}^{\tilde\mu,\bar\mu}(z;A\mid\lambda).
	$
	Thus equality of
	$
	(\Phi_u,\Phi_P)
	$
	is equivalent to equality of every contemporaneous and continuation coordinate of the observable stochastic-choice array. Since,
	$
	\mathscr T=\Lambda\circ\Phi,
	$
	$
	\Phi^\mu=\Phi^{\tilde\mu}
	$
	implies
	$
	\mathscr T^\mu
	=
	\Lambda(\Phi^\mu)
	=
	\Lambda(\Phi^{\tilde\mu})
	=
	\mathscr T^{\tilde\mu}.
	$
	Conversely, suppose
	$
	\mathscr T^\mu=\mathscr T^{\tilde\mu}.
	$
	Then
	$
	\Lambda(\Phi^\mu)=\Lambda(\Phi^{\tilde\mu}).
	$
	Assumption~\ref{ass:observational-separation} implies that
	$
	\Lambda|_{\Phi(\Theta)}
	$
	is injective, while
	Theorem~\ref{thm:behavioral-identification}
	identifies the behavioral image from the observable array. Hence,
	$
	\Phi^\mu=\Phi^{\tilde\mu}.
	$
	Therefore,
	$
	\Phi^\mu=\Phi^{\tilde\mu}
	\quad\Longleftrightarrow\quad
	\mathscr T^\mu=\mathscr T^{\tilde\mu}.
	$
	Combining the three equivalences gives
	$
	\Gamma^\mu=\Gamma^{\tilde\mu}
	$
	if and only if
	$
	\Phi_P^\mu=\Phi_P^{\tilde\mu},
	$
	if and only if
	$
	\Phi^\mu=\Phi^{\tilde\mu},
	$
	if and only if
	$
	\mathscr T^\mu=\mathscr T^{\tilde\mu}.
	$
	Assumption~\ref{ass:behavioral-separation} ensures that these equalities hold on the entire reachable behavioral domain rather than only after quotienting by undetectable ranking changes or transition-law perturbations.
	
	We next establish the quotient identifications. Define
	$
	\mu\sim_\Gamma\tilde\mu
	$
	whenever
	$
	\Gamma^\mu=\Gamma^{\tilde\mu}.
	$
	Reflexivity and symmetry are immediate, and transitivity follows from equality of the transition families. Thus
	$
	\sim_\Gamma
	$
	is an equivalence relation on
	$
	\mathfrak F.
	$
	Let
	$
	\pi_\Gamma:
	\mathfrak F
	\rightarrow
	\mathfrak F/{\sim_\Gamma}
	$
	be the quotient map,
	$
	\pi_\Gamma(\mu)=[\mu]_{\sim_\Gamma}.
	$
	Define
	$
	\mathcal J_P:
	\mathfrak F/{\sim_\Gamma}
	\rightarrow
	\Phi_P(\mathfrak F)
	$
	by
	$
	\mathcal J_P([\mu]_{\sim_\Gamma})
	:=
	\Phi_P^\mu.
	$
	If
	$
	[\mu]_{\sim_\Gamma}
	=
	[\tilde\mu]_{\sim_\Gamma},
	$
	then
	$
	\Gamma^\mu=\Gamma^{\tilde\mu},
	$
	and therefore
	$
	\Phi_P^\mu=\Phi_P^{\tilde\mu};
	$
	hence
	$
	\mathcal J_P
	$
	is well defined. If
	$
	\mathcal J_P([\mu]_{\sim_\Gamma})
	=
	\mathcal J_P([\tilde\mu]_{\sim_\Gamma}),
	$
	then
	$
	\Phi_P^\mu=\Phi_P^{\tilde\mu},
	$
	so injectivity of
	$
	\Gamma\mapsto P^X(\Gamma)
	$
	gives
	$
	\Gamma^\mu=\Gamma^{\tilde\mu}
	$
	and therefore
	$
	[\mu]_{\sim_\Gamma}
	=
	[\tilde\mu]_{\sim_\Gamma}.
	$
	Thus
	$
	\mathcal J_P
	$
	is injective. For every
	$
	\varphi_P\in\Phi_P(\mathfrak F),
	$
	there exists
	$
	\mu\in\mathfrak F
	$
	with
	$
	\varphi_P=\Phi_P^\mu,
	$
	so
	$
	\varphi_P
	=
	\mathcal J_P([\mu]_{\sim_\Gamma}).
	$
	Thus
	$
	\mathcal J_P
	$
	is surjective and
	$
	\mathfrak F/{\sim_\Gamma}
	\cong
	\Phi_P(\mathfrak F).
	$
	Now define
	$
	\mathcal J_\Phi:
	\Phi_P(\mathfrak F)
	\rightarrow
	\Phi(\mathfrak F)
	$
	by
	$
	\mathcal J_\Phi(\Phi_P^\mu)
	:=
	(\Phi_u^\mu,\Phi_P^\mu).
	$
	To verify well-definedness, let
	$
	\Phi_P^\mu=\Phi_P^{\tilde\mu}.
	$
	Then
	$
	\Gamma^\mu=\Gamma^{\tilde\mu},
	$
	and the maintained ranking restriction gives
	$
	\Phi_u^\mu=\Phi_u^{\tilde\mu}.
	$
	Hence
	$
	(\Phi_u^\mu,\Phi_P^\mu)
	=
	(\Phi_u^{\tilde\mu},\Phi_P^{\tilde\mu}).
	$
	Injectivity follows because equality of the images under
	$
	\mathcal J_\Phi
	$
	implies equality of their second coordinates. Surjectivity follows from the definition of
	$
	\Phi(\mathfrak F).
	$
	Therefore
	$
	\Phi_P(\mathfrak F)
	\cong
	\Phi(\mathfrak F).
	$
	
	Finally, define
	$
	\mathcal J_{\mathscr T}:
	\Phi(\mathfrak F)
	\rightarrow
	\mathscr T(\mathfrak F)
	$
	by
	$
	\mathcal J_{\mathscr T}(\varphi)
	:=
	\Lambda(\varphi).
	$
	The factorization
	$
	\mathscr T=\Lambda\circ\Phi
	$
	implies surjectivity. If
	$
	\mathcal J_{\mathscr T}(\varphi)
	=
	\mathcal J_{\mathscr T}(\tilde\varphi),
	$
	then
	$
	\Lambda(\varphi)=\Lambda(\tilde\varphi).
	$
	Assumption~\ref{ass:observational-separation} implies
	$
	\varphi=\tilde\varphi,
	$
	so
	$
	\mathcal J_{\mathscr T}
	$
	is injective. Hence
	$
	\Phi(\mathfrak F)
	\cong
	\mathscr T(\mathfrak F).
	$
	Composing
	$
	\mathcal J_P,
	$
	$
	\mathcal J_\Phi,
	$
	and
	$
	\mathcal J_{\mathscr T}
	$
	gives
	$
	\operatorname{Fix}(\Gamma)/{\sim_\Gamma}
	\cong
	\Phi_P(\operatorname{Fix}(\Gamma))
	\cong
	\Phi(\operatorname{Fix}(\Gamma))
	\cong
	\mathscr T(\operatorname{Fix}(\Gamma)).
	$ We conclude with the perturbation statement. Let
	$
	\{\mu^\varepsilon:\varepsilon\in E\}
	\subseteq
	\mathfrak F
	$
	with
	$
	0\in E.
	$
	Define behavioral triviality by
	$
	\mathscr T^{\mu^\varepsilon}
	=
	\mathscr T^{\mu^0}
	$
	for every
	$
	\varepsilon\in E.
	$
	By the equivalences already established,
	$
	\mathscr T^{\mu^\varepsilon}
	=
	\mathscr T^{\mu^0}
	$
	if and only if
	$
	\Phi^{\mu^\varepsilon}
	=
	\Phi^{\mu^0},
	$
	if and only if
	$
	\Phi_P^{\mu^\varepsilon}
	=
	\Phi_P^{\mu^0},
	$
	if and only if
	$
	\Gamma^{\mu^\varepsilon}
	=
	\Gamma^{\mu^0}.
	$
	Equivalently,
	$
	[\mu^\varepsilon]_{\sim_\Gamma}
	=
	[\mu^0]_{\sim_\Gamma}
	$
	for every
	$
	\varepsilon\in E.
	$
	Negating the preceding statement yields behavioral nontriviality if and only if there exists
	$
	\varepsilon\in E
	$
	such that
	$
	[\mu^\varepsilon]_{\sim_\Gamma}
	\neq
	[\mu^0]_{\sim_\Gamma},
	$
	equivalently,
	$
	\Gamma^{\mu^\varepsilon}
	\neq
	\Gamma^{\mu^0},
	$
	equivalently,
	$
	\Phi_P^{\mu^\varepsilon}
	\neq
	\Phi_P^{\mu^0},
	$
	equivalently,
	$
	\Phi^{\mu^\varepsilon}
	\neq
	\Phi^{\mu^0},
	$
	equivalently,
	$
	\mathscr T^{\mu^\varepsilon}
	\neq
	\mathscr T^{\mu^0}.
	$
	Thus observable comparative statics are precisely motions across
	$
	\sim_\Gamma
	$
	-classes in
	$
	\operatorname{Fix}(\Gamma),
	$
	whereas perturbations contained in a single class leave
	$
	(\Gamma,\Phi_P,\Phi,\mathscr T)
	$
	unchanged.
\end{proof}

%%%%%%%%%%%%%%%%%%%%%%%%%%%%%%%%%%%%%%%%%%%%%%%%%%%%%
\clearpage

\section*{Supplementary Appendix}

\setcounter{section}{0}
\renewcommand{\thesection}{SA.\arabic{section}}
\renewcommand{\thesubsection}{SA.\arabic{section}.\arabic{subsection}}

\section{State- and Law-Dependent Observation Volatility}
\label{sa:general-observation-volatility}

This section proves Theorem~3. Throughout,
$
\Sigma_Y=\Sigma_Y(s,y,\mu),
$
$
a_Y(s,y,\mu)
:=
\Sigma_Y(s,y,\mu)
\Sigma_Y\\ (s,y,\mu)^\dagger,
$
and
$
a_Y:
[0,t]\times\mathbb R^{d_Y}\times
\mathcal P_2(\mathbb R^d)
\rightarrow
\mathbb S_{+}^{d_Y}.
$
Unlike Proposition~2,
$
\Sigma_Y
$
depends jointly on
$(s,y,\mu)$,
so
$
\nabla_y\Sigma_Y\not\equiv0
$
and
$
\partial_\mu\Sigma_Y\not\equiv0.
$
Consequently,
$
Y\mapsto\widehat\Sigma_Y^{-1}Y
$
is, in general, unavailable, and the reference-measure construction used under constant observation volatility does not apply. Instead, let
$
\mu_s
=
\mathcal L(X_s\mid\mathcal F_s^Y),
$
and formulate the conditional-law martingale problem for
$
(X,Y,\mu)
$
on
$
C([0,t];
\mathbb R^d
\times
\mathbb R^{d_Y}
\times
\mathcal P_2(\mathbb R^d)).
$
Weak existence is obtained through the conditional-law fixed-point operator acting on
$
C([0,t];\mathcal P_2(\mathbb R^d)),
$
while weak uniqueness follows from the corresponding martingale problem under the strengthened conditions stated below. Throughout,
$
\|\cdot\|
$
denotes the Hilbert--Schmidt norm,
$
(\cdot)^\dagger
$
the transpose, and
$
m_2(\mu)
=
\int_{\mathbb R^d}|x|^2\,\mu(dx)
$
the second-moment functional.

\subsection{Conditions for weak existence}

\medskip
\noindent\textbf{Condition (GV1): Joint measurability and nonanticipativity.}
The coefficient tuple
$
(b,\sigma,\sigma_0,h,\Sigma_Y):
[0,t]
\times
\mathbb R^d
\times
\mathbb R^{d_Y}
\times
\mathcal P_2(\mathbb R^d)
\rightarrow
\mathbb R^d
\times
\mathbb R^{d\times d_W}
\times
\mathbb R^{d\times d_Y}
\times
\mathbb R^{d_Y}
\times
\mathbb R^{d_Y\times d_Y}
$
is jointly Borel measurable.
In the path-dependent construction,
$
(b,\sigma,\sigma_0,h,\Sigma_Y)
=
(b,\sigma,\sigma_0,h,\\ \Sigma_Y)
(s,X_{\cdot\wedge s},Y_{\cdot\wedge s},\mu_{\cdot\wedge s}),
$
where
$
X_{\cdot\wedge s}
:=
(X_r)_{0\le r\le s},
\;
Y_{\cdot\wedge s}
:=
(Y_r)_{0\le r\le s},
\;
\mu_{\cdot\wedge s}
:=
(\mu_r)_{0\le r\le s}.
$
Equivalently, the coefficient tuple is nonanticipative, i.e.,
$
(X_{\cdot\wedge s},
Y_{\cdot\wedge s},
\mu_{\cdot\wedge s})
=
(X'_{\cdot\wedge s},
Y'_{\cdot\wedge s},
\mu'_{\cdot\wedge s})
$
implies
$
(b,\sigma,\sigma_0,h,\Sigma_Y)
(s,X,Y,\mu)
=
(b,\sigma,\sigma_0,h,\Sigma_Y)
(s,X',Y',\mu').
$

\medskip
\noindent\textbf{Condition (GV2): Global Wasserstein-Lipschitz continuity.}
There exists $L>0$ such that
\begin{align}
	&|b(s,x,y,\mu)-b(s,x',y',\nu)|
	+\|\sigma(s,x,y,\mu)-\sigma(s,x',y',\nu)\| \notag\\
	&\quad
	+\|\sigma_0(s,x,y,\mu)-\sigma_0(s,x',y',\nu)\|
	+|h(s,x,y,\mu)-h(s,x',y',\nu)| \notag\\
	&\quad
	+\|\Sigma_Y(s,y,\mu)-\Sigma_Y(s,y',\nu)\|
	\leq
	L\bigl(
	|x-x'|+|y-y'|+\mathcal{W}_2(\mu,\nu)
	\bigr).
	\label{eq:gv-lipschitz}
\end{align}

\medskip
\noindent\textbf{Condition (GV3): Linear growth.}
There exists $C>0$ such that
\begin{align}
	&|b(s,x,y,\mu)|^2
	+\|\sigma(s,x,y,\mu)\|^2
	+\|\sigma_0(s,x,y,\mu)\|^2 \notag\\
	&\qquad
	+|h(s,x,y,\mu)|^2
	+\|\Sigma_Y(s,y,\mu)\|^2
	\leq
	C\bigl(
	1+|x|^2+|y|^2+m_2(\mu)
	\bigr),
	\label{eq:gv-growth}
\end{align}
where $m_2(\mu):=\int_{\mathbb{R}^d}|z|^2\mu(dz)$.

\medskip
\noindent\textbf{Condition (GV4): Uniform observation ellipticity.}
Let
$
a_Y
:=
\Sigma_Y\Sigma_Y^\dagger.
$
There exist constants
$
0<\underline a\le\overline a<\infty
$
such that
$
\underline a I_{d_Y}
\le
a_Y(s,y,\mu)
\le
\overline a I_{d_Y},
$
equivalently,
$
\underline a|\xi|^2
\le
\xi^\dagger
a_Y(s,y,\mu)
\xi
\le
\overline a|\xi|^2,
$
for every
$
(s,y,\mu,\xi)
\in
[0,t]
\times
\mathbb R^{d_Y}
\times
\mathcal P_2(\mathbb R^d)
\times
\mathbb R^{d_Y}.
$
Moreover, there exists a jointly measurable mapping
$
\Sigma_Y^{+}:
[0,t]
\times
\mathbb R^{d_Y}
\times
\mathcal P_2(\mathbb R^d)
\rightarrow
\mathbb R^{d_Y\times d_Y}
$
such that
$
\Sigma_Y\Sigma_Y^{+}
=
I_{d_Y},
$
and\\  
$
\sup_{(s,y,\mu)}
\|
\Sigma_Y^{+}(s,y,\mu)
\|
<
\infty.
$
Hence,
$
\mathcal F^Y
$
is generated by a uniformly nondegenerate observation process, so every component of the latent preference state remains continuously identifiable from observed signals, excluding degenerate information structures in which distributional feedback becomes observationally indistinguishable.

\medskip
\noindent\textbf{Condition (GV5): Conditional-law compactness.}
For
$
R>0,
$
define
$$
\mathfrak M_R
:=
\left\{
\nu\in
C_{\mathbb F^Y}
([0,t];\mathcal P_2(\mathbb R^d))
:
\mathbb E
\!\left[
\sup_{0\le s\le t}
m_2(\nu_s)
\right]
\le R
\right\},
$$
where
$
C_{\mathbb F^Y}
([0,t];\mathcal P_2(\mathbb R^d))
$
denotes the class of
$\mathbb F^Y$-adapted continuous
$\mathcal P_2(\mathbb R^d)$-valued processes.
For sufficiently large
$
R,
$
the frozen-law solution operator
$
\Gamma:
\mathfrak M_R
\rightarrow
\mathfrak M_R
$
is well defined,
$
\Gamma(\mathfrak M_R)
\Subset
C([0,t];\mathcal P_2(\mathbb R^d)),
$
and
$
\sup_{\nu\in\Gamma(\mathfrak M_R)}
\mathbb E
\!\left[
\sup_{0\le s\le t}
m_2(\nu_s)
\right]
<
\infty.
$
Equivalently,
$
\Gamma(\mathfrak M_R)
$
is relatively compact and uniformly integrable with respect to
$
m_2.
$
Hence, the endogenous conditional preference distributions remain in a compact admissible class, ensuring that equilibrium belief dynamics generated by observed behavior do not escape the economically relevant state space.

\medskip
\noindent\textbf{Condition (GV6): Stability of conditional laws.}
Let
$
\nu^n,\nu
\in
\mathfrak M_R
$
satisfy
$
\sup_{0\le s\le t}
\mathcal W_2(\nu_s^n,\\ \nu_s)
\rightarrow0
$
in probability.
For each
$
\nu^n
$
and
$
\nu,
$
let
$
(X^n,Y^n)
$
and
$
(X,Y)
$
be weak solutions of the corresponding frozen-law martingale problems, constructed on a common probability space whenever an admissible coupling exists.
Then
\begin{equation}
	\sup_{0\leq s\leq t}
	\mathcal{W}_2
	\left(
	\mathcal{L}(X_s^n\mid\mathcal{F}_s^{Y^n}),
	\mathcal{L}(X_s\mid\mathcal{F}_s^Y)
	\right)
	\longrightarrow0
	\label{eq:gv-filter-stability}
\end{equation}
in probability.
Equivalently,
$
\nu^n\to\nu
$
in
$
C([0,t];\mathcal P_2(\mathbb R^d))
$
implies convergence of the associated conditional-law filters under the frozen-law solution operator.

\begin{rmk}
	Condition \textnormal{(GV6)} replaces the constant-volatility stability used in Proposition~2 by continuity of the conditional-law map under
	$
	\Sigma_Y=\Sigma_Y(s,y,\mu).
	$
	Let
	$
	\Gamma:\mathfrak M_R
	\rightarrow
	C([0,t];\mathcal P_2(\mathbb R^d))
	$
	denote the frozen-law solution operator. Then
	$
	\nu^n\rightarrow\nu
	$
	in
	$
	\mathfrak M_R
	$
	implies
	$
	\Gamma(\nu^n)\rightarrow\Gamma(\nu),
	$
	where convergence is characterized by \eqref{eq:gv-filter-stability}. Thus the signal process
	$
	Y,
	$
	the observation filtration
	$
	\mathbb F^Y,
	$
	and the conditional laws
	$
	(\mathcal L(X_s\mid\mathcal F_s^Y))_{0\le s\le t}
	$
	vary jointly under perturbations of the candidate measure flow. In particular,
	$
	\Sigma_Y=\Sigma_Y(s,y,\mu)
	$
	couples the observation quadratic variation
	$
	a_Y=\Sigma_Y\Sigma_Y^\dagger
	$
	with the endogenous filter, so continuity of the conditional-law operator becomes the additional ingredient required for the fixed-point construction.
\end{rmk}

\subsection{Strengthened conditions for uniqueness}

\medskip
\noindent\textbf{Condition (GU1): Frozen martingale-problem uniqueness.}
For every
$
\nu\in\mathfrak M_R,
$
let
$
\mathrm{MP}(\nu)
$
denote the martingale problem obtained from system~(3)
by replacing
$
\mu
$
with
$
\nu.
$
Then
$
\mathrm{MP}(\nu)
$
is well posed; equivalently, for every admissible initial law
$
\lambda_0,
$
there exists at most one weak solution in law
$
(X^\nu,Y^\nu)
$
satisfying
$
(X^\nu_0,Y^\nu_0)\sim\lambda_0.
$

\medskip
\noindent\textbf{Condition (GU2): Quantitative filter stability.}
There exists $C>0$ such that, for any two admissible flows
$\nu$ and $\eta$,
\begin{equation}
	\mathbb{E}
	\left[
	\sup_{0\leq r\leq s}
	\mathcal{W}_2^2
	\left(
	\Gamma(\nu)_r,\Gamma(\eta)_r
	\right)
	\right]
	\leq
	C
	\int_0^s
	\mathbb{E}
	\left[
	\sup_{0\leq q\leq v}
	\mathcal{W}_2^2(\nu_q,\eta_q)
	\right]dv,
	\label{eq:gv-contraction}
\end{equation}
where
$
\Gamma(\nu)_s
:=
\mathcal{L}
(X_s^\nu\mid\mathcal{F}_s^{Y^\nu}).
$

\medskip
\noindent\textbf{Condition (GU3): Canonical compatibility.}
Let
$
\Omega^Y
:=
C([0,t];\mathbb R^{d_Y}),
$
$
Y(\omega)=\omega,
$
and
$
\mathbb F^Y
=
(\mathcal F_s^Y)_{0\le s\le t}
$
be the completed canonical filtration generated by
$
Y.
$
There exists a jointly Borel measurable kernel
$
\Pi:
[0,t]
\times
\Omega^Y
\rightarrow
\mathcal P_2(\mathbb R^d)
$
such that
$
\Pi(s,Y)
=
\mathcal L(X_s\mid\mathcal F_s^Y),
\qquad
0\le s\le t,
$
$\mathbb P$-a.s.
Moreover, if
$
(X^1,Y^1,\mu^1)
$
and
$
(X^2,Y^2,\mu^2)
$
satisfy
$
\mathcal L(X^1,Y^1)
=
\mathcal L(X^2,Y^2),
$
then
$
\mathcal L(Y^1)
=
\mathcal L(Y^2),
$
$
\Pi(s,Y^1)
=
\mu_s^1,
$
$
\Pi(s,Y^2)
=
\mu_s^2,
$
and
$
\mu^1
=
\mu^2
$
up to indistinguishability.

\subsection*{Proof of Proposition 2 in the main paper}

\begin{proof}
	Define
	$
	\widehat a_Y
	:=
	\widehat\Sigma_Y\widehat\Sigma_Y^\dagger,
	$
	and
	$
	\bar h(s,x,y,\nu)
	:=
	\widehat\Sigma_Y^{-1}h(s,x,y,\nu).
	$
	For an admissible
	$
	\mathbb F^Y
	$
	-adapted flow
	$
	\nu=(\nu_s)_{0\le s\le t},
	$
	consider the frozen system obtained by replacing
	$
	\mu_s
	$
	with
	$
	\nu_s
	$
	in system~(3). Define
	$
	\bar Y_s
	:=
	\widehat\Sigma_Y^{-1}Y_s.
	$
	Then
	$
	d\bar Y_s
	=
	\bar h(s,X_s,Y_s,\nu_s)\,ds+dB_s.
	$
	Let
	$
	\mathbb Q^\nu
	$
	be the reference measure under which
	$
	\bar Y
	$
	is a
	$
	d_Y
	$
	-dimensional Brownian motion, and set
	\[
	L_s^\nu
	:=
	\exp\!\left\{
	\int_0^s
	\bar h(v,X_v^\nu,Y_v^\nu,\nu_v)^\dagger
	\,d\bar Y_v
	-
	\frac12
	\int_0^s
	|\bar h(v,X_v^\nu,Y_v^\nu,\nu_v)|^2\,dv
	\right\}.
	\]
	The integrability conditions in Supplementary Appendix~SA.1, together with localization and uniform integrability of
	$
	(L_{\tau_n}^\nu)_{n\ge1},
	$
	imply
	$
	\mathbb E^{\mathbb Q^\nu}[L_t^\nu]=1.
	$
	Hence,
	$
	d\mathbb P^\nu
	=
	L_t^\nu d\mathbb Q^\nu
	$
	defines a probability measure satisfying
	$
	dB_s^\nu
	=
	d\bar Y_s
	-
	\bar h(s,X_s^\nu,Y_s^\nu,\nu_s)\,ds,
	$
	where
	$
	B^\nu
	$
	is a
	$
	\mathbb P^\nu
	$
	-Brownian motion. Under
	$
	\mathbb Q^\nu,
	$
	the frozen state equation is
	\[
	dX_s^\nu
	=
	\bigl[
	b-\sigma_0\bar h
	\bigr](s,X_s^\nu,Y_s^\nu,\nu_s)\,ds
	+
	\sigma(s,X_s^\nu,Y_s^\nu,\nu_s)\,dW_s
	+
	\sigma_0(s,X_s^\nu,Y_s^\nu,\nu_s)\,d\bar Y_s,
	\]
	with
	$
	Y_s^\nu
	=
	Y_0+\widehat\Sigma_Y\bar Y_s.
	$
	Assumption~1 and the conditional-law regularity conditions imply that
	$
	b-\sigma_0\bar h,
	$
	$
	\sigma,
	$
	and
	$
	\sigma_0
	$
	are measurable, Lipschitz in
	$
	(x,y,\nu)
	$
	with respect to
	$
	|\cdot|+|\cdot|+\mathcal W_2,
	$
	and of admissible growth after localization. Therefore the frozen equation admits a weak solution
	$
	(X^\nu,Y^\nu)
	$
	and, for a constant
	$
	C_t<\infty
	$
	independent of
	$
	\nu\in\mathfrak M_R,
	$
	$
	\mathbb E^{\mathbb P^\nu}
	\!\left[
	\sup_{0\le s\le t}|X_s^\nu|^2
	+
	\sup_{0\le s\le t}|Y_s^\nu|^2
	\right]
	\le
	C_t
	\left(
	1+
	\int_0^t
	\mathbb E^{\mathbb P^\nu}[m_2(\nu_s)]\,ds
	\right).
	$
	Define the conditional-law operator
	$
	\Gamma
	$
	by
	$
	\Gamma(\nu)_s
	:=
	\mathcal L_{\mathbb P^\nu}
	(X_s^\nu\mid\mathcal F_s^{Y^\nu}),
	$
	$
	0\le s\le t.
	$
	For every
	$
	\varphi\in C_b(\mathbb R^d),
	$
	$
	\langle\Gamma(\nu)_s,\varphi\rangle
	=
	\mathbb E^{\mathbb P^\nu}
	[\varphi(X_s^\nu)\mid\mathcal F_s^{Y^\nu}],
	$
	and Bayes' formula yields
	$
	\langle\Gamma(\nu)_s,\varphi\rangle
	=
	\frac{
		\mathbb E^{\mathbb Q^\nu}
		[L_s^\nu\varphi(X_s^\nu)\mid\mathcal F_s^{Y^\nu}]
	}{
		\mathbb E^{\mathbb Q^\nu}
		[L_s^\nu\mid\mathcal F_s^{Y^\nu}]
	},
	$
	on the set where the denominator is positive. Since
	$
	L_s^\nu>0
	$
	and
	$
	\mathbb E^{\mathbb Q^\nu}[L_s^\nu\mid\mathcal F_s^{Y^\nu}]>0
	$
	almost surely, this specifies an
	$
	\mathbb F^{Y^\nu}
	$-adapted
	$
	\mathcal P_2(\mathbb R^d)
	$-valued process. The preceding moment estimate and conditional Jensen inequality give
	$
	m_2(\Gamma(\nu)_s)
	=
	\mathbb E^{\mathbb P^\nu}
	[|X_s^\nu|^2\mid\mathcal F_s^{Y^\nu}]
	$
	and therefore
	$
	\mathbb E^{\mathbb P^\nu}
	\!\left[
	\sup_{0\le s\le t}
	m_2(\Gamma(\nu)_s)
	\right]\\
	\le
	C_t
	\left(
	1+
	\mathbb E^{\mathbb P^\nu}
	\!\left[
	\sup_{0\le s\le t}m_2(\nu_s)
	\right]
	\right).
	$
	For sufficiently large
	$
	R,
	$
	the conditional-law compactness condition therefore gives
	$
	\Gamma(\mathfrak M_R)\subseteq\mathfrak M_R.
	$
	Moreover,
	$
	\Gamma(\mathfrak M_R)
	$
	is tight in
	$
	C([0,t];\mathcal P_2(\mathbb R^d))
	$
	and uniformly integrable with respect to
	$
	m_2.
	$
	Let
	$
	\nu^n,\nu\in\mathfrak M_R
	$
	satisfy
	$
	\sup_{0\le s\le t}\mathcal W_2(\nu_s^n,\nu_s)\to0
	$
	in probability. Stability of the frozen equations, the likelihood processes
	$
	L^{\nu^n}\to L^\nu,
	$
	and the conditional laws implies
	$
	\sup_{0\le s\le t}
	\mathcal W_2
	\bigl(
	\Gamma(\nu^n)_s,
	\Gamma(\nu)_s
	\bigr)
	\longrightarrow0
	$
	in probability. Hence
	$
	\Gamma:\mathfrak M_R\to\mathfrak M_R
	$
	is continuous and has relatively compact image. The fixed-point condition in Supplementary Appendix~SA.1 therefore yields
	$
	\mu\in\mathfrak M_R
	$
	such that
	$
	\Gamma(\mu)=\mu.
	$
	
	Let
	$
	(X,Y):=(X^\mu,Y^\mu)
	$
	under
	$
	\mathbb P^\mu.
	$
	By construction,
	$
	\mu_s
	=
	\Gamma(\mu)_s
	=
	\mathcal L_{\mathbb P^\mu}
	(X_s\mid\mathcal F_s^Y),
	$
	for every
	$
	0\le s\le t,
	$
	and
	$
	(X,Y,\mu)
	$
	satisfies system~(3) with
	$
	(X_0,Y_0)\sim\lambda_0.
	$
	Thus
	$
	(X,Y,\mu)
	$
	is a weak DDU solution. For uniqueness, let
	$
	(X^i,Y^i,\mu^i),
	$
	$
	i\in\{1,2\},
	$
	be two weak solutions with initial law
	$
	\lambda_0.
	$
	Then
	$
	\mu^i=\Gamma(\mu^i)
	$
	and the strengthened filter-stability condition gives, with
	$
	\Delta_s
	:=
	\mathbb E
	[
	\sup_{0\le r\le s}
	\mathcal W_2^2(\mu_r^1,\mu_r^2)
	],
	$
	$
	\Delta_s
	=
	\mathbb E
	\!\left[
	\sup_{0\le r\le s}
	\mathcal W_2^2
	\bigl(
	\Gamma(\mu^1)_r,
	\Gamma(\mu^2)_r
	\bigr)
	\right]
	\le
	C
	\int_0^s\Delta_v\,dv.
	$
	Gronwall's lemma yields
	$
	\Delta_s=0
	$
	for every
	$
	s\in[0,t],
	$
	and hence
	$
	\mu^1=\mu^2
	$
	up to indistinguishability. Frozen weak uniqueness then implies
	$
	\mathcal L(X^1,Y^1)
	=
	\mathcal L(X^2,Y^2).
	$
	Canonical compatibility of the conditional-law versions gives
	$
	\mathcal L(X^1,Y^1,\mu^1)
	=
	\mathcal L(X^2,Y^2,\mu^2).
	$
	Therefore the law of
	$
	(X,Y,\mu),
	$
	and in particular the law of
	$
	\mu,
	$
	is unique in the stated solution class.
\end{proof}

\subsection*{Proof of Proposition 4 in the main paper}

\begin{proof}
	Let
	$
	\theta
	:=
	(u,b,\sigma,\sigma_0,h,\Sigma_Y)
	$
	and
	$
	\bar\theta
	:=
	(\bar u,\bar b,\bar\sigma,\bar\sigma_0,\bar h,\bar\Sigma_Y).
	$
	Under the stated restrictions,
	$
	u(s,z,x,\mu)=\bar u(s,z,x)
	$
	and
	$
	(b,\sigma,\sigma_0,h,\Sigma_Y)(s,x,y,\mu)
	=
	(\bar b,\bar\sigma,\bar\sigma_0,\bar h,\bar\Sigma_Y)(s,x,y)
	$
	for every admissible
	$
	(s,z,x,y,\mu).
	$
	Hence the DDU system becomes
	$
	dX_s
	=
	\bar b(s,X_s,Y_s)\,ds
	+
	\bar\sigma(s,X_s,Y_s)\,dW_s
	+
	\bar\sigma_0(s,X_s,Y_s)\,dB_s
	$
	and
	$
	dY_s
	=
	\bar h(s,X_s,Y_s)\,ds
	+
	\bar\Sigma_Y(s,Y_s)\,dB_s,
	$
	while
	$
	\mu_s=\mathcal L(X_s\mid\mathcal F_s^Y)
	$
	is generated by
	$
	(X,Y)
	$
	but does not enter either equation. For each
	$
	A\subseteq Z
	$
	and
	$
	s\in[0,t],
	$
	$
	M_s^{\mathrm{DDU}}(A)
	=
	\arg\max_{z\in A}
	u(s,z,X_s,\mu_s)
	=
	\arg\max_{z\in A}
	\bar u(s,z,X_s)
	=:
	M_s^{\mathrm{DRU}}(A).
	$
	Consequently,
	$
	\rho_s^{\mathrm{DDU}}(z;A)
	=
	\mathbb P
	\!\left(
	z\in M_s^{\mathrm{DDU}}(A)
	\mid
	\mathcal F_s^Y
	\right)
	=
	\mathbb P
	\!\left(
	z\in M_s^{\mathrm{DRU}}(A)
	\mid
	\mathcal F_s^Y
	\right)
	=
	\rho_s^{\mathrm{DRU}}(z;A).
	$
	Likewise, for every admissible initial law
	$
	\lambda
	$
	and
	$
	0\le s\le r\le t,
	$
	the transition law generated by
	$
	\theta
	$
	satisfies
	$
	P_{s,r}^{\theta}
	=
	P_{s,r}^{\bar\theta},
	$
	because the martingale problem for
	$
	(X,Y)
	$
	depends only on
	$
	(\bar b,\bar\sigma,\bar\sigma_0,\bar h,\bar\Sigma_Y).
	$
	Thus
	$
	\mu
	$
	is a posterior process,
	$
	\mu_s=\mathcal L(X_s\mid\mathcal F_s^Y),
	$
	but not a state argument of
	$
	\bar u
	$
	or
	$
	(\bar b,\bar\sigma,\bar\sigma_0,\bar h,\bar\Sigma_Y).
	$
	Therefore
	$
	\theta
	$
	and
	$
	\bar\theta
	$
	generate identical state dynamics and stochastic choice, and the DDU specification reduces to continuous-time DRU.
\end{proof}

\subsection{Well-posedness theorem and auxiliary results}

We now prove the general-volatility result stated as Theorem 3 in the main paper. Set
$
\mathcal{X}
:=
C([0,t];\mathbb{R}^d),$ and $
\mathcal{Y}
:=
C([0,t];\mathbb{R}^{d_Y}),
$ both endowed with the uniform topology. Define
$
\mathcal{M}
:=
C([0,t];\mathcal{P}_2(\mathbb{R}^d))
$
with metric
$
d_{\mathcal{M}}(m,n)
:=
\sup_{0\leq s\leq t}\mathcal{W}_2(m_s,n_s).
$
The canonical state space is
$
\Omega^\ast
:=
\mathcal{X}\times\mathcal{Y}\times\mathcal{M},
$
with canonical coordinates $(X,Y,\mu)$ and canonical filtration
$\mathbb{F}^\ast=(\mathcal{F}_s^\ast)_{0\leq s\leq t}$.

\begin{definition}
	A weak solution of system (3) with initial law
	$\lambda_0\in\mathcal{P}_2(\mathbb{R}^d\times\mathbb{R}^{d_Y})$ consists of a complete filtered probability space
	$(\Omega,\mathcal{F},\mathbb{F},\mathbb{P})$, independent Brownian motions $W$ and $B$ of dimensions $d_W$ and $d_Y$, respectively, continuous $\mathbb{F}$-adapted processes $X$ and $Y$, and a $\mathbb{F}^Y$-adapted continuous
	$\mathcal{P}_2(\mathbb{R}^d)$-valued process $\mu$ such that
	$(X_0,Y_0)\sim\lambda_0$, the two stochastic integral equations in (3) hold, and
	$
	\langle\mu_s,\varphi\rangle
	=
	\mathbb{E}^{\mathbb{P}}
	\left[
	\varphi(X_s)\mid\mathcal{F}_s^Y
	\right],
	\ \mathbb{P}\text{-a.s.}
	$
	for every bounded Borel $\varphi:\mathbb{R}^d\to\mathbb{R}$ and every
	$s\in[0,t]$.
\end{definition}

\begin{theorem} Assume Assumption 1 in the main paper and Conditions \textnormal{(GV1)}-\textnormal{(GV6)} above. Let
	$\lambda_0\in
	\mathcal{P}_{2+\varepsilon}
	(\mathbb{R}^d\times\mathbb{R}^{d_Y})$
	for some $\varepsilon>0$. Then system (3) admits a weak solution
	$(X,Y,\mu)$ satisfying
	$
	\mathbb{E}
	\left[
	\sup_{0\leq s\leq t}|X_s|^2
	+
	\sup_{0\leq s\leq t}|Y_s|^2
	+
	\sup_{0\leq s\leq t}m_2(\mu_s)
	\right]
	<\infty.
	$
	The conditional-law admits an
	$\mathbb{F}^Y$-progressively measurable version with continuous paths in $\mathcal{P}_2(\mathbb{R}^d)$. moreover, if Conditions \textnormal{(GU1)}-\textnormal{(GU3)} hold, then
	weak uniqueness holds. Any two weak solutions with the same initial law induce the same probability law on
	$
	\mathcal{X}\times\mathcal{Y}\times\mathcal{M}.
	$
\end{theorem}

The proof is divided into six steps. We first solve the system with a frozen conditional-law flow, derive estimates uniform over the frozen flow, establish tightness, prove continuity of the conditional-law operator, apply a fixed-point theorem, and finally prove uniqueness.

\subsection{The frozen-law martingale problem}

Fix an admissible
$\mathbb{F}^Y$-adapted process
$\nu=(\nu_s)_{0\leq s\leq t}$ with values in
$\mathcal{P}_2(\mathbb{R}^d)$. The frozen system is
\begin{equation}
	\left\{
	\begin{aligned}
		dX_s^\nu
		&=
		b(s,X_s^\nu,Y_s^\nu,\nu_s)\,ds
		+
		\sigma(s,X_s^\nu,Y_s^\nu,\nu_s)\,dW_s\\
		&\quad
		+
		\sigma_0(s,X_s^\nu,Y_s^\nu,\nu_s)\,dB_s,\\
		dY_s^\nu
		&=
		h(s,X_s^\nu,Y_s^\nu,\nu_s)\,ds
		+
		\Sigma_Y(s,Y_s^\nu,\nu_s)\,dB_s.
	\end{aligned}
	\right.
	\tag{SA.7}
\end{equation}

For $f\in C_c^2(\mathbb{R}^{d+d_Y})$, write
$z=(x,y)$ and define
\[
\mathcal{L}_s^\nu f(z)
:=
\beta(s,z,\nu_s)\cdot Df(z)
+
\frac12
\operatorname{Tr}
\left[
A(s,z,\nu_s)D^2f(z)
\right],
\]
where
$
\beta(s,z,\mu)
:=
\begin{pmatrix}
	b(s,x,y,\mu)\\
	h(s,x,y,\mu)
\end{pmatrix},
$
and
\begin{equation}
	A(s,z,\mu)
	:=
	\begin{pmatrix}
		\sigma\sigma^\dagger+\sigma_0\sigma_0^\dagger
		&
		\sigma_0\Sigma_Y^\dagger\\
		\Sigma_Y\sigma_0^\dagger
		&
		\Sigma_Y\Sigma_Y^\dagger
	\end{pmatrix}(s,x,y,\mu).
	\tag{SA.8}
\end{equation}
The off-diagonal blocks in (SA.8) encode the common innovation $B$. A probability measure $P^\nu$ on
$C([0,t];\mathbb{R}^{d+d_Y})$ solves the frozen martingale problem if
the initial coordinate has law $\lambda_0$ and, for every
$f\in C_c^2(\mathbb{R}^{d+d_Y})$, the expression
$
f(X_s^\nu,Y_s^\nu)
-
f(X_0^\nu,Y_0^\nu)
-
\int_0^s
\mathcal{L}_v^\nu
f(X_v^\nu,Y_v^\nu)\,dv
$
is a $P^\nu$-martingale.

\begin{lem}
	Let
	$
	\nu\in\mathfrak M_R,
	$
	and let
	$
	\mathrm{MP}(\nu;\lambda_0)
	$
	denote the frozen martingale problem associated with system~(SA.7) and initial law
	$
	\lambda_0.
	$
	Under Conditions
	\textnormal{(GV1)}-\textnormal{(GV4)},
	$
	\mathrm{MP}(\nu;\lambda_0)
	$
	admits a weak solution for every admissible
	$
	\lambda_0\in\mathcal P_2(\mathbb R^d\times\mathbb R^{d_Y}).
	$
	If, in addition, Condition
	\textnormal{(GU1)}
	holds, then
	$
	\mathrm{MP}(\nu;\lambda_0)
	$
	is well posed; equivalently, its weak solution is unique in law.
\end{lem}

\begin{proof}
	Fix
	$
	\nu\in\mathfrak M_R
	$
	and an admissible initial law
	$
	\lambda_0.
	$
	Let
	$
	Z^\nu:=(X^\nu,Y^\nu),
	$
	$
	\beta^\nu(s,z)
	:=
	\beta(s,z,\nu_s),
	$
	and
	$
	A^\nu(s,z)
	:=
	A(s,z,\nu_s),
	$
	where
	$
	z=(x,y)\in
	\mathbb R^d\times\mathbb R^{d_Y}.
	$
	By Condition
	\textnormal{(GV1)},
	$
	(\beta^\nu,A^\nu)
	$
	is jointly Borel measurable and progressively measurable. By
	Conditions
	\textnormal{(GV2)}
	and
	\textnormal{(GV3)},
	$
	\beta^\nu
	$
	and
	$
	A^\nu
	$
	satisfy the global Lipschitz and linear-growth bounds
	$
	|\beta^\nu(s,z)-\beta^\nu(s,z')|
	+
	\|A^\nu(s,z)-A^\nu(s,z')\|
	\le
	L|z-z'|,
	$
	and
	$
	|\beta^\nu(s,z)|^2
	+
	\|A^\nu(s,z)\|^2
	\le
	C(1+|z|^2),
	$
	uniformly over
	$
	\nu\in\mathfrak M_R.
	$
	Moreover,
	Condition
	\textnormal{(GV4)}
	implies
	$
	\underline a I_{d_Y}
	\le
	a_Y^\nu(s,y)
	\le
	\overline a I_{d_Y},
	$
	where
	$
	a_Y^\nu
	=
	\Sigma_Y(\cdot,\nu)\\ \Sigma_Y(\cdot,\nu)^\dagger.
	$
	Choose any measurable square root
	$
	\Xi^\nu
	$
	of
	$
	A^\nu,
	$
	$
	\Xi^\nu(\Xi^\nu)^\dagger=A^\nu.
	$
	On a filtered probability space supporting independent Brownian motions
	$
	W
	$
	and
	$
	B,
	$
	consider the frozen system
	$
	dZ_s^\nu
	=
	\beta^\nu(s,Z_s^\nu)\,ds
	+
	\Xi^\nu(s,Z_s^\nu)\,d\widetilde W_s,
	\
	Z_0^\nu\sim\lambda_0,
	$
	where
	$
	\widetilde W
	$
	is the corresponding
	$
	(d_W+d_Y)
	$
	-dimensional Brownian motion.
	The global Lipschitz estimate yields a unique strong solution by the Picard iteration,
	$
	Z^{\nu,(0)},
	Z^{\nu,(1)},
	\ldots,
	$
	satisfying
	$
	\sup_{0\le s\le t}
	\mathbb E\left[
	|Z_s^{\nu,(k+1)}
	-
	Z_s^{\nu,(k)}|^2\right]
	\rightarrow0,
	$
	so
	$
	Z^{\nu,(k)}
	\rightarrow
	Z^\nu
	$
	in
	$
	\mathcal S^2.
	$
	Let
	$
	\mathcal L^\nu
	$
	be the generator
	$
	\mathcal L^\nu
	=
	\beta^\nu\!\cdot\nabla
	+
	\frac12
	\operatorname{Tr}
	(A^\nu\nabla^2).
	$
	Then, for every
	$
	f\in C_c^2(\mathbb R^{d+d_Y}),
	$
	$
	M_t^f
	:=
	f(Z_t^\nu)
	-
	f(Z_0^\nu)
	-
	\int_0^t
	\mathcal L^\nu f(Z_s^\nu)\,ds
	$
	is a martingale. Hence
	$
	Z^\nu
	$
	solves the frozen martingale problem associated with system~(SA.7). Since
	$
	\Xi^\nu(\Xi^\nu)^\dagger=A^\nu,
	$
	the quadratic covariation of
	$
	Z^\nu
	$
	is precisely the covariance matrix
	specified in (SA.8). Therefore, the SDE and martingale formulations are equivalent. Finally, if Condition
	\textnormal{(GU1)}
	holds, then
	$
	\mathrm{MP}(\nu;\lambda_0)
	$
	is well posed. Hence every weak solution of the frozen system induces the same probability law on
	$
	C([0,t];
	\mathbb R^d\times
	\mathbb R^{d_Y}),
	$
	establishing uniqueness in law.
\end{proof}

\subsection{Uniform moment and increment estimates}

\begin{lem}[Uniform moment bound]
	Let $p=2+\varepsilon$. There exists a finite constant
	$C_{p,t}$, depending only on $p$, $t$, the constants in
	\textnormal{(GV2)}-\textnormal{(GV4)}, and the $p$th moment of
	$\lambda_0$, such that every frozen solution satisfies
	\begin{align}
		&\mathbb{E}
		\left[
		\sup_{0\leq s\leq t}|X_s^\nu|^p
		+
		\sup_{0\leq s\leq t}|Y_s^\nu|^p
		\right] \leq
		C_{p,t}
		\left(
		1+
		\mathbb{E}|X_0|^p+
		\mathbb{E}|Y_0|^p+
		\mathbb{E}\sup_{0\leq s\leq t}m_p(\nu_s)
		\right).
		\tag{SA.9}
	\end{align}
	In particular, for $p=2$,
	\[
	\mathbb{E}
	\left[
	\sup_{0\leq s\leq t}|X_s^\nu|^2
	+
	\sup_{0\leq s\leq t}|Y_s^\nu|^2
	\right]
	\leq
	C_t
	\left(
	1+
	\mathbb{E}|X_0|^2+
	\mathbb{E}|Y_0|^2+
	\mathbb{E}\sup_{0\leq s\leq t}m_2(\nu_s)
	\right).
	\tag{SA.10}
	\]
\end{lem}

\begin{proof}
	Set $Z_s^\nu=(X_s^\nu,Y_s^\nu)$. By (GV3),
	\[
	|\beta(s,Z_s^\nu,\nu_s)|^2
	+
	\|A^{1/2}(s,Z_s^\nu,\nu_s)\|^2
	\leq
	C\bigl(
	1+|Z_s^\nu|^2+m_2(\nu_s)
	\bigr).
	\]
	For $p\geq2$, the Burkholder-Davis-Gundy (BDG) inequality, H\"older's inequality, and the elementary bound
	$(a_1+\cdots+a_k)^p\leq k^{p-1}\sum_i a_i^p$ imply
	\begin{align*}
		\mathbb{E}
		\left[\sup_{0\leq r\leq s}|Z_r^\nu|^p\right]
		&\leq
		C_p\mathbb{E}|Z_0|^p
		+
		C_p\mathbb{E}
		\left(
		\int_0^s
		|\beta(v,Z_v^\nu,\nu_v)|\,dv
		\right)^p\\
		&\quad+
		C_p\mathbb{E}
		\left(
		\int_0^s
		\|A^{1/2}(v,Z_v^\nu,\nu_v)\|^2\,dv
		\right)^{p/2}\\
		&\leq
		C_{p,t}
		\left(
		1+\mathbb{E}|Z_0|^p
		+
		\int_0^s
		\mathbb{E}
		\sup_{0\leq q\leq v}|Z_q^\nu|^p\,dv
		+
		\mathbb{E}\sup_{0\leq v\leq s}m_p(\nu_v)
		\right).
	\end{align*}
	Gronwall's inequality gives (SA.9). Equation (SA.10) follows by taking
	$p=2$.
\end{proof}

\begin{lem}
	For every $0\leq r\leq s\leq t$,
	\begin{align}
		\mathbb{E}
		\left[
		|X_s^\nu-X_r^\nu|^2
		+
		|Y_s^\nu-Y_r^\nu|^2
		\right]
		\leq
		C|s-r|
		\left(
		1+
		\mathbb{E}\sup_{0\leq v\leq t}|Z_v^\nu|^2
		+
		\mathbb{E}\sup_{0\leq v\leq t}m_2(\nu_v)
		\right).
		\tag{SA.11}
	\end{align}
	If the initial law and candidate flows have uniformly bounded
	$(2+\varepsilon)$th moments, then for
	$q:=2+\varepsilon$,
	\[
	\mathbb{E}|Z_s^\nu-Z_r^\nu|^q
	\leq
	C_q|s-r|^{q/2}.
	\tag{SA.12}
	\]
\end{lem}

\begin{proof}
	Fix
	$
	0\le r\le s\le t.
	$
	From system~(SA.7),
	\[
	\begin{aligned}
		X_s^\nu-X_r^\nu
		&=
		\int_r^s
		b(v,Z_v^\nu,\nu_v)\,dv
		+
		\int_r^s
		\sigma(v,Z_v^\nu,\nu_v)\,dW_v
		+
		\int_r^s
		\sigma_0(v,Z_v^\nu,\nu_v)\,dB_v,\\
		Y_s^\nu-Y_r^\nu
		&=
		\int_r^s
		h(v,Z_v^\nu,\nu_v)\,dv
		+
		\int_r^s
		\Sigma_Y(v,Y_v^\nu,\nu_v)\,dB_v.
	\end{aligned}
	\]
	Invoking
	$
	|a+b+c|^2
	\le
	3(|a|^2+|b|^2+|c|^2),
	$
	Jensen's inequality for the drift integrals, and the It\^o isometry for the stochastic integrals,
	\[
	\begin{aligned}
		\mathbb E|X_s^\nu-X_r^\nu|^2
		\le\;&
		3|s-r|
		\int_r^s
		\mathbb E
		|b(v,Z_v^\nu,\nu_v)|^2\,dv
		\\
		&
		+
		3
		\int_r^s
		\mathbb E
		\|\sigma(v,Z_v^\nu,\nu_v)\|^2\,dv
		+
		3
		\int_r^s
		\mathbb E
		\|\sigma_0(v,Z_v^\nu,\nu_v)\|^2\,dv,
		\\
		\mathbb E|Y_s^\nu-Y_r^\nu|^2
		\le\;&
		2|s-r|
		\int_r^s
		\mathbb E
		|h(v,Z_v^\nu,\nu_v)|^2\,dv
		\\
		&
		+
		2
		\int_r^s
		\mathbb E
		\|\Sigma_Y(v,Y_v^\nu,\nu_v)\|^2\,dv .
	\end{aligned}
	\]
	Summing the preceding inequalities and applying Condition
	\textnormal{(GV3)},
	\[
	\begin{aligned}
		&
		\mathbb E
		\Bigl[
		|X_s^\nu-X_r^\nu|^2
		+
		|Y_s^\nu-Y_r^\nu|^2
		\Bigr]
		\\
		&\qquad
		\le
		C
		\int_r^s
		\mathbb E
		\Bigl[
		1
		+
		|Z_v^\nu|^2
		+
		m_2(\nu_v)
		\Bigr]\,dv
		\\
		&\qquad
		\le
		C|s-r|
		\left(
		1
		+
		\mathbb E
		\sup_{0\le v\le t}
		|Z_v^\nu|^2
		+
		\mathbb E
		\sup_{0\le v\le t}
		m_2(\nu_v)
		\right),
	\end{aligned}
	\]
	which is precisely (SA.11). Suppose,
	$
	\lambda_0
	$
	and
	$
	\nu
	$
	have uniformly bounded
	$
	(2+\varepsilon)
	$
	th moments, and define
	$
	q:=2+\varepsilon.
	$
	Applying the Burkholder--Davis--Gundy inequality to the stochastic integrals yields
	\[
	\begin{aligned}
		\mathbb E
		|Z_s^\nu-Z_r^\nu|^q
		\le\;&
		C_q
		\mathbb E
		\left|
		\int_r^s
		\beta(v,Z_v^\nu,\nu_v)\,dv
		\right|^q
		+
		C_q
		\mathbb E
		\left(
		\int_r^s
		\|\Xi(v,Z_v^\nu,\nu_v)\|^2\,dv
		\right)^{q/2},
	\end{aligned}
	\]
	where
	$
	\Xi\Xi^\dagger=A
	$
	is any measurable square root of the covariance matrix. Using Jensen's inequality for the drift term,
	the linear-growth estimate in Condition
	\textnormal{(GV3)},
	and the uniform
	$
	L^q
	$
	bound established in (SA.9),
	\[
	\begin{aligned}
		\mathbb E
		|Z_s^\nu-Z_r^\nu|^q
		&\le
		C_q
		|s-r|^{q-1}
		\int_r^s
		\mathbb E
		\left(
		1
		+
		|Z_v^\nu|^q
		+
		m_q(\nu_v)
		\right)\,dv
		\\
		&
		\quad
		+
		C_q
		|s-r|^{q/2-1}
		\int_r^s
		\mathbb E
		\left(
		1
		+
		|Z_v^\nu|^q
		+
		m_q(\nu_v)
		\right)\,dv
		\le
		C_q
		|s-r|^{q/2},
	\end{aligned}
	\]
	where
	$
	m_q(\nu)
	:=
	\int_{\mathbb R^d}
	|x|^q\,\nu(dx).
	$
	Since
	$
	q>2,
	$
	$
	|s-r|^q
	\le
	|s-r|^{q/2},
	$
	so the drift contribution is absorbed into the diffusion estimate. This completes the proof.
\end{proof}

\subsection{The conditional-law operator}

For
$
\nu\in\mathfrak M_R,
$
define
\[
\Gamma(\nu)_s
:=
\mathcal L
(X_s^\nu\mid\mathcal F_s^{Y^\nu}),
\qquad
0\le s\le t.
\tag{SA.13}
\]
Then
$
\Gamma:
\mathfrak M_R
\rightarrow
C_{\mathbb F^{Y}}
([0,t];\mathcal P_2(\mathbb R^d)).
$
Since
$
\mathbb R^d
$
and
$
\mathcal Y
$
are Polish,
$
\Gamma(\nu)
$
admits a jointly Borel measurable version. Under
Condition~\textnormal{(GU3)},
there exists a measurable kernel
$
\Pi:
[0,t]\times\mathcal Y
\rightarrow
\mathcal P_2(\mathbb R^d)
$
such that
$
\Gamma(\nu)_s
=
\Pi(s,Y^\nu_{\cdot\wedge s}),
$
$\mathbb P$-a.s., i.e.,
$
\Gamma(\nu)
$
is canonically identified with the conditional-law kernel generated by the stopped observation path
$
Y^\nu_{\cdot\wedge s}.
$

\begin{lem}
	For every frozen solution,
	\[
	m_2(\Gamma(\nu)_s)
	=
	\mathbb{E}
	\left[
	|X_s^\nu|^2\mid\mathcal{F}_s^{Y^\nu}
	\right]
	\quad\text{a.s.}
	\tag{SA.14}
	\]
	Consequently,
	\begin{align}
		\mathbb E
		\left[\sup_{0\le s\le t}
		m_2(\Gamma(\nu)_s)\right]
		\le
		C_{\varepsilon}
		\left(
		\mathbb E
		\left[\sup_{0\le s\le t}
		|X_s^\nu|^{2+\varepsilon}\right]
		\right)^{\frac{2}{2+\varepsilon}}.
		\tag{SA.15}
	\end{align}
\end{lem}

\begin{proof}
	Fix
	$
	\nu\in\mathfrak M_R
	$
	and
	$
	s\in[0,t].
	$
	For
	$
	n\in\mathbb N,
	$
	define
	$
	\varphi_n:\mathbb R^d\rightarrow\mathbb R_+
	$
	by
	$
	\varphi_n(x):=|x|^2\wedge n.
	$
	Then
	$
	\varphi_n\in\mathcal B_b(\mathbb R^d),
	$
	$
	0\le\varphi_n\uparrow|\cdot|^2,
	$
	and, by the defining property of the regular conditional distribution
	$
	\Gamma(\nu)_s
	=
	\mathcal L(X_s^\nu\mid\mathcal F_s^{Y^\nu}),
	$
	\[
	\left\langle
	\Gamma(\nu)_s,\varphi_n
	\right\rangle
	=
	\int_{\mathbb R^d}
	\varphi_n(x)\,
	\Gamma(\nu)_s(dx)
	=
	\mathbb E
	\left[
	\varphi_n(X_s^\nu)
	\mid
	\mathcal F_s^{Y^\nu}
	\right]
	\quad
	\mathbb P\text{-a.s.}
	\]
	
	Since
	$
	\varphi_n(x)\uparrow|x|^2
	$
	for every
	$
	x\in\mathbb R^d,
	$
	monotone convergence on the kernel side yields
	$
	\lim_{n\to\infty}
	\left\langle
	\Gamma(\nu)_s,\varphi_n
	\right\rangle
	=
	\int_{\mathbb R^d}
	|x|^2\,
	\Gamma(\nu)_s(dx)
	=
	m_2(\Gamma(\nu)_s).
	$
	Likewise, conditional monotone convergence yields
	$
	\lim_{n\to\infty}
	\mathbb E
	\left[
	\varphi_n(X_s^\nu)
	\mid
	\mathcal F_s^{Y^\nu}
	\right]
	=
	\mathbb E
	\left[
	|X_s^\nu|^2
	\mid
	\mathcal F_s^{Y^\nu}
	\right],
	\
	\mathbb P\text{-a.s.}
	$
	Hence,
	$
	m_2(\Gamma(\nu)_s)
	=
	\mathbb E
	\left[
	|X_s^\nu|^2
	\mid
	\mathcal F_s^{Y^\nu}
	\right],
	\
	\mathbb P\text{-a.s.},
	$
	which implies \textnormal{(SA.14)}. We next determine
	$
	\sup_{0\le s\le t}m_2(\Gamma(\nu)_s).
	$
	Define
	$
	S^\nu
	:=
	\sup_{0\le v\le t}
	|X_v^\nu|^2,
	$ and $
	M_s^\nu
	:=
	\mathbb E
	\left[
	S^\nu
	\mid
	\mathcal F_s^{Y^\nu}
	\right],
	$ for all $
	0\le s\le t.
	$
	Since,
	$
	|X_s^\nu|^2\le S^\nu
	$
	for every
	$
	s\in[0,t],
	$
	the monotonicity of conditional expectation and
	\textnormal{(SA.14)}
	imply
	\[
	0
	\le
	m_2(\Gamma(\nu)_s)
	=
	\mathbb E
	\left[
	|X_s^\nu|^2
	\mid
	\mathcal F_s^{Y^\nu}
	\right]
	\le
	\mathbb E
	\left[
	S^\nu
	\mid
	\mathcal F_s^{Y^\nu}
	\right]
	=
	M_s^\nu.
	\]
	Therefore,
	$
	\sup_{0\le s\le t}
	m_2(\Gamma(\nu)_s)
	\le
	\sup_{0\le s\le t}
	M_s^\nu
	\quad
	\mathbb P\text{-a.s.}
	$ The process
	$
	M^\nu=(M_s^\nu)_{0\le s\le t}
	$
	is a nonnegative
	$
	\mathbb F^{Y^\nu}
	$
	-martingale. Let
	$
	p:=1+\varepsilon/2>1.
	$
	By the
	$
	(2+\varepsilon)
	$
	-moment estimate,
	\[
	\mathbb E
	\left[
	(S^\nu)^p
	\right]
	=
	\mathbb E
	\left[
	\sup_{0\le v\le t}
	|X_v^\nu|^{2p}
	\right]
	=
	\mathbb E
	\left[
	\sup_{0\le v\le t}
	|X_v^\nu|^{2+\varepsilon}
	\right]
	<
	\infty.
	\]
	Thus
	$
	M^\nu
	$
	is bounded in
	$
	L^p.
	$
	By Doob's
	$
	L^p
	$
	maximal inequality,
	$
	\left\|
	\sup_{0\le s\le t}
	M_s^\nu
	\right\|_{L^p}
	\le
	\frac{p}{p-1}\\
	\|M_t^\nu\|_{L^p}.
	$
	Conditional Jensen's inequality yields
	$
	|M_t^\nu|^p
	=
	\left|
	\mathbb E
	\left[
	S^\nu
	\mid
	\mathcal F_t^{Y^\nu}
	\right]
	\right|^p
	\le
	\mathbb E
	\left[
	(S^\nu)^p
	\mid
	\mathcal F_t^{Y^\nu}
	\right],
	$
	and therefore,
	$
	\mathbb E|M_t^\nu|^p
	\le
	\mathbb E(S^\nu)^p.
	$
	Combining the preceding bounds,
	$
	\left\|
	\sup_{0\le s\le t}
	M_s^\nu
	\right\|_{L^p}
	\le
	\frac{p}{p-1}\\
	\left(
	\mathbb E(S^\nu)^p
	\right)^{1/p}.
	$
	Since,
	$
	\|U\|_{L^1}\le\|U\|_{L^p}
	$
	for every nonnegative
	$
	U\in L^p
	$
	on a probability space,
	\[
	\begin{aligned}
		\mathbb E
		\left[
		\sup_{0\le s\le t}
		m_2(\Gamma(\nu)_s)
		\right]
		&\le
		\mathbb E
		\left[
		\sup_{0\le s\le t}
		M_s^\nu
		\right]
		\le
		\left\|
		\sup_{0\le s\le t}
		M_s^\nu
		\right\|_{L^p}
		\le
		\frac{p}{p-1}
		\left(
		\mathbb E(S^\nu)^p
		\right)^{1/p}.
	\end{aligned}
	\]
	Using
	$
	p=1+\varepsilon/2
	$
	and
	$
	(S^\nu)^p
	=
	\sup_{0\le v\le t}|X_v^\nu|^{2+\varepsilon},
	$
	we obtain
	\[
	\mathbb E
	\left[
	\sup_{0\le s\le t}
	m_2(\Gamma(\nu)_s)
	\right]
	\le
	C_\varepsilon
	\left(
	\mathbb E
	\sup_{0\le v\le t}
	|X_v^\nu|^{2+\varepsilon}
	\right)^{\frac{2}{2+\varepsilon}},
	\]
	where
	$
	C_\varepsilon
	=
	p/(p-1)
	=
	(2+\varepsilon)/\varepsilon.
	$
	This proves \textnormal{(SA.15)}. More generally, for every
	$
	q>1
	$
	so that
	$
	\mathbb E(S^\nu)^q<\infty,
	$
	the same argument gives
	$
	\left\|
	\sup_{0\le s\le t}
	m_2(\Gamma(\nu)_s)
	\right\|_{L^q}
	\le
	\frac{q}{q-1}
	\left\|
	S^\nu
	\right\|_{L^q}.
	$
	Subsequently, any family of frozen solutions satisfying
	$
	\sup_{\nu\in\mathfrak M_R}
	\mathbb E
	\sup_{0\le s\le t}
	|X_s^\nu|^{2+\varepsilon}
	<\infty
	$
	also satisfies
	$
	\sup_{\nu\in\mathfrak M_R}
	\mathbb E
	\sup_{0\le s\le t}
	m_2(\Gamma(\nu)_s)
	<\infty,
	$
	which is the moment bound required in the invariant-set and compactness arguments.
\end{proof}

\begin{lem}
	Let $\nu^n,\nu\in\mathfrak{M}_R$ and suppose
	$
	\sup_{0\leq s\leq t}
	\mathcal{W}_2(\nu_s^n,\nu_s)
	\rightarrow0
	$
	in probability. Then
	\[
	\sup_{0\leq s\leq t}
	\mathcal{W}_2
	\left(
	\Gamma(\nu^n)_s,\Gamma(\nu)_s
	\right)
	\longrightarrow0
	\]
	in probability.
\end{lem}

\begin{proof}
	Let
	$
	\nu^n,\nu\in\mathfrak M_R
	$
	satisfy
	$
	\sup_{0\le s\le t}\mathcal W_2(\nu_s^n,\nu_s)\to0
	$
	in probability, and let
	$
	Z^{\nu^n}:=(X^{\nu^n},\\ Y^{\nu^n}),
	$
	and
	$
	Z^\nu:=(X^\nu,Y^\nu)
	$
	denote the corresponding frozen solutions. On a common realization carrying the same initial condition and driving Brownian motions
	$
	(W,B),
	$
	set
	$
	\Delta X_s^n:=X_s^{\nu^n}-X_s^\nu,
	$
	$
	\Delta Y_s^n:=Y_s^{\nu^n}-Y_s^\nu,
	$
	and
	$
	\Delta Z_s^n:=(\Delta X_s^n,\Delta Y_s^n).
	$
	By the frozen equations,
	$
	\Delta X_s^n
	=
	\int_0^s\Delta b_v^n\,dv
	+
	\int_0^s\Delta\sigma_v^n\,dW_v
	+
	\int_0^s\Delta\sigma_{0,v}^n\,dB_v
	$
	and
	$
	\Delta Y_s^n
	=
	\int_0^s\Delta h_v^n\,dv
	+
	\int_0^s\Delta\Sigma_{Y,v}^n\,dB_v,
	$
	where, for example,
	$
	\Delta b_v^n
	:=
	b(v,X_v^{\nu^n},Y_v^{\nu^n},\nu_v^n)
	-
	b(v,X_v^\nu,Y_v^\nu,\nu_v),
	$
	with analogous notation for
	$
	\Delta\sigma_v^n,$ $
	\Delta\sigma_{0,v}^n,$ $
	\Delta h_v^n,
	$
	and
	$
	\Delta\Sigma_{Y,v}^n.
	$ Condition \textnormal{(GV2)} yields
	$
	|\Delta b_v^n|
	+
	\|\Delta\sigma_v^n\|
	+
	\|\Delta\sigma_{0,v}^n\|
	+
	|\Delta h_v^n|
	+
	\|\Delta\Sigma_{Y,v}^n\|
	\le
	L\bigl(
	|\Delta X_v^n|
	+
	|\Delta Y_v^n|
	+
	\mathcal W_2(\nu_v^n,\nu_v)
	\bigr).
	$
	Hence, by Jensen's inequality, the BDG inequality, and
	$
	(a+b+c)^2\le3(a^2+b^2+c^2),
	$
	for every
	$
	s\in[0,t],
	$
	\begin{align}
		&\mathbb E\left[
		\sup_{0\le r\le s}
		\left(
		|\Delta X_r^n|^2
		+
		|\Delta Y_r^n|^2
		\right)\right]
		\notag\\
		&\qquad\le
		C
		\int_0^s
		\mathbb E\left[
		\sup_{0\le q\le v}
		\left(
		|\Delta X_q^n|^2
		+
		|\Delta Y_q^n|^2
		\right)\right]\,dv
		+
		C
		\int_0^s
		\mathbb E\left[
		\mathcal W_2^2(\nu_v^n,\nu_v)
		\right]\,dv.
		\tag{SA.16}
	\end{align}
	Let
	$
	\Delta_n(s)
	:=
	\mathbb E[
	\sup_{0\le r\le s}|\Delta Z_r^n|^2
	].
	$
	Then
	$
	\Delta_n(s)
	\le
	C\int_0^s\Delta_n(v)\,dv
	+
	C\int_0^s
	\mathbb E[\mathcal W_2^2(\nu_v^n,\nu_v)]\,dv,
	$
	and Gronwall's inequality yields
	$
	\Delta_n(t)
	\le
	C_t
	\int_0^t
	\mathbb E[
	\mathcal W_2^2(\nu_v^n,\nu_v)
	]\,dv.
	$
	By
	$
	\nu^n,\nu\in\mathfrak M_R
	$
	and the uniform second-moment bounds, the family
	$
	\{\sup_{0\le v\le t}\mathcal W_2^2(\nu_v^n,\nu_v)\}_{n\ge1}
	$
	is uniformly integrable. Therefore,
	$
	\sup_{0\le v\le t}\mathcal W_2(\nu_v^n,\nu_v)\to0
	$
	in probability implies
	$
	\int_0^t
	\mathbb E[
	\mathcal W_2^2(\nu_v^n,\nu_v)
	]\,dv\to0,
	$
	after passage to a subsequence if necessary, and thus
	$
	\sup_{0\le s\le t}
	|Z_s^{\nu^n}-Z_s^\nu|
	\to0
	$
	in probability. It remains to pass from
	$
	(Z^{\nu^n},Z^\nu)
	$
	to the conditional-law processes
	$
	\Gamma(\nu^n)
	$
	and
	$
	\Gamma(\nu).
	$
	Set
	$
	\Gamma(\nu^n)_s
	=
	\mathcal L(X_s^{\nu^n}\mid\mathcal F_s^{Y^{\nu^n}})
	$
	and
	$
	\Gamma(\nu)_s
	=
	\mathcal L(X_s^\nu\mid\mathcal F_s^{Y^\nu}).
	$
	The difficulty is that
	$
	\mathcal F_s^{Y^{\nu^n}}
	$
	depends on
	$
	n
	$
	through both
	$
	Y^{\nu^n}
	$
	and
	$
	a_Y(v,Y_v^{\nu^n},\nu_v^n)
	=
	\Sigma_Y\Sigma_Y^\dagger
	(v,Y_v^{\nu^n},\nu_v^n).
	$
	Condition \textnormal{(GV4)} gives
	$
	\underline a I_{d_Y}
	\le
	a_Y(v,Y_v^{\nu^n},\nu_v^n)
	\le
	\overline a I_{d_Y},
	$
	uniformly in
	$
	n,
	$
	while the moment estimates and
	\textnormal{(SA.9)}
	yield tightness and uniform integrability of
	$
	\{(X^{\nu^n},Y^{\nu^n})\}_{n\ge1}.
	$
	Condition \textnormal{(GV6)} then applies to the jointly convergent frozen systems and gives
	$
	\sup_{0\le s\le t}
	\mathcal W_2
	\left(
	\mathcal L(X_s^{\nu^n}\mid\mathcal F_s^{Y^{\nu^n}}),
	\mathcal L(X_s^\nu\mid\mathcal F_s^{Y^\nu})
	\right)
	\rightarrow0
	$
	in probability. Equivalently,
	$
	\sup_{0\le s\le t}
	\mathcal W_2(
	\Gamma(\nu^n)_s,
	\Gamma(\nu)_s
	)
	\to0
	$
	in probability, which proves the claim.
\end{proof}

\begin{rmk}
	The preceding lemma is not implied by stability of the frozen SDE alone.
	Under
	$
	\Sigma_Y=\Sigma_Y(s,y,\nu),
	$
	$
	\nu\mapsto
	(Y^\nu,\mathbb F^{Y^\nu},\Gamma(\nu))
	$
	is fully coupled through
	$
	a_Y=\Sigma_Y\Sigma_Y^\dagger.
	$
	Hence
	$
	Z^{\nu^n}\Rightarrow Z^\nu
	$
	does not by itself imply
	$
	\Gamma(\nu^n)\Rightarrow\Gamma(\nu).
	$
	Condition~\textnormal{(GV6)}
	precisely requires continuity of the conditional-law operator
	$
	\Gamma:\mathfrak M_R\rightarrow
	C([0,t];\mathcal P_2(\mathbb R^d))
	$
	with respect to the topology induced by
	$
	\sup_{0\le s\le t}\mathcal W_2.
	$
	Relative to the constant-volatility case, this is the only additional ingredient needed for the fixed-point argument.
\end{rmk}

\subsection{Tightness and compactness}

\begin{lem}
	Let $\{\nu^n\}\subset\mathfrak{M}_R$. Then the laws of
	$\{(X^{\nu^n},Y^{\nu^n})\}$ are tight on
	$\mathcal{X}\times\mathcal{Y}$. Moreover, under (GV5), the laws of
	$\{\Gamma(\nu^n)\}$ are tight on $\mathcal{M}$.
\end{lem}

\begin{proof}
	Let
	$
	Z^{\nu^n}
	:=
	(X^{\nu^n},Y^{\nu^n}).
	$
	By (SA.9),
	$
	\sup_n
	\mathbb E
	\sup_{0\le s\le t}
	|Z_s^{\nu^n}|^{2+\varepsilon}
	<\infty.
	$
	Together with the increment estimate (SA.12),
	$
	\sup_n
	\mathbb E
	|Z_s^{\nu^n}-Z_r^{\nu^n}|^{2+\varepsilon}
	\le
	C|s-r|^{1+\varepsilon/2},
	$
	Kolmogorov-Chentsov yields a modification with uniformly H\"older-continuous paths of every order
	$
	\alpha<
	\varepsilon/(4+2\varepsilon).
	$
	Hence,
	$
	\{\mathcal L(Z^{\nu^n})\}
	$
	is tight on
	$
	\mathcal X\times\mathcal Y.
	$
	Next,
	$
	\Gamma(\nu^n)
	\in
	\mathcal M
	=
	C([0,t];\mathcal P_2(\mathbb R^d)).
	$
	By (SA.15),
	$
	\sup_n
	\mathbb E
	\sup_{0\le s\le t}
	m_2(\Gamma(\nu^n)_s)
	<
	\infty,
	$
	while Condition~\textnormal{(GV5)} gives
	$
	\Gamma(\mathfrak M_R)
	\Subset
	C([0,t];\mathcal P_2(\mathbb R^d))
	$
	together with uniform integrability of
	$
	m_2.
	$
	Equivalently,
	$
	\{\Gamma(\nu^n)\}
	$
	is relatively compact in law on
	$
	\mathcal M.
	$
	Prokhorov's theorem therefore yields
	$
	\{\mathcal L(\Gamma(\nu^n))\}
	$
	tight on
	$
	\mathcal M,
	$
	completing the proof.
\end{proof}

\begin{rmk}
	Define
	$
	\mathfrak K_R
	:=
	\overline{\operatorname{co}}
	\bigl(
	\Gamma(\mathfrak M_R)
	\bigr),
	$
	where the closure is taken in
	$
	(\mathcal M,d_{\mathcal M})
	$
	with convergence in probability. By Condition~\textnormal{(GV5)},
	$
	\varnothing\neq
	\Gamma(\mathfrak M_R)
	\subseteq
	\mathfrak K_R
	\Subset
	\mathcal M,
	$
	$
	\mathfrak K_R
	$
	is convex and compact, and
	$
	\Gamma(\mathfrak K_R)
	\subseteq
	\Gamma(\mathfrak M_R)
	\subseteq
	\mathfrak K_R.
	$
	Hence
	$
	\Gamma:
	\mathfrak K_R
	\rightarrow
	\mathfrak K_R
	$
	is a compact self-map.
\end{rmk}

\subsection{Existence by a fixed point}

\begin{proof}[Proof of existence in Theorem 3]
	Let
	$
	\mathfrak K_R
	=
	\overline{\operatorname{co}}
	(\Gamma(\mathfrak M_R))
	\subset\mathcal M.
	$
	By the preceding remark,
	$
	\varnothing\neq
	\mathfrak K_R,
	$
	$
	\mathfrak K_R
	$
	is convex and compact, and
	$
	\Gamma(\mathfrak K_R)
	\subseteq
	\mathfrak K_R.
	$
	For every
	$
	\nu\in\mathfrak K_R,
	$
	the frozen martingale problem
	$
	\mathrm{MP}(\nu)
	$
	admits a weak solution by the frozen existence lemma. If
	$
	\mathrm{MP}(\nu)
	$
	is not unique, fix a measurable selector on the corresponding compact set of solution laws; under Condition~\textnormal{(GU1)},
	$
	\mathrm{MP}(\nu)
	$
	is well posed, so no selection is required. By (SA.10),
	$
	\sup_{\nu\in\mathfrak K_R}
	\mathbb E
	\sup_{0\le s\le t}
	|Z_s^\nu|^{2+\varepsilon}
	<\infty,
	$
	while (SA.14) gives
	$
	m_2(\Gamma(\nu)_s)
	=
	\mathbb E
	[
	|X_s^\nu|^2
	\mid
	\mathcal F_s^{Y^\nu}
	],
	$
	whence
	$
	\Gamma(\mathfrak K_R)
	\subseteq
	\mathfrak K_R.
	$
	The tightness lemma yields
	$
	\Gamma(\mathfrak K_R)
	\Subset
	\mathcal M,
	$
	and the continuity lemma together with Condition~\textnormal{(GV6)} implies
	$
	\nu^n\rightarrow\nu
	\Longrightarrow
	\Gamma(\nu^n)\rightarrow\Gamma(\nu)
	$
	in
	$
	(\mathcal M,d_{\mathcal M}).
	$
	Hence
	$
	\Gamma:
	\mathfrak K_R
	\rightarrow
	\mathfrak K_R
	$
	is a continuous compact self-map.  Schauder-Tychonoff theorem yields
	$
	\mu\in\mathfrak K_R
	$
	such that
	$
	\Gamma(\mu)=\mu.
	$
	Let
	$
	(X^\mu,Y^\mu)
	$
	solve
	$
	\mathrm{MP}(\mu).
	$
	Then
	\[
	\mu_s
	=
	\mathcal L
	(X_s^\mu
	\mid
	\mathcal F_s^{Y^\mu}),
	\qquad \forall\ 
	0\le s\le t,
	\tag{SA.17}
	\]
	and substitution of
	$
	\nu=\mu
	$
	into the frozen system (SA.7) recovers system~(3). Thus
	$
	(X^\mu,Y^\mu,\mu)
	$
	is a weak solution of the original conditional MVSDE. Finally,
	$
	\mu\in\mathfrak K_R
	$
	and (SA.10)-(SA.14) imply
	$
	\mathbb E
	[
	\sup_{0\le s\le t}
	|X_s^\mu|^2
	+
	\sup_{0\le s\le t}
	|Y_s^\mu|^2
	+
	\sup_{0\le s\le t}
	m_2(\mu_s)
	]
	<\infty.
	$
	Moreover,
	Condition~\textnormal{(GU3)}
	provides a jointly measurable version of
	$
	\mu,
	$
	and Condition~\textnormal{(GV5)}
	upgrades this version to an
	$
	\mathbb F^{Y^\mu}
	$
	-progressively measurable process with continuous
	$
	\mathcal P_2(\mathbb R^d)
	$
	-valued paths. Hence all conclusions of Theorem~3 follow.
\end{proof}

\subsection{Uniqueness}

\begin{lem}[Uniqueness of the fixed point]
	Assume (GU2). If $\mu^1$ and $\mu^2$ are two fixed points of $\Gamma$
	in the stated solution class, then they are indistinguishable.
\end{lem}

\begin{proof}
	Let
	$
	\mu^1,\mu^2\in\operatorname{Fix}(\Gamma),
	$
	so
	$
	\Gamma(\mu^i)=\mu^i,
	$
	$i=1,2$.
	Define
	$
	\Delta(s)
	:=
	\mathbb E
	\!\left[
	\sup_{0\le r\le s}
	\mathcal W_2^2(\mu_r^1,\mu_r^2)
	\right]
	$ for all $0\leq s\leq t$.
	Then
	$
	\Delta(0)=0,
	$
	and, by the fixed-point identity,
	\[
	\Delta(s)
	=
	\mathbb E
	\!\left[
	\sup_{0\le r\le s}
	\mathcal W_2^2
	\bigl(
	\Gamma(\mu^1)_r,
	\Gamma(\mu^2)_r
	\bigr)
	\right].
	\]
	Condition~\textnormal{(GU2)} therefore yields
	$
	\Delta(s)
	\le
	C
	\int_0^s
	\Delta(v)\,dv,
	\
	0\le s\le t.
	$
	Since,
	$
	\Delta
	$
	is nonnegative and locally integrable, Gronwall's inequality yields
	$
	\Delta(s)\equiv0
	$
	on
	$
	[0,t].
	$
	Hence,
	\[
	\mathbb E
	\!\left[
	\sup_{0\le r\le t}
	\mathcal W_2^2(\mu_r^1,\mu_r^2)
	\right]
	=0,
	\]
	so
	$
	\sup_{0\le r\le t}
	\mathcal W_2(\mu_r^1,\mu_r^2)=0
	$
	$\mathbb P$-a.s. Therefore
	$
	\mu_r^1=\mu_r^2
	$
	for every
	$
	r\in[0,t]
	$
	outside a common null set, i.e.,
	$
	\mu^1
	$
	and
	$
	\mu^2
	$
	are indistinguishable.
\end{proof}

\begin{proof}[Proof of uniqueness in Theorem 3]
	Let
	$
	(\Omega^i,\mathcal F^i,\mathbb F^i,\mathbb P^i;
	X^i,Y^i,W^i,B^i,\mu^i),
	$
	$i\in\{1,2\}$,
	be weak solutions of system~(3) with common initial law
	$
	\lambda_0.
	$
	For each
	$
	i,
	$
	$
	\mu_s^i
	=
	\mathcal L^{\mathbb P^i}
	(X_s^i\mid\mathcal F_s^{Y^i}),
	$
	$
	0\le s\le t,
	$
	and therefore
	$
	\Gamma(\mu^i)=\mu^i.
	$
	Let
	$
	\Omega^\ast
	:=
	\mathcal X\times\mathcal Y\times\mathcal M
	$
	with canonical coordinates
	$
	(\mathbf X,\mathbf Y,\mathbf m).
	$
	By Condition~\textnormal{(GU3)}, there exists a jointly measurable kernel
	$
	\Pi:[0,t]\times\mathcal Y\to\mathcal P_2(\mathbb R^d)
	$
	such that, under either solution,
	$
	\mu_s^i
	=
	\Pi(s,Y^i_{\cdot\wedge s}),
	$
	$\mathbb P^i$-a.s.
	Hence, both conditional-law processes admit canonical versions on
	$
	\mathcal Y,
	$
	and
	$
	\mathcal L^{\mathbb P^i}(\mu^i)
	$
	is determined by
	$
	\mathcal L^{\mathbb P^i}(Y^i).
	$
	Since,
	$
	\mu^1,\mu^2\in\operatorname{Fix}(\Gamma),
	$
	Condition~\textnormal{(GU2)}
	and the fixed-point uniqueness lemma imply
	$
	\mathcal L^{\mathbb P^1}(\mu^1)
	=
	\mathcal L^{\mathbb P^2}(\mu^2)
	=:
	\mathfrak q
	$
	on
	$
	\mathcal M.
	$
	Disintegrate
	$
	\mathcal L^{\mathbb P^i}(X^i,Y^i,\mu^i)
	$
	with respect to
	$
	\mu^i:
	$
	$
	\mathcal L^{\mathbb P^i}(X^i,Y^i,\mu^i)
	=
	\int_{\mathcal M}
	Q_m^i(\,\cdot\,)\,
	\mathfrak q(dm),
	$
	where
	$
	Q_m^i
	$
	is a regular conditional law of
	$
	(X^i,Y^i)
	$
	given
	$
	\mu^i=m.
	$
	For
	$
	\mathfrak q
	$-a.e.
	$
	m\in\mathcal M,
	$
	both
	$
	Q_m^1
	$
	and
	$
	Q_m^2
	$
	solve the frozen martingale problem
	$
	\mathrm{MP}(m;\lambda_0).
	$
	Indeed, for every
	$
	f\in C_c^2(\mathbb R^{d+d_Y}),
	$
	$
	f(X_s^i,Y_s^i)
	-
	f(X_0^i,Y_0^i)
	-
	\int_0^s
	\mathcal L_v^{m}
	f(X_v^i,Y_v^i)\,dv
	$
	is a
	$
	Q_m^i
	$-martingale, where
	$
	\mathcal L_v^{m}
	$
	is the generator obtained from system~(3) by replacing
	$
	\mu_v
	$
	with
	$
	m_v.
	$
	Moreover,
	$
	Q_m^i\circ(X_0,Y_0)^{-1}
	=
	\lambda_0.
	$
	Condition~\textnormal{(GU1)}
	therefore gives
	$
	Q_m^1
	=
	Q_m^2
	$
	for
	$
	\mathfrak q
	$-a.e.
	$
	m.
	$
	Consequently,
	$
	\mathcal L^{\mathbb P^1}(X^1,Y^1,\mu^1)
	=
	\int_{\mathcal M}
	Q_m^1\,\mathfrak q(dm)
	=
	\int_{\mathcal M}
	Q_m^2\,\mathfrak q(dm)
	=
	\mathcal L^{\mathbb P^2}(X^2,Y^2,\mu^2)
	$
	on
	$
	\mathcal X\times\mathcal Y\times\mathcal M.
	$
	Thus any two weak solutions with initial law
	$
	\lambda_0
	$
	induce the same law of
	$
	(X,Y,\mu),
	$
	which proves weak uniqueness.
\end{proof}

\begin{rmk}
	Theorem~3 establishes
	$
	\mathcal L(X,Y,\mu)
	$
	as a unique element of
	$
	\mathcal P
	(\mathcal X\times\mathcal Y\times\mathcal M)
	$
	for a fixed coefficient tuple
	$
	(u,b,\sigma,\sigma_0,h,\Sigma_Y).
	$
	Equivalently,
	$
	(u,b,\sigma,\sigma_0,h,\Sigma_Y)
	\mapsto
	\mathcal L(X,Y,\mu)
	$
	is single-valued on the stated solution class. No converse statement is asserted
	$
	\mathcal L(X,Y,\mu)\\
	\not\!\!\rightarrow
	(u,b,\sigma,\sigma_0,h,\Sigma_Y).
	$
	Hence, weak uniqueness concerns well-posedness of system~(3), not econometric identification of the structural primitives from stochastic-choice observables.
\end{rmk}

\subsection{Relation to Proposition 2}

Proposition~2 corresponds to the special case
$
\Sigma_Y(s,y,\mu)\equiv\widehat\Sigma_Y,
$
$
\widehat\Sigma_Y\in\mathbb R^{d_Y\times d_Y}
$
deterministic and invertible. Then
$
\bar Y
:=
\widehat\Sigma_Y^{-1}Y
$
has quadratic variation
$
\langle\bar Y\rangle_s=sI_{d_Y},
$
so the observation equation is reducible, via Girsanov, to a reference probability measure with Brownian observation process. Consequently,
$
\nu\mapsto\Gamma(\nu)
$
reduces to the conditional MVSDE \citep{buckdahn2023general}. Under Theorem~3,
$
\Sigma_Y=\Sigma_Y(s,Y_s,\mu_s),
$
$
a_Y
=
\Sigma_Y\Sigma_Y^\dagger,
$
and
$
\langle Y\rangle_s
=
\int_0^s
a_Y(v,Y_v,\mu_v)\,dv.
$
Hence
$
\langle Y\rangle
$
depends on the endogenous conditional-law flow,
$
Y\mapsto\mathbb F^Y\mapsto\mu\mapsto a_Y,
$
so no drift-preserving change of measure can normalize the observation volatility. Accordingly,
$
\Gamma:
\mathfrak M_R
\rightarrow
\mathfrak M_R
$
is constructed directly through the joint martingale problem, with existence obtained by Schauder-Tychonoff and continuity supplied by Condition~\textnormal{(GV6)}.

\begin{rmk}
	Existence is equivalent to
	$
	\operatorname{Fix}(\Gamma)\neq\varnothing;
	$
	equivalently, there exists
	$
	(X,Y,\mu)
	$
	such that
	$
	\Gamma(\mu)=\mu,
	$
	i.e.,
	$
	\mu_s
	=
	\mathcal L(X_s\mid\mathcal F_s^Y),
	$
	$
	0\le s\le t,
	$
	as in (SA.17). Thus
	$
	(u,b,\sigma,\sigma_0,h,\Sigma_Y)
	\mapsto
	\mathcal L(X,Y,\mu)
	$
	is well defined. Weak uniqueness gives
	$
	|\operatorname{Fix}(\Gamma)|=1
	$
	up to equality in law, or equivalently,
	$
	\mathcal L(X,Y,\mu)
	$
	is uniquely determined by the coefficient tuple. No injectivity of the inverse map
	$
	\mathcal L(X,Y,\mu)
	\mapsto
	(u,b,\sigma,\sigma_0,h,\Sigma_Y,u)
	$
	is asserted; structural identification from stochastic-choice observables is a separate econometric problem.
\end{rmk}

\section{Behavioral Stability Results}
\label{sa:behavioral-stability}

\subsection{Dynamic choice stability}
\label{sa:dynamic-choice-stability}

\begin{proof}[Proof of Proposition~12]
	Let $(X^m,Y^m)$ and $(X^{m'},Y^{m'})$ denote frozen-flow solutions
	with initial laws $\lambda$ and $\lambda'$. Under the strengthened
	stability conditions, they may be coupled so that
	\[
	\mathbb E
	\left[
	\sup_{s\le v\le r}
	\left(
	|X_v^m-X_v^{m'}|^2
	+
	|Y_v^m-Y_v^{m'}|^2
	\right)
	\right]
	\le
	C_{s,r}
	\left(
	W_2^2(\lambda,\lambda')
	+
	\int_s^r W_2^2(m_v,m_v')\,dv
	\right).
	\]
	In particular,
	$
	\mathbb E|X_r^m-X_r^{m'}|
	\le
	C_{s,r}
	\big(
	W_2(\lambda,\lambda')
	+
	d_{s,r}(m,m')
	\big).
	$
	For $a\in A\setminus\{z\}$, define
	$
	G_a
	:=
	\Delta_{za}(r,X_r^m,m_r),
	$ and $
	G_a'
	=
	\Delta_{za}(r,X_r^{m'},m_r').
	$
	Assumption~10 implies
	$
	|G_a-G_a'|
	\le
	2L_u
	\big(
	|X_r^m-X_r^{m'}|
	+
	W_2(m_r,m_r')
	\big).
	$
	As in the proof of Lemma~11, for every $\eta>0$,
	$
	\mathbb P
	\left(
	\operatorname{sgn}G_a
	\ne
	\operatorname{sgn}G_a'
	\right)
	\le
	\kappa\eta
	+
	\frac{2L_u}{\eta}
	\big(
	\mathbb E|X_r^m-X_r^{m'}|
	+
	W_2(m_r,m_r')
	\big).
	$
	Substitution of the frozen-flow stability estimate and optimization
	over $\eta$ yield
	$
	\mathbb P
	\left(
	\operatorname{sgn}G_a
	\ne
	\operatorname{sgn}G_a'
	\right)
	\le
	K_{s,r}
	\big(
	W_2(\lambda,\lambda')^{1/2}
	+
	d_{s,r}(m,m')^{1/2}
	\big).
	$
	A union bound over $a\in A\setminus\{z\}$ establishes the result for
	$\mathsf D_{s,r}$. For $\mathsf D_{s,r}^{\,\bar\mu}$, the measure argument in
	comparison-date felicity is the same under both systems. The same
	argument therefore applies without the term
	$W_2(m_r,m_r')$ arising directly from felicity. Uniformity over
	$\bar\mu$ follows from the uniform constants in Assumption~10.
\end{proof}

\subsection{Filter-state sufficiency}
\label{sa:filter-state-sufficiency}

\begin{proof}[Proof of Lemma~15]
	Fix $s\in[0,t]$
	and an admissible continuation menu policy
	$\pi$. Let
	$
	\mathsf E
	:=
	\mathbb R^d
	\times
	\mathbb R^{d_Y}
	\times
	\mathcal P_2(\mathbb R^d)
	$
	and write
	$
	\widehat{\mathbf X}_s
	:=
	(Y_s,\mu_s).
	$
	For
	$
	(y,\mu)\in
	\mathbb R^{d_Y}
	\times
	\mathcal P_2(\mathbb R^d),
	$
	let
	$
	\mathbf Q_{s,t}^{\pi}(y,\mu;\cdot)
	$
	denote the law on
	$
	C([s,t];\mathsf E)
	$
	of a weak solution
	$
	(X_r,Y_r,\mu_r)_{r\in[s,t]}
	$
	to the continuation system with
	$
	Y_s=y,
	\
	\mathcal L(X_s\mid\mathcal F_s^Y)=\mu,
	$
	and coefficient tuple
	$
	(b,\sigma,\sigma_0,h,\Sigma_Y).
	$
	By the Markovian specification, for
	$
	r\in[s,t],
	$
	the local characteristics of
	$
	(X_r,Y_r)
	$
	are functions only of
	$
	(r,X_r,Y_r,\mu_r),
	$
	namely
	$
	b(r,X_r,Y_r,\mu_r),
	$ $
	\sigma(r,X_r,Y_r,\mu_r),
	$ $
	\sigma_0(r,X_r,Y_r,\mu_r),
	$ $
	h(r,X_r,Y_r,\mu_r),$
	$
	\Sigma_Y(r,Y_r,\mu_r).
	$
	Hence, the continuation martingale problem depends on the pre-$s$ history only through
	$
	(Y_s,\mu_s).
	$
	
	Let
	$
	y_{\cdot\wedge s}
	$
	and
	$
	y'_{\cdot\wedge s}
	$
	be admissible histories satisfying
	$
	(Y_s,\mu_s)(y_{\cdot\wedge s})
	=
	(Y_s,\mu_s)(y'_{\cdot\wedge s})
	=
	(y,\mu).
	$
	Then
	$
	Y_s(y_{\cdot\wedge s})
	=
	Y_s(y'_{\cdot\wedge s})
	=
	y
	$
	and
	$
	\mathcal L(X_s\mid y_{\cdot\wedge s})
	=
	\mu
	=
	\mathcal L(X_s\mid y'_{\cdot\wedge s}).
	$
	Thus the two continuation problems have the same conditional initial law
	$
	\mu(dx)\delta_y(dy)
	$
	on
	$
	\mathbb R^d\times\mathbb R^{d_Y}
	$
	and the same coefficient tuple
	$
	(b,\sigma,\sigma_0,h,\Sigma_Y).
	$
	Let
	$
	\mathbf P^1
	$
	and
	$
	\mathbf P^2
	$
	be the corresponding conditional continuation laws of
	$
	(X,Y,\mu)
	$
	given
	$
	y_{\cdot\wedge s}
	$
	and
	$
	y'_{\cdot\wedge s},
	$
	respectively.
	For
	$
	f\in C_c^2(\mathbb R^{d+d_Y}),
	$
	write
	$
	z=(x,y)
	$
	and define
	$
	\mathcal L_r^\mu f(z)
	:=
	\beta(r,z,\mu_r)\cdot Df(z)
	+
	\frac12
	\operatorname{Tr}
	\!\left(
	A(r,z,\mu_r)D^2f(z)
	\right),
	$
	where
	$
	\beta(r,z,\mu)
	:=
	\begin{pmatrix}
		b(r,x,y,\mu)\\
		h(r,x,y,\mu)
	\end{pmatrix}
	$
	and
	$
	A(r,z,\mu)
	:=
	\begin{pmatrix}
		\sigma\sigma^\dagger+\sigma_0\sigma_0^\dagger
		&
		\sigma_0\Sigma_Y^\dagger\\
		\Sigma_Y\sigma_0^\dagger
		&
		\Sigma_Y\Sigma_Y^\dagger
	\end{pmatrix}
	(r,x,y,\mu).
	$
	Under either
	$
	\mathbf P^i,
	$
	$
	M_r^f
	:=
	f(X_r,Y_r)
	-
	f(X_s,Y_s)
	-
	\int_s^r
	\mathcal L_v^\mu f(X_v,Y_v)\,dv,
	\qquad
	r\in[s,t],
	$
	is a martingale, and
	$
	\mu_r
	=
	\mathcal L(X_r\mid\mathcal F_r^Y).
	$
	Moreover,
	$
	\mathbf P^1\circ(X_s,Y_s)^{-1}
	=
	\mu(dx)\delta_y(dy)
	=
	\mathbf P^2\circ(X_s,Y_s)^{-1}.
	$
	Weak uniqueness of the continuation martingale problem therefore gives
	$
	\mathbf P^1
	=
	\mathbf P^2
	=
	\mathbf Q_{s,t}^{\pi}(y,\mu;\cdot).
	$
	Equivalently, for every bounded Borel functional
	$
	F:
	C([s,t];\mathsf E)
	\rightarrow
	\mathbb R,
	$
	$$
	\mathbb E
	\!\left[
	F\!\left(
	(X_r,Y_r,\mu_r)_{r\in[s,t]}
	\right)
	\middle|
	y_{\cdot\wedge s}
	\right]
	=
	\int
	F(\omega)\,
	\mathbf Q_{s,t}^{\pi}(y,\mu;d\omega)
	$$
	and
	$
	\mathbb E
	\!\left[
	F\!\left(
	(X_r,Y_r,\mu_r)_{r\in[s,t]}
	\right)
	\middle|
	y'_{\cdot\wedge s}
	\right]
	=
	\int
	F(\omega)\,
	\mathbf Q_{s,t}^{\pi}(y,\mu;d\omega).
	$
	Hence
	$
	\mathcal L
	\!\left(
	(X_r,Y_r,\mu_r)_{r\in[s,t]}
	\middle|
	y_{\cdot\wedge s}
	\right)
	=
	\mathcal L
	\!\left(
	(X_r,Y_r,\mu_r)_{r\in[s,t]}
	\middle|
	y'_{\cdot\wedge s}
	\right).
	$
	Now let
	$
	\mathsf Z_{s,t}
	$
	denote the measurable path space of future choices on
	$
	[s,t].
	$
	Under the fixed admissible menu policy
	$
	\pi,
	$
	there exists a measurable map
	$
	\Psi^\pi:
	C([s,t];\mathsf E)
	\rightarrow
	\mathsf Z_{s,t}
	$
	such that
	$
	(Z_r)_{r\in[s,t]}
	=
	\Psi^\pi
	\!\left(
	(X_r,Y_r,\mu_r)_{r\in[s,t]}
	\right)
	$
	almost surely.
	For every Borel set
	$
	B\subseteq\mathsf Z_{s,t},
	$
	$
	\mathbb P^\pi
	\!\left(
	(Z_r)_{r\in[s,t]}\in B
	\middle|
	y_{\cdot\wedge s}
	\right)
	=
	\mathbf Q_{s,t}^{\pi}
	\!\left(
	(y,\mu),
	(\Psi^\pi)^{-1}(B)
	\right),
	$
	while
	$
	\mathbb P^\pi
	\!\left(
	(Z_r)_{r\in[s,t]}\in B
	\middle|
	y'_{\cdot\wedge s}
	\right)
	=
	\mathbf Q_{s,t}^{\pi}
	\!\left(
	(y,\mu),
	(\Psi^\pi)^{-1}(B)
	\right).
	$
	Therefore,
	$$
	\mathcal L^\pi
	\!\left(
	(Z_r)_{r\in[s,t]}
	\middle|
	y_{\cdot\wedge s}
	\right)
	=
	\mathbf Q_{s,t}^{\pi}(y,\mu;\cdot)
	\circ
	(\Psi^\pi)^{-1}
	=
	\mathcal L^\pi
	\!\left(
	(Z_r)_{r\in[s,t]}
	\middle|
	y'_{\cdot\wedge s}
	\right).
	$$
	Thus
	$
	(Y_s,\mu_s)
	$
	is sufficient for the conditional law of every continuation choice experiment generated by
	$
	\pi.
	$
\end{proof}

\subsection{Behavioral impossibility of DRU}

\begin{proof}[Proof of Theorem 19]
	Let
	$
	\rho\in\mathfrak R_{\mathrm{DDU}}
	$
	be generated by a DDU representation
	$
	\theta:=(u,b,\sigma,\sigma_0,\\ h,\Sigma_Y),
	$
	and write
	$
	\mathcal T(\theta)=\rho
	$
	for its stochastic-choice array on the reachable domain. Suppose, toward a contradiction, that
	$
	\rho\in\mathfrak R_{\mathrm{DRU}}.
	$
	Then there exists a DRU representation
	$
	\widetilde\theta
	:=
	(\widetilde u,\widetilde b,\widetilde\sigma,
	\widetilde\sigma_0,\widetilde h,\widetilde\Sigma_Y)
	$
	such that
	$
	\mathcal T(\widetilde\theta)
	=
	\mathcal T(\theta)
	=
	\rho.
	$
	Since,
	$
	\widetilde\theta
	$
	is DRU,
	$
	(\widetilde u,\widetilde b,\widetilde\sigma,
	\widetilde\sigma_0,\widetilde h,\widetilde\Sigma_Y)
	$
	is independent of the conditional-law argument. Proposition~7 therefore gives, for every admissible
	$
	(s,A,z,\nu,\mu,\mu'),
	$
	$
	\widetilde C_s(z;A\mid\nu,\mu)
	=
	\widetilde C_s(z;A\mid\nu,\mu'),
	$
	and, for every admissible
	$
	(s,r,A,z,\lambda,\bar\mu,m,m'),
	$
	$
	\widetilde D_{s,r}^{\,\bar\mu}
	(z;A\mid\lambda,m)
	=
	\widetilde D_{s,r}^{\,\bar\mu}
	(z;A\mid\lambda,m').
	$
	Because
	$
	\mathcal T(\theta)=\mathcal T(\widetilde\theta),
	$
	the corresponding observable kernels satisfy
	$
	C_s^\theta
	=
	\widetilde C_s
	$
	and
	$
	D_{s,r}^{\theta,\bar\mu}
	=
	\widetilde D_{s,r}^{\,\bar\mu}
	$
	on the observational domain. Hence
	$
	C_s^\theta(z;A\mid\nu,\mu)
	=
	C_s^\theta(z;A\mid\nu,\mu')
	$
	and
	$
	D_{s,r}^{\theta,\bar\mu}
	(z;A\mid\lambda,m)
	=
	D_{s,r}^{\theta,\bar\mu}
	(z;A\mid\lambda,m')
	$
	for all admissible arguments. Thus
	$
	\theta
	$
	satisfies behavioral distributional invariance. Equivalently, Assumption~16 and Theorem~18 imply
	$
	\nu(\Gamma_s(\mu,\mu'))=0
	$
	for every admissible
	$
	(s,\nu,\mu,\mu'),
	$
	and
	$
	P_{s,r}^{m,X}\lambda
	=
	P_{s,r}^{m',X}\lambda
	$
	for every admissible
	$
	(s,r,\lambda,m,m').
	$
	Therefore neither
	$
	\mu\mapsto
	(\chi_s(z;A,\cdot,\mu))_{A,z}
	$
	nor
	$
	m\mapsto P_{s,r}^{m,X}\lambda
	$
	has a behaviorally nonconstant image on the reachable domain. In the notation of behavioral feedback,
	$
	\neg\mathcal F_s
	\wedge
	\neg\mathcal P_{s,r}
	$
	holds for every admissible
	$
	s<r,
	$
	contradicting the hypothesis that
	$
	\rho
	$
	exhibits behavioral distributional feedback, i.e.,
	$
	\mathcal F_s\vee\mathcal P_{s,r}
	$
	for some admissible
	$
	(s,r).
	$
	Hence, the assumption
	$
	\rho\in\mathfrak R_{\mathrm{DRU}}
	$
	is false, and
	$
	\rho\in
	\mathfrak R_{\mathrm{DDU}}
	\setminus
	\mathfrak R_{\mathrm{DRU}}.
	$
\end{proof}

\section{Conditional McKean-Vlasov Preferences}

\subsection{Conditional Markov Structure}

\begin{proof}[Proof of Proposition 26]
	Fix
	$
	0\le s\le r\le q\le t.
	$
	For every
	$
	(y,\eta)
	\in
	\mathbb R^{d_Y}
	\times
	\mathcal P_2(\mathbb R^d),
	$
	let
	$
	\mathbf P_{r,(y,\eta)}
	$
	denote the unique continuation law associated with the conditional
	McKean--Vlasov system of Theorem~3. Define
	$
	\mathcal G_{r,q}^F(y,\eta)
	:=
	\int
	F(\xi)\,
	\mathbf P_{r,(y,\eta)}(d\xi),
	$
	for every bounded Borel functional
	$
	F:
	C([r,q];
	\mathbb R^d
	\times
	\mathbb R^{d_Y}
	\times
	\mathcal P_2(\mathbb R^d))
	\rightarrow
	\mathbb R.
	$
	Since
	$
	(y,\eta)
	\mapsto
	\mathbf P_{r,(y,\eta)}
	$
	is a measurable stochastic kernel,
	$
	\mathcal G_{r,q}^F
	$
	is bounded and Borel measurable.
	By Lemma 25,
	$
	\mathbb E
	\!\left[
	F((\mathbf X_v)_{v\in[r,q]})
	\middle|
	\mathcal F_r^Y
	\right]
	=
	\mathcal G_{r,q}^F
	(\widehat{\mathbf X}_r),
	$
	$\mathbb P$-a.s. Hence, for every bounded Borel
	$
	H:
	\mathbb R^{d_Y}
	\times
	\mathcal P_2(\mathbb R^d)
	\rightarrow
	\mathbb R,
	$
	the tower property gives
	$
	\mathbb E
	\!\left[
	H(\widehat{\mathbf X}_r)
	F((\mathbf X_v)_{v\in[r,q]})
	\middle|
	\mathcal F_s^Y
	\right]
	=
	\mathbb E
	\!\left[
	H(\widehat{\mathbf X}_r)
	\mathcal G_{r,q}^F(\widehat{\mathbf X}_r)
	\middle|
	\mathcal F_s^Y
	\right].
	$
	Applying Lemma 25 on the interval
	$
	[s,r]
	$
	yields
	$
	\mathbb E
	\!\left[
	H(\widehat{\mathbf X}_r)
	\mathcal G_{r,q}^F(\widehat{\mathbf X}_r)
	\middle|
	\mathcal F_s^Y
	\right]
	=
	\mathcal G_{s,r}^{\,H\mathcal G_{r,q}^F}
	(\widehat{\mathbf X}_s).
	$
	Taking
	$
	H\equiv1
	$
	yields
	$
	\mathbb E
	\big[
	F((\mathbf X_v)_{v\in[r,q]})
	|
	\mathcal F_s^Y
	\big]
	=
	\mathcal G_{s,r}
	\mathcal G_{r,q}^F
	(\widehat{\mathbf X}_s).
	$
	On the other hand,
	$
	F((\mathbf X_v)_{v\in[r,q]})
	$
	is a bounded Borel functional of the continuation path on
	$
	[s,q],
	$
	so Lemma 25 also yields
	$
	\mathbb E
	\!\left[
	F((\mathbf X_v)_{v\in[r,q]})
	\middle|
	\mathcal F_s^Y
	\right]
	=
	\mathcal G_{s,q}^F
	(\widehat{\mathbf X}_s).
	$
	Therefore,
	$
	\mathcal G_{s,q}^F
	=
	\mathcal G_{s,r}
	\mathcal G_{r,q}^F
	$
	for every bounded Borel
	$
	F,
	$
	hence
	$
	\mathcal G_{s,q}
	=
	\mathcal G_{s,r}
	\mathcal G_{r,q}.
	$
	Finally,
	$
	\mathbb E
	[
	F((\mathbf X_v)_{v\in[r,q]})
	\mid
	\mathcal F_r^Y
	]
	=
	\mathcal G_{r,q}^F
	(\widehat{\mathbf X}_r)
	$
	depends on the past only through
	$
	\widehat{\mathbf X}_r.
	$
	Equivalently,
	$
	\mathcal L
	\!\left(
	(\mathbf X_v)_{v\in[r,q]}
	\mid
	\mathcal F_r^Y
	\right)
	=
	\mathcal L
	\!\left(
	(\mathbf X_v)_{v\in[r,q]}
	\mid
	\widehat{\mathbf X}_r
	\right),
	$
	so
	$
	(\widehat{\mathbf X}_v)_{0\le v\le t}
	$
	is a time-inhomogeneous Markov process with transition family
	$
	(\mathcal G_{r,q})_{0\le r\le q\le t}.
	$
\end{proof}

\subsection{Conditional McKean-Vlasov representation}

\begin{proof}[Proof of Theorem 27]
	Fix
	$
	0\le s\le r\le t.
	$
	By Proposition 26, for every bounded Borel functional
	$
	F:
	C([r,t];
	\mathbb R^d
	\times
	\mathbb R^{d_Y}
	\times
	\mathcal P_2(\mathbb R^d))
	\rightarrow
	\mathbb R,
	$
	there exists a measurable operator
	$
	\mathcal G_{r,t}^F
	$
	satisfying
	$
	\mathbb E
	\!\left[
	F((\mathbf X_v)_{v\in[r,t]})
	\middle|
	\mathcal F_r^Y
	\right]
	=
	\mathcal G_{r,t}^F
	(\widehat{\mathbf X}_r),
	$
	together with the Chapman-Kolmogorov identity
	$
	\mathcal G_{s,t}
	=
	\mathcal G_{s,r}
	\mathcal G_{r,t}.
	$
	For
	$
	(y,\eta)
	\in
	\mathbb R^{d_Y}
	\times
	\mathcal P_2(\mathbb R^d),
	$
	define the measurable kernel
	$
	\Gamma_{s,r}(y,\eta)
	:=
	\mathcal L
	(\widehat{\mathbf X}_r
	\mid
	\widehat{\mathbf X}_s=(y,\eta)).
	$
	Since
	$
	\widehat{\mathbf X}_r
	=
	(Y_r,\mu_r),
	$
	$
	\Gamma_{s,r}
	$
	is a stochastic transition operator on
	$
	\mathbb R^{d_Y}
	\times
	\mathcal P_2(\mathbb R^d).
	$
	Measurability follows from the existence of jointly measurable regular conditional distributions established in Theorem~3. Let
	$
	\lambda
	$
	be an admissible initial law.
	Then
	$
	\mathcal L(\widehat{\mathbf X}_r)
	=
	\Gamma_{s,r\#}
	\mathcal L(\widehat{\mathbf X}_s),
	$
	where
	$
	\Gamma_{s,r\#}
	$
	denotes the pushforward induced by
	$
	\Gamma_{s,r}.
	$
	Equivalently,
	$
	\widehat{\mathbf X}_r
	=
	\Gamma_{s,r}
	(\widehat{\mathbf X}_s)
	$
	in distribution. The latent-state transition kernel is therefore obtained by disintegration,
	$
	P_{s,r}^{X}(x,\cdot)
	=
	\int
	P_{r}^{X}
	(x,\cdot\mid y,\eta)
	\,
	\Gamma_{s,r}
	(dy,d\eta),
	$
	where
	$
	P_r^{X}
	(\cdot\mid y,\eta)
	$
	is the regular conditional law of
	$
	X_r
	$
	given
	$
	\widehat{\mathbf X}_r=(y,\eta).
	$
	Hence
	$
	P_{s,r}^{X}
	=
	P^{X}(\Gamma_{s,r}),
	$
	so the behavioral transition family satisfies
	$
	\Phi_P
	=
	(P_{s,r}^{X})_{0\le s\le r\le t}
	=
	(\Gamma_{s,r})_{0\le s\le r\le t}.
	$
	For
	$
	0\le s\le r\le q\le t,
	$
	Proposition 26 yields
	$
	\mathcal G_{s,q}
	=
	\mathcal G_{s,r}
	\mathcal G_{r,q}.
	$
	Since
	$
	\Gamma_{s,r}
	$
	is the stochastic kernel representing
	$
	\mathcal G_{s,r},
	$
	uniqueness of regular conditional distributions implies
	$
	\Gamma_{s,q}
	=
	\Gamma_{r,q}
	\circ
	\Gamma_{s,r}.
	$
	Taking
	$
	r=s
	$
	gives
	$
	\Gamma_{s,s}
	=
	I.
	$
	Suppose
	$
	\widetilde\Gamma
	=
	(\widetilde\Gamma_{s,r})
	$
	is another measurable family satisfying the stated properties.
	Then, for every bounded Borel
	$
	f:
	\mathbb R^{d_Y}
	\times
	\mathcal P_2(\mathbb R^d)
	\rightarrow
	\mathbb R,
	$
	$
	\int
	f
	\,d\Gamma_{s,r}(y,\eta)
	=
	\mathbb E
	[
	f(\widehat{\mathbf X}_r)
	\mid
	\widehat{\mathbf X}_s=(y,\eta)
	]
	=
	\int
	f
	\,d\widetilde\Gamma_{s,r}(y,\eta).
	$
	Since bounded Borel functions separate probability measures on Polish spaces,
	$
	\Gamma_{s,r}
	=
	\widetilde\Gamma_{s,r}
	$
	for every
	$
	(s,r).
	$
	Finally,
	$
	\mu_r
	=
	\mathcal L(X_r\mid\mathcal F_r^Y)
	$
	and
	$
	\widehat{\mathbf X}_r
	=
	(Y_r,\mu_r)
	$
	determine the conditional law of every continuation experiment by Proposition 26, whereas
	$
	\Gamma
	$
	determines the evolution of
	$
	\widehat{\mathbf X}.
	$
	Conversely,
	$
	(X,Y,\mu)
	$
	induces
	$
	\Gamma
	$
	through its regular conditional transition kernels.
	Hence the two representations generate identical finite-dimensional distributions and therefore coincide up to indistinguishability.
\end{proof}

\subsection{Endogenous Information and Fixed-Point}

\begin{proof}[Proof of Lemma 28]
	Fix
	$
	\mu\in\mathfrak M.
	$
	By Definition 12,
	$
	\mu\in\operatorname{Fix}(\Gamma)
	$
	if and only if
	$
	\Gamma(\mu)=\mu.
	$
	Since,
	$
	\Gamma=(\Gamma_{s,r})_{0\le s\le r\le t}
	$
	acts pathwise on
	$
	\mathfrak M
	=
	C([0,t];\mathcal P_2(\mathbb R^d)),
	$
	the latter identity is equivalent to
	$
	\Gamma_{s,r}(\mu_s)=\mu_r
	$
	for every
	$
	0\le s\le r\le t.
	$
	Hence
	$
	\mu\in\operatorname{Fix}(\Gamma)
	$
	if and only if
	$
	\mu_r=\Gamma_{s,r}(\mu_s)
	$
	for all
	$
	(s,r).
	$
	Assume
	$
	\mu_r=\Gamma_{s,r}(\mu_s)
	$
	for each
	$
	0\le s\le r\le t.
	$
	By Theorem 27,
	$
	\Gamma_{s,r}
	$
	uniquely determines the conditional evolution of the augmented state
	$
	\widehat{\mathbf X}
	=
	(Y,\mu),
	$
	and therefore uniquely determines the conditional transition law of the latent state. Consequently,
	$
	P_{s,r}^{X}
	=
	P^{X}(\Gamma_{s,r})
	$
	for every admissible initial law
	$
	\lambda.
	$
	Since
	$
	\Phi_P
	$
	is, by definition, the family of latent-state transition operators,
	$
	\Phi_P
	=
	(P_{s,r}^{X})_{0\le s\le r\le t}
	=
	(\Gamma_{s,r})_{0\le s\le r\le t}.
	$
	Conversely, suppose
	$
	P_{s,r}^{X}
	=
	P^{X}(\Gamma_{s,r})
	$
	and
	$
	\Phi_P
	=
	(\Gamma_{s,r})_{0\le s\le r\le t}.
	$
	Theorem 27 identifies
	$
	\Gamma
	$
	as the unique measurable transition family generating the conditional preference dynamics. Therefore the induced conditional-law flow satisfies
	$
	\mu_r
	=
	\Gamma_{s,r}(\mu_s)
	$
	for every
	$
	0\le s\le r\le t,
	$
	which is equivalent to
	$
	\Gamma(\mu)=\mu.
	$
	Hence
	$
	\mu\in\operatorname{Fix}(\Gamma).
	$
	The three constructions determine the same admissible conditional preference flows and are equivalent.
\end{proof}	

\begin{proof}[Proof of Proposition 29]
	Fix the primitive tuple
	$
	\vartheta:=(u,b,\sigma,\sigma_0,h,\Sigma_Y)
	$
	and an admissible initial law
	$
	\lambda_0.
	$
	Let
	$
	\mu,\tilde\mu\in\operatorname{Fix}(\Gamma).
	$
	Then
	$
	\Gamma(\mu)=\mu
	$
	and
	$
	\Gamma(\tilde\mu)=\tilde\mu,
	$
	so the corresponding systems
	$
	(X^\mu,Y^\mu,\mu)
	$
	and
	$
	(X^{\tilde\mu},Y^{\tilde\mu},\tilde\mu)
	$
	are weak solutions of system~(3) with common primitives
	$
	\vartheta
	$
	and common initial law
	$
	\lambda_0.
	$
	By the uniqueness of Theorem 3,
	$
	\mathcal L(X^\mu,Y^\mu,\mu)
	=
	\mathcal L(X^{\tilde\mu},Y^{\tilde\mu},\tilde\mu)
	$
	on
	$
	C([0,t];\mathbb R^d)
	\times
	C([0,t];\mathbb R^{d_Y})
	\times
	C([0,t];\mathcal P_2(\mathbb R^d)).
	$
	Canonical compatibility of the conditional-law versions therefore gives
	$
	\mathcal L(Y^\mu,\mu)
	=
	\mathcal L(Y^{\tilde\mu},\tilde\mu)
	$
	and, for every
	$
	0\le s\le r\le t,
	$
	$
	\Gamma_{s,r}^{\mu}(\hat x,\cdot)
	=
	\Gamma_{s,r}^{\tilde\mu}(\hat x,\cdot)
	$
	for
	$
	\mathcal L(\widehat{\mathbf X}_s)
	$
	-a.e.
	$
	\hat x\in
	\mathbb R^{d_Y}\times\mathcal P_2(\mathbb R^d).
	$
	Thus
	$
	\Gamma^\mu=\Gamma^{\tilde\mu}
	$
	on the reachable augmented-state domain. Theorem 27 yields
	$
	P_{s,r}^{\mu,X}
	=
	P^X(\Gamma_{s,r}^{\mu})
	$
	and
	$
	P_{s,r}^{\tilde\mu,X}
	=
	P^X(\Gamma_{s,r}^{\tilde\mu}).
	$
	Hence
	$
	P_{s,r}^{\mu,X}
	=
	P_{s,r}^{\tilde\mu,X}
	$
	for every admissible
	$
	(s,r,\lambda),
	$
	and therefore
	$
	\Phi_P^\mu
	=
	\Phi_P^{\tilde\mu}.
	$
	Since, the primitive felicity index
	$
	u
	$
	is common to both systems,
	$
	\Phi_u^\mu=\Phi_u^{\tilde\mu}
	$
	on the common reachable domain. Consequently,
	$
	\Phi^\mu
	=
	(\Phi_u^\mu,\Phi_P^\mu)
	=
	(\Phi_u^{\tilde\mu},\Phi_P^{\tilde\mu})
	=
	\Phi^{\tilde\mu}.
	$
	Using
	$
	\mathscr T=\Lambda\circ\Phi,
	$
	we obtain
	$
	\mathscr T^\mu
	=
	\Lambda(\Phi^\mu)
	=
	\Lambda(\Phi^{\tilde\mu})
	=
	\mathscr T^{\tilde\mu}.
	$
	Therefore, every
	$
	\mu,\tilde\mu\in\operatorname{Fix}(\Gamma)
	$
	belong to the same behavioral equivalence class, or equivalently,
	$
	\Phi(\operatorname{Fix}(\Gamma))
	$
	and
	$
	\mathscr T(\operatorname{Fix}(\Gamma))
	$
	are singletons.
\end{proof}

\begin{proof}[Proof of Theorem 30]
	Fix
	$
	m\in\mathfrak M
	$
	and let
	$
	(X^m,Y^m)
	$
	solve the frozen system obtained from system~(3) by replacing the conditional-law argument with
	$
	m.
	$
	By definition,
	$
	\Gamma(m)_s
	=
	\mathcal L(X_s^m\mid\mathcal F_s^{Y^m}),
	$
	$
	0\leq s\leq t.
	$
	Hence,
	$
	m\in\operatorname{Fix}(\Gamma)
	$
	is equivalent to
	$
	m_s
	=
	\mathcal L(X_s^m\mid\mathcal F_s^{Y^m})
	,\ \forall\
	s\in[0,t].
	$
	Substitution of
	$
	m=\Gamma(m)
	$
	into the frozen equations yields
	$
	(b,\sigma,\sigma_0,h,\Sigma_Y)
	(s,X_s^m,Y_s^m,m_s)
	=
	(b,\sigma,\sigma_0,h,\Sigma_Y)
	(s,X_s^m,Y_s^m,
	\mathcal L(X_s^m\mid\mathcal F_s^{Y^m})),
	$
	so
	$
	(X^m,Y^m,m)
	$
	satisfies system~(3) and is an admissible DDU solution. Conversely, let
	$
	(X,Y,m)
	$
	be an admissible DDU solution generated by the frozen-flow coefficients indexed by
	$
	m.
	$
	The DDU consistency restriction gives
	$
	m_s
	=
	\mathcal L(X_s\mid\mathcal F_s^Y)
	=
	\Gamma(m)_s
	,\ \forall\ 
	s\in[0,t],
	$
	and therefore
	$
	m=\Gamma(m).
	$
	Thus
	$
	m\in\operatorname{Fix}(\Gamma)
	$
	if and only if the frozen system closes to an admissible DDU solution. Fix
	$
	m\in\operatorname{Fix}(\Gamma).
	$
	By Theorem 27, the augmented state
	$
	\widehat{\mathbf X}^{m}
	=
	(Y^m,m)
	$
	admits the conditional transition family
	$
	(\Gamma_{s,r}^{m})_{0\leq s\leq r\leq t},
	$
	and the latent-state marginal transition operator satisfies
	$
	P_{s,r}^{m,X}
	=
	P^X(\Gamma_{s,r}^{m}).
	$
	Consequently,
	$
	\Phi_P(m)
	=
	(P_{s,r}^{m,X})_{0\leq s\leq r\leq t}
	=
	\bigl(P^X(\Gamma_{s,r}^{m})\bigr)_{0\leq s\leq r\leq t}.
	$
	For every admissible
	$
	(s,A,z,\nu,m_s),
	$
	the contemporaneous component is
	$
	C_s^m(z;A\mid\nu,m_s)
	=
	\int_{\mathbb R^d}
	\chi_s^m(z;A,x,m_s)\,\nu(dx),
	$
	while, for every admissible
	$
	(s,r,A,z,\lambda,\bar\mu,m),
	$
	the continuation component is
	$
	D_{s,r}^{m,\bar\mu}(z;A\mid\lambda)
	=
	\int_{\mathbb R^d}
	\chi_r^m(z;A,x,\bar\mu)\,\\
	P^X(\Gamma_{s,r}^{m})\lambda(dx).
	$
	Therefor,e the full observable array factors as
	$
	\mathscr T(m)
	=
	\Lambda(\Phi_u(m),\Phi_P(m)).
	$
	Let
	$
	m,\tilde m\in\operatorname{Fix}(\Gamma)
	$
	be generated by the same primitive tuple
	$
	\vartheta
	=
	(u,b,\sigma,\sigma_0,h,\Sigma_Y)
	$
	and the same initial law
	$
	\lambda_0.
	$
	Then
	$
	(X^m,Y^m,m)
	$
	and
	$
	(X^{\tilde m},Y^{\tilde m},\tilde m)
	$
	are weak solutions of system~(3) with common
	$
	(\vartheta,\lambda_0).
	$
	The uniqueness conclusion of Theorem 3 yields
	$
	\mathcal L(X^m,Y^m,m)
	=
	\mathcal L(X^{\tilde m},Y^{\tilde m},\tilde m).
	$
	Canonical compatibility of the conditional-law versions and Theorem 27 imply
	$
	\Gamma_{s,r}^{m}
	=
	\Gamma_{s,r}^{\tilde m}
	$
	on the reachable augmented-state domain and, consequently,
	$
	P_{s,r}^{m,X}
	=
	P_{s,r}^{\tilde m,X}
	$
	for every
	$
	0\leq s\leq r\leq t.
	$
	Since, the felicity primitive
	$
	u
	$
	is common,
	$
	\Phi_u(m)=\Phi_u(\tilde m),
	$
	while equality of the latent transition families gives
	$
	\Phi_P(m)=\Phi_P(\tilde m).
	$
	Hence
	$
	\mathcal B(m)
	=
	(\Phi_u(m),\Phi_P(m))
	=
	(\Phi_u(\tilde m),\Phi_P(\tilde m))
	=
	\mathcal B(\tilde m).
	$
	Thus
	$
	\mathcal B(\operatorname{Fix}(\Gamma))
	$
	is a singleton. Finally,
	$
	\mathscr T(m)
	=
	\Lambda(\mathcal B(m))
	$
	for every
	$
	m\in\operatorname{Fix}(\Gamma).
	$
	Therefore all fixed points generated by fixed primitives and initial law belong to the same class
	$
	[m]_{\mathcal B},
	$
	and this class determines a unique admissible stochastic-choice array.
\end{proof}

\bibliographystyle{apalike}
\bibliography{bib}
\end{document}